\documentclass[a4paper,11pt]{article}
\usepackage[english]{babel}
\usepackage[margin=35mm]{geometry}
\usepackage{amsmath}
\usepackage{amsfonts}
\usepackage{amssymb}
\usepackage{mathrsfs}
\usepackage{mathtools}
\usepackage{amsthm}
\usepackage{graphicx}
\usepackage{natbib}
\usepackage{comment} 
\usepackage[affil-it]{authblk} 
\usepackage{hyperref}
\usepackage{xcolor}
\hypersetup{
  colorlinks   = true, 
  urlcolor     = blue, 
  linkcolor    = red, 
  citecolor   = blue 
}
\usepackage{color}

\providecommand{\keywords}[1]
{
 {\small
  \textbf{\textit{Keywords:}}} #1
}

\newtheorem{theorem}{Theorem}[section]
\newtheorem{lemma}[theorem]{Lemma}

\newtheorem{corollary}[theorem]{Corollary}

\theoremstyle{definition}

\newtheorem{remark}[theorem]{Remark}
{\theoremstyle{plain}\newtheorem{assumption}{Assumption}}

\numberwithin{equation}{section}

\newcommand{\Bfrak}{\mathfrak{B}}
\newcommand{\bfrak}{\mathfrak{b}}
\newcommand{\cov}{{\rm Cov}}
\newcommand{\dd}{{\rm d}}
\newcommand{\E}{{\rm E}}
\newcommand{\eps}{\varepsilon}
\newcommand{\Ealg}{\mathcal{E}}
\newcommand{\Falg}{\mathcal{F}}

\newcommand{\Galg}{\mathcal{G}}

\newcommand{\Zalg}{\mathcal{Z}}
\newcommand{\normal}{{\rm N}}
\newcommand{\pr}{{\rm Pr}}
\newcommand{\real}{\mathbb{R}}
\newcommand{\sigalg}{$\sigma$-algebra}
\newcommand{\tr}{{\rm t}}

\newcommand{\Zscr}{\mathscr{Z}}

\newcommand{\Gscr}{\mathscr{G}}
\newcommand{\Fscr}{\mathscr{F}}

\newcommand{\var}{{\rm Var}}

\newcommand{\Xalg}{\mathcal{X}}
\newcommand{\rdd}{{regression discontinuity design}}

\newcommand{\ate}{{\text{\sc ate}}}
\newcommand{\cate}{{\text{\sc cate}}}
\newcommand{\late}{{\text{\sc late}}}
\newcommand{\indep}{\perp \!\!\! \perp}
\newcommand{\abs}[1]{\lvert#1\lvert} 
\newcommand{\norm}[1]{\lVert#1\rVert} 

\title{Regression discontinuity design with\\ right-censored survival data}
\author{Emil Aas Stoltenberg\\\
	\texttt{\small emilstoltenberg@gmail.com} } 
\affil{\normalsize Department of Data Science\\ BI Norwegian Business School, Oslo}	

\date{\today} 

\begin{document}
\maketitle

\begin{abstract} {\small In this paper the regression discontinuity design is adapted to the survival analysis setting with right-censored data, studied in an intensity based counting process framework. In particular, a local polynomial regression version of the Aalen additive hazards estimator is introduced as an estimator of the difference between two covariate dependent cumulative hazard rate functions. Large-sample theory for this estimator is developed, including confidence intervals that take into account the uncertainty associated with bias correction. As is standard in the causality literature, the models and the theory are embedded in the potential outcomes framework. Two general results concerning potential outcomes and the multiplicative hazards model for survival data are presented.}
\end{abstract} \hspace{10pt}

\keywords{Asymptotic bias; bias correction; causality; censoring; confounding; counting processes; hazard rate functions; intensity processes; martingales; multiplicative hazard model; treatment effects.}

\section{Introduction}\label{sec::intro} The {\rdd} is a widely used technique for causal inference in economics, political science, and sociology (see~\citet{van2008regression}, \citet{imbens2008regression}, \citet{lee2010regression}, and \citet{cattaneo2022regression} for recent reviews). In the {\rdd}, possible confounders are {`}controlled for{'} by exploiting that units are assigned to treatment based on whether their value of an observed covariate is above or below some known cut-off, the idea being that units with values of this observed covariate just above the known cut-off are similar to the subjects with values just below the cut-off.

For situations where the observed outcomes are noncensored random variables, the inference theory, building on that of local polynomial regression, is well developed (see e.g., \citet{fan1996local}, \citet{hahn2001identification}, and \citet{calonico2014robust,calonico2014supplement}, and the reviews above). In this paper we extend the {\rdd} to the survival analysis setting. In particular, we study the {\rdd} applied to right-censored survival data in an intensity based counting process framework (see \citet{ABGK93}, and \citet{ABG08,aalen2010history}). 
The survival analysis models we study are all instances of the broad class of multiplicative intensity models, meaning that the modelling and inference revolves around the hazard rate function. 

The setup is as follows: On a probability space $(\Omega,\mathscr{H},\pr)$, let $Z \in \real$ be an observed covariate; $X\in \{0,1\}$ a treatment indicator; $U$ a vector of unobserved confounders; $\widetilde{T}^0 \geq 0$ the potential outcome for the non-treated, and $\widetilde{T}^1 \geq 0$ the potential outcome under treatment (for potential outcomes theory, see, for example,~\citet{holland1986statistics}, \citet{morgan2015counterfactuals}, \citet{imbens2015causal}, and \citet{imbens2020potential}\footnote{\citet{imbens2020potential} is an excellent and very readable article comparing the potential outcomes theory to the directed acyclic graph theory associated with Judea Pearl. See~\citep{pearl2009causality} for the canonical treatise on DAGs, and \citet{pearl2018why} or~\citet{pearl2016causal} for more accessible accounts.}). Since $\widetilde{T}^0$ and $\widetilde{T}^1$ take values on the positive half of the real line, we refer to these two random variables as {\it potential lifetimes}. That we are in a regression discontinuity setting, means that treatment is determined by the value of an observed covariate. We take $X = I\{Z \geq z_0\}$ for some known cut-off $z_0$, where $I\{A\}$ is the indicator function of the event $A$.\footnote{What we describe here is the {\it sharp} regression discontinuity design. In contrast, a {\it fuzzy} regression discontinuity design occurs when the probability of receiving treatment does not jump from zero to one at the cut-off, rather $z \mapsto \pr(X = 1\mid Z = z)$ has a point of discontinuity at the cut-off. This paper deals only with the sharp regression discontinuity design.} The two potential lifetimes $\widetilde{T}^0$ and $\widetilde{T}^1$ are assumed to stem from distributions with hazard rate functions $\alpha_0(t,Z,U)$ and  $\alpha_1(t,Z,U)$, respectively, meaning that 
\begin{equation}
\lim_{\eps \to 0}\frac{1}{\eps}\,\pr(\widetilde{T}^g \in [t,t + \eps) \mid \widetilde{T}^g  \geq t, W = w , U = u) 
= \alpha_g(t,z,u), \quad\text{for $g = 0,1$}.  
\notag
\end{equation}
In analogy with the average treatment effect (the {\ate}) and conditional or local average treatment effects ({\cate}s or {\late}s) (see, e.g., \citet{imbens2015causal}), the causal estimands of the present paper are all defined in terms $\alpha_0(t,Z,U)$ and  $\alpha_1(t,Z,U)$ when the confounder $U$ is averaged out. Specifically, with 
\begin{equation}
\theta(t,z) = \E\,\{\alpha_1(t,Z,U) - \alpha_0(t,Z,U) \mid Z = z\},
\label{eq::theta_z0}
\end{equation}
the main estimand in this paper is
\begin{equation}
\Theta(t,z) = \int_0^t \theta(s,z) \,\dd s,
\notag
\end{equation}
evaluated in the cut-off $z = z_0$. The estimands $\theta(t,z_0)$ and $\Theta(t,z_0)$ are, we contend, natural survival analysis counterparts of the classical estimand in the standard {\rdd}, namely the average treatment effect at the cut-off, see, for example, Eq.~(2.1) in \citet[p.~617]{imbens2008regression}. 

The paper proceed as follows. In Section~\ref{sec::model_and_estimands} we provide a general presentation of the potential outcomes framework as it applies to right-censored survival data studied in a counting process framework. In Section~\ref{sec::data_and_assumptions}, we narrow in on the {\rdd} and present a lemma that is key to making the {\rdd} feasible in the present setting. Section~\ref{subsec::estimands} contains a discussion of various causal estimands, and assumptions under which they may be identified. In Section~\ref{sec::rdd_in_surv} we lay out the assumptions made about the model generating the data, and present a special case of the covariate-localised Aalen additive hazards estimator. Section~\ref{subsec::large_sample} contains large-sample theory for the general version of the estimator, with results on variance estimation and bias correction in Sections~\ref{subsec::variance_estimation} and \ref{sec::bias_correction}, respectively. 

\subsection{Potential outcomes and right-censored data}\label{sec::model_and_estimands}
Let $(\widetilde{T}^0,\widetilde{T}^1,X,Z,U,C)$ be random variables on the probability space $(\Omega,\mathscr{H},\pr)$. Here, $\widetilde{T}^0$ and $\widetilde{T}^1$ are the potential lifetimes corresponding to whether a subject is not treated or treated, respectively; $X \in \{0,1\}$ is an indicator of treatment; $Z$ is a vector of observed covariates; $U$ is a vector of unobserved possible confounders;\footnote{A confounder is a covariate that is associated with the outcome {\it and} with the treatment variable of interest. In the multiplicative intensity models studied here, this means that a covariate, to be a confounder, must have an effect on the hazard rate {\it and} be correlated with the treatment variable. If a covariate only satisfy one of these two conditions, then it is not a confounder.} and $C \geq 0$ is a censoring variable. We use $g$ as an index for untreated ($g = 0$) and treated ($g = 1$), and write, for example, $\widetilde{T}^g$ to indicate one of the two potential lifetimes. The units under study are observed over the time interval $[0,\tau]$, where $\tau < \infty$. Assume that, conditionally on $(Z,U)$, the potential lifetimes stem from distributions with hazard rate functions $\alpha_{0}(t,Z,U)$ and $\alpha_{1}(t,Z,U)$.\footnote{Here we model the unconfoundedness assumption $(\widetilde{T}^0,\widetilde{T}^1)\indep X \mid (Z,U)$ directly. See~\citet[Eq.~(2)]{yadlowsky2018bounds}. In the {\rdd} this assumption is trivially satisfied \citet[p.~618]{imbens2008regression}.}

What distinguishes the potential outcomes theory when right-censoring is present from when there is no censoring, is that with no censoring only one of the two potential lifetimes is observed for the same unit, while for right-censored data {\it at most} one of the potential lifetimes is actually observed for the same unit. This means that when there is no censoring we observe 
\begin{equation}
\widetilde{T} = X\widetilde{T}^1 + (1 - X)\widetilde{T}^0, 
\notag
\end{equation}
while, when the data are right-censored we only observe $\widetilde{T}$ if it is smaller than the censoring time $C$, that is
\begin{equation}
T = \min( \widetilde{T}, C) =  XT^1 + (1 - X)T^0, 
\notag
\end{equation}
where $T^0 = \min(\widetilde{T}^0,C)$ and $T^1 = \min(\widetilde{T}^1,C)$ are the possibly right-censored potential lifetimes. Throughout the paper, we assume that the two potential lifetimes are conditionally independent given the covariates, which we express by
\begin{equation}
\widetilde{T}^0\indep\widetilde{T}^1\mid (X,Z,U).
\label{eq::indep_lifetimes}
\end{equation} 
and we work under the assumption of {\it random right-censoring} (see \citet[Example~III.25, p.~143]{ABGK93}), which means that $(\widetilde{T}^0,\widetilde{T}^1)$ is independent of the censoring time $C$  given $(X,Z,U)$, with symbols
\begin{equation}
(\widetilde{T}^0,\widetilde{T}^1) \indep C \mid (X,Z,U).
\label{eq::indep_censoring}
\end{equation}
Two potential lifetimes in turn leads to two indicators of noncensoring, 
\begin{equation}
\delta^0 = I\{\widetilde{T}^0 \leq C\},\quad\text{and} \quad\delta^1 = I\{\widetilde{T}^1 \leq C\}, 
\notag
\end{equation}
and to two potential counting and potential at-risk processes 
\begin{equation}
N^g(t) = I\{T^g \leq t,\delta^g = 1 \}\quad \text{and}\quad Y^g(t) = I\{T^g \geq t\},
\quad \text{for $g = 0,1$}.
\notag
\end{equation}
Let $\Xalg = \sigma(X,Z,U)$, define the filtrations $\mathscr{E}_{t}^g =  \sigma(\{N^g(s),Y^g(s)\}_{s\leq t})$ for $g = 0,1$, and set    
\begin{equation}
\Gscr_t^g = \mathscr{E}_{t}^g \vee \Xalg,\quad \text{for $g = 0,1$}.
\notag
\end{equation}
It is assumed that $X$, $Z$, and $U$ are realised at time zero, meaning that $\Gscr_0^g = \Xalg$. With respect to $\Gscr_t^0$ and $\Gscr_t^1$ we have, using the assumptions in~\eqref{eq::indep_lifetimes} and \eqref{eq::indep_censoring} (see \citet[Ch.~III.2.2, pp.~138--141]{ABGK93}), two {`}potential{'} local square integrable martingales $M^0$ and $M^1$, respectively, given by
\begin{equation}
M^{g}(t) = N^g(t) - \int_0^t Y^g(s)\alpha_{g}(s,Z,U)\,\dd s,\quad \text{for $g = 0,1$}.
\label{eq::potential_MGs}
\end{equation}
This ends our description of our basic modelling assumptions. The reader familiar with the counting process approach to, and martingale methods in, survival analysis will see that the above is nothing more than what one gets the when taking the potential outcomes framework seriously and applying it to the standard intensity based counting process approach to survival analysis. In fact, the presentation so far is just a slight notational and semantic reformulation of the theory presented in Chapter~III.2 of \citet{ABGK93}. It is important to notice that the modelling undertaken up to this point has taken place in the two potential worlds, so to speak, culminating in the {`}potential world{'} martingales of~\eqref{eq::potential_MGs}. We now proceed to the consequences of this model for quantities of this world, that is, the observed quantities. Denote the observed counting process 
\begin{equation}
N(t) = XN^1(t) + (1 - X)N^0(t), 
\label{eq::obs_count}
\end{equation}
and the observed at-risk process
\begin{equation}
Y(t) = XY^1(t) + (1 - X)Y^0(t).
\label{eq::obs_atrisk}
\end{equation} 
Let $\Ealg_{t} = \sigma(\{N(s),Y(s)\}_{s\leq t})$ be the filtration generated by the observable counting and at-risk processes, and let $\Xalg^{\rm obs} = \sigma(X,Z)$ be the {\sigalg} generated by the observed covariates. Define the filtrations
\begin{equation}
\Galg_{t} = \Ealg_t\vee \Xalg,\quad \text{and}
\quad \Falg_{t} = \Ealg_t\vee \Xalg^{\rm obs}, 
\label{eq::GalgFalgdef}
\end{equation}
where, as above, $\Galg_{0} = \Xalg$ and $\Falg_0 = \Xalg^{\rm obs}$. The next lemma says that the intensity process of the observed counting process $N$ with respect to $\Galg_t$ takes the form it it intuitively should take. 

\begin{lemma}\label{lemma::composite_intensity} Assume that~\eqref{eq::indep_lifetimes} and~\eqref{eq::indep_censoring} hold. The $\Galg_{t}$-intensity of the observed counting process $N = XN^1 + (1 -X)N^0$ is $Y(t) \{X \alpha_1(t,Z,U) + (1 - X) \alpha_0(t,Z,U) \}$. In particular, 
\begin{equation}
I_{X = g}M^g(t) = I_{X = g}\{ N^g(t) - \int_0^t Y^g(s)\alpha_{g}(s,Z,U)\,\dd s \}, 
\label{eq::Galg_MGs}
\end{equation}
is a locally square integrable martingale with respect to $\Galg_t$.
\end{lemma}

\begin{proof} Define $\Gscr_t = \Gscr_t^0 \vee \Gscr_t^1$. It follows directly from the theory on independent right-censoring \citep[pp.~138--140]{ABGK93} that the $\Gscr_t$-intensity of $N^g(t)$ is also $Y^g(t)\alpha_g(t,Z,U)$. Since $I_{X=g}$ is $\Gscr_{t}$ measurable for all $t$, the $\Gscr_t$-intensity of $I_{X=g}N^g(t)$ is $I_{X=g}Y^g(t)\alpha_g(t,Z,U)$ (this is a consequence of Theorem~II.3.1 in \citet[p.~71]{ABGK93}, considering $I_{X=g}$ as a constant process in $t$). Since $\Galg_{t}\subset \Gscr_t$, the innovation theorem (see \citet[Theorem~3.4, p.~706]{aalen1978nonparametric} or \citet[p.~80]{ABGK93}) entails that the $\Galg_t$-intensity process, $\lambda_g^{\Galg}(t)$ say, of $I_{X=g}N^g$ is
\begin{equation}
\lambda_g^{\Galg}(t) = \E\,\{ I_{X=g}Y^g(t) \alpha_g(t,Z,U) \mid \Galg_{t-}\}.
\notag
\end{equation}
Since $I_{X=g}Y^g(t) = I_{X=g}Y(t)$, in terms of the observed at-risk process in~\eqref{eq::obs_atrisk}; $I_{X = g}$ is $\Galg_t$-measurable for all $t$, and $Y(t)$ is predictable with respect to $\Galg_{t}$,   
\begin{equation}
\begin{split}
\lambda_g^{\Galg}(t) & = \E\,\{ I_{X=g}Y^g(t) \alpha_g(t,Z,U) \mid \Galg_{t-}\}\\
& = \E\,\{ I_{X=g}Y(t) \alpha_g(t,Z,U) \mid \Galg_{t-}\}
=  I_{X=g}Y(t) \alpha_g(t,Z,U).
\end{split}
\notag
\end{equation}
This entails that $I_{X=g} M^g(t) = I_{X = g}\{N^g(t) - \int_0^t \lambda_g^{\Galg}(s)\,\dd s\}$ is a $\Galg_t$ martingale, and, because it is a counting process martingale, it is locally square integrable (see \citet[p.~73]{ABGK93}). Since $X$ is $\Galg_0$-measurable ($\Galg_0 = \Xalg$) and does not depend on time, $XM^1 + (1 - X)M^0$ is also a locally square integrable martingale.   
\end{proof}
An important, though rather intuitive, thing to note in the preceding lemma is that it does {\it not} say that $M^g$ is a $\Galg_t$-martingale. Instead, the lemma says that $I_{X = g}M^g$ is a $\Galg_t$-martingale. This is intuitive because $\Galg_t$ is generated by the observables $N$ and $Y$ in~\eqref{eq::obs_count} and \eqref{eq::obs_atrisk}, respectively, and if $X=0$, for example, then the path of $N^1$ is certainly not observable.  
In the next section we specialise the potential outcomes model for right-censored survival data to the regression discontinuity setting.

\subsection{The regression discontinuity setting}\label{sec::data_and_assumptions}
We retain the definitions and assumptions from the previous section, with the following two exceptions: To conform with the {\rdd}, we require $Z \in \real$ and set $X = I\{Z \geq z_0\}$ for some known cut-off $z_0$. Henceforth, we often refer to $Z$ as the {\it forcing variable}. Notice also that $\sigma(X)$ is in included in $\sigma(Z)$, with the consequence that $\Xalg = \sigma(Z,U)$ and $\Xalg^{\rm obs} = \sigma(Z)$, and thus $\Galg_t = \Ealg_t \vee \sigma(Z,U)$ and $\Falg_t = \Ealg_t \vee \sigma(Z)$ (compare with~\eqref{eq::GalgFalgdef}). 

In order to make what follows clear, we now take a short detour via the standard regression discontinuity design, as presented, for example, in~\citet{imbens2008regression} or \citet[Section~2]{calonico2014robust}. Let $y^0$ and $y^1$ be two real valued potential outcomes,\footnote{We use lowercase letters to distinguish these two potential outcomes from the potential at-risk processes introduced in Section~\ref{sec::model_and_estimands}.} $Z \in \real$ the forcing variable, and $X = I\{Z \geq z_0\}$ the treatment indicator. The observed outcome is $y = Xy^1 + (1 - X)y^0$, and the estimand of interest is $\tau_{\rm srd} = \E\,( y^1 - y^0 \mid Z = z_0)$, called the average treatment effect at the cut-off (or threshold). To estimate $\tau_{\rm srd}$, the limits $\lim_{z \uparrow z_0}\E\,( y \mid Z = z)$ and $\lim_{z \downarrow z_0}\E\,( y \mid Z = z)$ are estimated using local polynomial regressions to the left and to the right of the cut-off, respectively. 
Since the roles of the hazard rate functions $\alpha_0(t,Z,U)$ and $\alpha_1(t,Z,U)$ in multiplicative intensity models are analogoues to those of the conditional expectations $\E\,(y^0 \mid Z)$ and $\E\,(y^1 \mid Z)$ in the standard {\rdd}, we would like to work with hazard rate functions that only depend on the forcing variable (and time). More to the point, if $\E\,\{\alpha_g(t,Z,U) \mid \sigma(Z)\}$ was the hazard rate function of $I_{X = g} N^g$ with respect to the filtration of observables, then the standard local polyonimial regression theory would be straighforward to mimic. This is not quite the case, but nearly, in a sense made precise by the next lemma.


\begin{lemma}\label{lemma::A4} Let $T,Z,U,N,Y$, $\Ealg_t, \Falg_t$, and $\Galg_t$ be as defined in Section~\ref{sec::data_and_assumptions}. Assume that $Z$ has density $f_Z(z)$ with support $[z_1,z_2]$, and that $f_Z(z)$ is bounded above and below on $[z_1,z_2]$, and that $z \mapsto \pr(T \geq \tau \mid Z = z)$ is bounded below on $[z_1,z_2]$. For each $t \in [0,\tau]$, define 
\begin{equation} 
\mathcal{Z}_t = \sigma(Y(t))\vee \sigma(Z),
\notag
\end{equation}
and let $\xi_t$ be a non-negative $\Galg_t$-adapted and almost surely left-continuous process, with $\E\,\sup_{t \leq \tau}\abs{\xi_t}^2\, < \infty$. Then the left-continuous modifications of the processes 
\begin{equation}
\E\,(\xi_t \mid \Falg_{t-})\quad\text{and}\quad I_{T \geq t}\E\,(\xi_t \mid \Zalg_t) + I_{T < t}\E\,(\xi_t \mid \Falg_{t-}), 
\notag
\end{equation}
are indistinguishable.
\end{lemma}
\begin{proof} The proof of this lemma is in Appendix~\ref{subsec::proofA4}. 
\end{proof}
Our application of this lemma is when $\xi_t$ is one of the hazard rate functions $\alpha_{0}(t,Z,U)$ or $\alpha_{1}(t,Z,U)$. The importance of this lemma derives from the fact that $Y(t)\E\,(\xi_t \mid \Falg_{t-})$ is, at first sight, a complicated function as it may depend on the forcing variable $Z$ as well as the paths of $N(t)$ and $Y(t)$. But an implication of Lemma~\ref{lemma::A4}, since $Y(t)$ is left-continuous, is that  
\begin{equation}
\pr\{ Y(t)\E\,(\xi_t \mid \Falg_{t-}) = Y(t) \E\,(\xi_t \mid \Zalg_t),\; \text{for all $t \in [0,\tau]$}\} =  1,
\notag
\end{equation}
and since $\Zalg_t = \{\{T \geq t\},\{T < t\},\emptyset,\Omega\}\vee \sigma(Z)$, the conditional expectation $Y(t) \E\,(\xi_t \mid \Zalg_t)$ is a function of the forcing variable only (see, for example, Lemma~1.13, p.~7 and the discussion on p.~106 in \citet{kallenberg2002foundations}).  This motivates defining the functions,\footnote{ That $\E\,(\xi_t \mid \Zalg_t) = \E\,\{I_{T \geq t}\xi_t \mid \sigma(Z)\}/\pr\{T \geq t\mid \sigma(Z)\} + \E\,\{I_{T < t}\xi_t \mid \sigma(Z)\}/\pr\{T < t\mid \sigma(Z)\}$ almost surely, can be proved following the steps in Ex.~34.4 of \citet[p.~455]{billingsley1995probability}.}    
\begin{equation}
\bar{\alpha}_{g}(t,z)  = \frac{ \E\,\{\alpha_g(t,Z,U)Y(t)\mid Z = z	\}	}{\pr\{T \geq t \mid Z = z\}	},\quad\text{for $g = 0,1$}.
\label{eq::def_alpha0}
\end{equation} 
Notice that if $\alpha_g(t,Z,U) \leq g(Z,U)$ for some $g$ such that $\E\,g(Z,U) < \infty$, then we can pass the derivative under the intergal sign in $\partial/\partial t\,\E\,\{S_g(t,Z,U) \mid \sigma(Z)\} = \E\,\{\partial/\partial tS_g(t,Z,U) \mid \sigma(Z)\}$,\footnote{The dominated convergence theorem extends to conditional expectations, see, e.g.~\citet[Lemma~2.4.3, p.~57]{cohen2015stochastic}} and consequently we have the relation $I_{X = g}\bar{\alpha}_g(t,z) =  - I_{X = g}(\partial \bar{S}_g(t,z)/\partial t)/ \bar{S}_g(t,z)$. We summarise the above in a lemma that is used repeatedly in the remainder of the paper.

\begin{lemma}\label{lemma::barMG} For $g = 0,1$ define 
\begin{equation}
\bar{M}^g(t) = I_{X = g} \{ N^g(t) - \int_0^t Y^g(s)\bar{\alpha}_g(s,Z)\,\dd s\}. 
\notag
\end{equation}
Then $\bar{M}^g(t)$ are local square integrable martingales with respect to the filtration $\Falg_t$.
\end{lemma}
\begin{proof} From Lemma~\ref{lemma::composite_intensity} we have that $I_{X = g}Y^g(t) \alpha_g(t,Z,U)$ is the $\Galg_t$-intensity of $I_{X = g}N^g(t)$. By the innovation theorem (references above), the $\Falg_t$-intensity $\lambda^{\Falg}(t)$ of $I_{X = g}N^g(t)$ is $I_{X = g}Y^g(t) \E\,\{\alpha_g(t,Z,U)\mid \Falg_{t-}\}$. Since $I_{X=g}Y^g(t) = I_{X=g}Y(t)$, Lemma~\ref{lemma::A4} gives that $\lambda^{\Falg}(t)$ and $I_{X=g}Y^g(t) \bar{\alpha}_g(t,Z)$ are indistinguishable. Thus, $I_{X=g}\int_0^tY^g(s) \bar{\alpha}_g(s,Z)\,\dd s$ is the compensator of the counting process $I_{X = g}N^g(t)$, and $\bar{M}^g(t)$ are locally square integrable martingales \citep[pp.~72--74]{ABGK93}.
\end{proof}

\subsection{Estimands}\label{subsec::estimands}
If the difference $\alpha_{1}(t,z,u) - \alpha_{0}(t,z,u)$ is functionally independent of $u$, meaning that $\alpha_{1}(t,Z,Y) - \alpha_{0}(t,Z,U)$ is a $\sigma(Z)$-measurable random variable, then 
\begin{equation}
\bar{\alpha}_1(t,z) - \bar{\alpha}_0(t,z) = \theta(t,z),
\label{eq::easy_eq}
\end{equation}
where $\bar{\alpha}_g(t,z)$ for $g = 0,1$ are defined in~\eqref{eq::def_alpha0} and $\theta(t,z)$ is defined in~\eqref{eq::theta_z0}. The functional independence assumption just introduced is crucial for whether $\theta(t,z_0)$ is identifiable or not. 
We state the functional independence assumption here for easy reference, note, however, that it is not assumed throughout. 
\begin{assumption}\label{assumption::func_indep} The difference $\alpha_1(t,Z,U) - \alpha_0(t,Z,U)$ is functionally independent of the confounder $U$.
\end{assumption}
If this assumption is dropped, we are only able to estimate the average treatment effect at the cut-off {\it among those at-risk}, a quantity we denote $\theta_{\rm risk}(t,z_0)$, it is 
\begin{equation}
\theta_{\rm risk}(t,z) = \E\,\{\alpha_{1}(t,Z,U)  - \alpha_{0}(t,Z,U)\mid Z = z, Y(t) = 1\},
\label{eq::theta_z0_risk}
\end{equation}
with cumulative $\Theta_{\rm risk}(t,z) = \int_0^t \theta_{\rm risk}(s,z)\,\dd s$.\footnote{By the definition in~\eqref{eq::theta_z0_risk} we mean that $\theta_{\rm risk}(t,Z)$ is the $\sigma(Z)$-measurable function such that $ Y(t)\E\,\{\alpha_{1}(t,Z,U)  - \alpha_{0}(t,Z,U)\mid \mathcal{Z}_t\} = Y(t)\theta_{\rm risk}(t,Z)$ almost surely.} The assumption of the difference $\alpha_1(t,Z,U) - \alpha_0(t,Z,U)$ being functionally independent of the confounder is crucial for whether we are estimating $\theta(t,z_0)$ or the at-risk version $\theta_{\rm risk}(t,z_0)$, but is otherwise immaterial to the theory developed in the subsequent sections. 

\section{The data, the assumptions, and the estimator}\label{sec::rdd_in_surv}
In this section, we first elaborate on the assumptions made about what are in statistics and econometrics jargon, respectively, called the true model or the data generating process. Subsequently, in Section~\ref{subsec::covariate_localised_Aalen_estimator} we introduce a special case of our estimator, and provide some theory for this special case. The general large-sample theory is deferred to Section~\ref{subsec::large_sample}. 

\subsection{Data and assumptions}\label{sec::new_data_and_assumptions}
Let $(\widetilde{T}_i^0,\widetilde{T}_i^1,Z_i,U_i,C_i),\, i = 1,\ldots,n$ be independent replicates of $(\widetilde{T}^0,\widetilde{T}^1,Z,U,C)$, where this latter is as described in Sections~\ref{sec::model_and_estimands} and \ref{sec::data_and_assumptions}. In particular, the forcing variable $Z \in \real$ and $X = I\{Z \geq z_0\}$ for a known cut-off $z_0$. This entails that $(N_1,Y_1),\ldots,(N_n,Y_n)$ are independent replicates of $(N,Y)$ with $N(t) = XN^1(t) + (1 - X)N^0(t)$ and $Y(t) = XY^1(t) + (1 - X)Y^0(t)$ as defined in Section~\ref{sec::model_and_estimands}. As above, we assume that the processes $N_i$ and $Y_i$ are observed over the finite time interval $[0,\tau]$. The filtrations $\Galg_t$ and $\Falg_t$ are now $\Galg_t = \Ealg_t \vee \sigma(Z_i,U_i,\,i = 1,\ldots,n)$ and $\Falg_t =   \Ealg_t \vee \sigma(Z_i,\,i = 1,\ldots,n)$, with $\Ealg_t^g = \sigma(\{(N_i(s),Y_i(s)\}_{s \leq t},\,i = 1,\ldots,n)$. For $g = 0,1$, denote $M_1^g,\ldots,M_n^g$ the independent replicates of the martingale $M^g$ in~\eqref{eq::potential_MGs}. These are orthogonal local square integrable martingales with respect to the filtration $\Gscr_t^g$. Similarly, $I_{X_1 = g}M_1^{g}(t),\ldots,I_{X_n = g}M_n^{g}(t)$ are orthogonal local square integrable martingales with respect to the filtration $\Galg_t$ (see Lemma~\ref{lemma::composite_intensity}), and $I_{X_1 = g}\bar{M}_1^{g}(t),\ldots,I_{X_n = g}\bar{M}_n^{g}(t)$ are orthogonal local square integrable martingales with respect to the filtration $\Falg_t$ (see Lemma~\ref{lemma::barMG}).

The survival functions associated with the hazard rates $\alpha_{0}(t,Z,U)$ and $\alpha_{1}(t,Z,U)$ are denoted $S_{g}(t,Z,U) = \exp\{- \int_0^t \alpha_{g}(s,Z,U)\,\dd s\}$ for $g = 0,1$, and we set  
\begin{equation}
\bar{S}_{g}(t,z) = \E\,\{S_{g}(t,Z,U) \mid Z = z\},\quad \quad\text{for $g = 0,1$}.
\label{eq::def_Splusmin}
\end{equation}
Let $y_g(t,z)$ to be the conditional expectations of the at-risk process $Y^g(t)$ given $Z = z$. Using the assumption in~\eqref{eq::indep_censoring}, these functions are  
\begin{equation}
y_g(t,z) = \E\,\{Y^g(t) \mid Z = z\} = \{1 - H(t)\}\bar{S}_{g}(t,z),\quad \text{for $g = 0,1$},
\label{eq::smallyg}
\end{equation}
where $H(t)$ is the distribution function of the censoring variable $C$. Without further mention, the following is assumed throughout the paper
\begin{assumption}\label{assumption::indep} The potential lifetimes and the censoring variable are independent given the covariates, that is $\widetilde{T}^0\indep \widetilde{T}^1 \indep C \mid (Z,U)$. 
\end{assumption}
\begin{assumption}\label{assumption::lifetimes} Conditionally on $Z = z$ and $U=u$, the potential lifetimes stem from distributions with hazard rate functions $\alpha_0(t,z,u)$ and $\alpha_1(t,z,u)$ are continuous in $t$ for all $(z,u)$. 
\end{assumption}
\begin{assumption}\label{assumption::censoring} The censoring time $C$ stems from a distribution with a continuous distribution function $H(t)$ that is such that $H(\tau) < 1$. 
\end{assumption} 
\begin{assumption}\label{assumption::potential_params} For some $\kappa_0 > 0$ the following hold on an interval $(z_0 - \kappa_0,z_0 + \kappa_0)$ around the cut-off $z_0$:
\begin{itemize}\itemsep-0.2em
\item[{\rm (a)}] The density $f_Z(z)$ is continuous and bounded away from zero;
\item[{\rm (b)}] The conditional hazards $\bar{\alpha}_0(t,z)$ and $\bar{\alpha}_1(t,z)$, defined in~\eqref{eq::def_alpha0}, are $S$ times continuously differentiable in $z$ for all $t \in [0,\tau]$.
\end{itemize}
\end{assumption}
\begin{assumption}\label{assumption::kernel} For some $\kappa > 0$, the kernel function $k \colon [0,\kappa] \to \real$ is bounded and nonnegative, zero outside its support, and positive and continuous on $(0,\kappa)$. 
\end{assumption}
Assumption~\ref{assumption::indep} gives the martingale representation in~\eqref{eq::potential_MGs}, and is the key to Lemma~\ref{lemma::composite_intensity}, and thereby also to Lemma~\ref{lemma::barMG}. Assumption~\ref{assumption::lifetimes} is standard in survival analysis, and is equivalent to requiring absolute continuity of the survival functions. The assumption also entails that the conditional hazards $\bar{\alpha}_0(t,z)$ and $\bar{\alpha}_1(t,z)$ are continuous in $t$ for all $z$. Assumption~\ref{assumption::censoring} is needed to ensure that $y_g(t,z)$ is bounded below (a fact that is, for example, used in the proof of Theorem~\ref{theorem::Falg_clt}). Assumption~\ref{assumption::potential_params} allows for the estimation of the parameters of interest using a regression discontinuity design. For a discussion of the analogue of this assumption in the standard regression discontinuity design, see~\citet{hahn2001identification}, and \citet[Assumption~2.1, p.~618]{imbens2008regression}. In the next section a special case of our main estimator is presented.

\subsection{A covariate-localised Aalen estimator}\label{subsec::covariate_localised_Aalen_estimator}
The estimator we propose for $\Theta(t,z_0)$ is a weighted version of the Aalen additive hazards estimator (\citet{Aalen80,Aalen89,Aalen93}, \citet[Ch.~4.2]{ABG08}, and \citet[Ch.~VII.4]{ABGK93}).\footnote{Throughout this section we assume, for notational convenience, that Assumption~\ref{assumption::func_indep} holds. If this assumption does not hold, then all the theory of this section translates directly to the estimation of $\Theta_{\rm risk}(t,z_0)$.} The idea is to fit local polynomial regression models in the vicinity of the cut-off $z_0$. Thus, {`}local{'} here refers to an interval on the real line on which the forcing variable $Z$ takes its values, and does {\it not} refer to the time axis.  Specifically, the conditional expectations $\bar{\alpha}_{0}(t,z)$ and $\bar{\alpha}_{1}(t,z)$, as defined in~\eqref{eq::def_alpha0}, are approximated on intervals to the left and to the right of $z_0$ by $p$th order local polynomials. Since, for $g = 0,1$, the conditional expectation $\bar{\alpha}_{g}(t,z)$ may be approximated by
\begin{equation}
\bar{\alpha}_{g}(t,z_0)
+ \bar{\alpha}_{g}^{(1)}(t,z_0)(z - z_0)
+ \frac{\bar{\alpha}_{g}^{(2)}(t,z_0)}{2!}(z - z_0)^2
+ \cdots 
+ \frac{\bar{\alpha}_{g}^{(p)}(t,z)}{p!}(z - z_0)^{p},
 \notag
\end{equation}
where $\bar{\alpha}_{g}^{(\nu)}(t,z)$ is the $\nu$th derivative of $\bar{\alpha}_{g}(t,z)$ with respect to $z$, the local polynomial regression version of the Aalen additive hazards estimator that we introduce, is an estimator of $\int_0^t \bar{\alpha}_{g}^{(\nu)}(s,z)/\nu!\,\dd s$ for $\nu = 0,1,\ldots,p$. We start by presenting the local linear estimator, and then move on to general results for $p$th order local polynomial estimators in Section~\ref{subsec::large_sample}. 

For a bandwidth $h > 0$ and a kernel function $K_h(u) = K(u/h)/h$ with $K(u) = k(-u)I\{u < 0\} + k(u)I\{u \geq 0\}$,\footnote{In principle, different kernels could be used on either side of the cut-off. For simplicity, we employ the same kernel on both sides as this does not affect the theory.} where $k$ is a function satisfying Assumption~\ref{assumption::kernel}, the local linear estimator $(\widehat{B}_{g,1},\widehat{B}_{g,1}^{(1)})$ is, for $g = 0,1$, given by 
\begin{equation}
\begin{pmatrix}
\dd\widehat{B}_{g,1}(t,h)\\
\dd\widehat{B}_{g,1}^{(1)}(t,h)
\end{pmatrix} 
= G_{g,1,n}(t,h)^{-1} n^{-1}\sum_{i=1}^n I_{X_i = g} K_h(Z_i - z_0)
\begin{pmatrix}
1\\
Z_i - z_0
\end{pmatrix}
\,\dd N_i(t),
\notag
\end{equation}
with 
\begin{equation}
G_{g,1,n}(t,h) = n^{-1}\sum_{i=1}^n I_{X_i = g}K_h(Z_i - z_0)Y_i(t)
\begin{pmatrix}
1\\
Z_i - z_0
\end{pmatrix}
\begin{pmatrix}
1\\
Z_i - z_0
\end{pmatrix}^{\tr}.
\notag
\end{equation}
This is seen to be a version the Aalen additive hazards estimator estimator with kernel weights on the forcing variable. An estimator for $\Theta(t,z_0)$ is then given by
\begin{equation}
\widehat{\Theta}(t,h) = \widehat{B}_{1,1}(t,h) - \widehat{B}_{0,1}(t,h).
\label{eq::estimator1}
\end{equation}
For $t$ fixed, this estimator is simply the difference between the intercepts of two local linear regression to the left and to the right of the cut-off. In other words, it is the analogue of the most common estimator of the average treatment effect at the cut-off in the standard regression discontinuity design, see, for example, the estimator $\widehat{\tau}_{\rm srd}$ in~\citet[pp.~623--625]{imbens2008regression}.

Under the assumptions of Lemma~\ref{lemma::general_bias} below, the bias of the estimator $\widehat{\Theta}(t,h)$ can be described by  
\begin{equation}
\E\,(\dd \widehat{\Theta}(t,h)\mid \Falg_{t-})/\dd t = \theta(t,z_0)
+ h^2\, {\rm bias}(t) + o_p(h^2),
\label{eq::Thetahat_expectation1}
\end{equation}
with
\begin{equation}
{\rm bias}(t) = e_{1,0}^{\tr} \Gamma_1^{-1}\vartheta_{1,2}
\frac{\bar{\alpha}_{1}^{(2)}(t,z_0)
- \bar{\alpha}_{0}^{(2)}(t,z_0)}{2}.
\notag
\end{equation}
Here $\bar{\alpha}_{g}^{(2)}(s,z),\,g = 0,1$ are the unknown second derivates of the conditional expectations defined in~\eqref{eq::def_alpha0}; while $e_{1,0} = (1,0)^{\tr}$, and $\Gamma_1$, and $\vartheta_{1,2}$ are quantities that only depend on the chosen kernel, and need not be estimated from the data. They are
\begin{equation}
\Gamma_{1} = \int_0^1 K(u)
\begin{pmatrix}
1 & u \\
u & u^2\\
\end{pmatrix}
\,\dd u,
\quad \text{and}
\quad 
\notag \vartheta_{1,2} = \int_0^1 K(u) u^2 
\begin{pmatrix}
1  \\
u \\
\end{pmatrix}
\,\dd u.
\end{equation} 
Here and elsewhere in the paper, the notation is borrowed from \citet{calonico2014robust,calonico2014supplement}. If the bandwidth minimising the mean squared error is chosen, namely $h_n = cn^{-1/5}$ for some constant $c > 0$, then, conditionally on the filtration $\Falg_t$ (here Lemma~\ref{lemma::barMG} is invoked), we have process convergence in the space of c{\`a}dl{\`a}g functions on $[0,\tau]$, 
\begin{equation}
n^{2/5}\{\widehat{\Theta}(\cdot,h) - \Theta(\cdot,z_0)\} \Rightarrow c^2 \int_0^{\cdot}{\rm bias}(s)\,\dd s + c^{-1/2} e_{1,0}^{\tr}(\bar{M}_{1,1} - \bar{M}_{0,1}), 
\label{eq::clt_illustration}
\end{equation}
as $n \to \infty$, where $\bar{M}_{0,1}$ and $\bar{M}_{1,1}$ are independent bivariate Gaussian martingales with variation processes 
\begin{equation}
\langle \bar{M}_{g,1},\bar{M}_{g,1}\rangle_t = \frac{1}{f_Z(z_0)} 
\int_0^t \frac{\bar{\alpha}_g(s,z_0)}{y_g(s,z_0)}\,\dd s \, 
\Gamma_{1}^{-1}\vartheta_{1,2}\Gamma_{1}^{-1}, \quad \text{for $g = 0,1$}.
\notag
\end{equation}
The general version of this result is the content of Corollary~\ref{corollary::Falg_Thetaconv}. A consistent estimator of $\langle \bar{M}_{g,1},\bar{M}_{g,1}\rangle_t$ is introduced in Section~\ref{subsec::variance_estimation}.

To avoid or to get rid of the bias term that appears on the right hand side of~\eqref{eq::clt_illustration}, two approaches are discussed in this paper. First, one can choose a bandwidth $h_n$ such that $nh_n^{5}\to 0$, meaning that the bandwidth must be {`}smaller{'} than the mean squared error optimal one. Second, the bias term may be removed by subtracting off a consistent estimate of the bias term. The standard approch is the following (see, e.g.,~\citet[Section~4.4]{fan1996local}): Suppose that ${\rm Bias}_n(t,b_n)$ is consistent for $\int_0^t {\rm bias}(s)\,\dd s$ (uniformly in $t$) as $n \to \infty$ and the so-called pilot bandwidth $b_n \to 0$, then (see Corollary~\ref{corollary::bias_correction}) 
\begin{equation}
n^{2/5}\{\widehat{\Theta}(\cdot,h_n) - \Theta(\cdot,z_0) - h_n^2 {\rm Bias}_n(\cdot,b_n)\} \Rightarrow  c^{-1/2} e_{1,0}^{\tr}(\bar{M}_{1,1} - \bar{M}_{0,1}), 
\notag
\end{equation}
as $n \to \infty$, provided, among other things, that $h_n/b_n \to 0$. This latter condition ensures that the variability of the bias correction estimation disappears, meaning that $\bar{M}_{0,1}$ and $ \bar{M}_{1,1}$ are the Gaussian martingales from~\eqref{eq::clt_illustration} (the variation process does not change). As pointed out in the influential paper \citet{calonico2014robust}, $h_n/b_n$ is never zero in finite samples, and therefore, the variability associated with the bias correction ought to be accounted for in the limiting distribution. In the present paper, results of this type are presented in Section~\ref{sec::bias_correction}.

\section{General theory}\label{subsec::large_sample}
In this section we first consider the general $p$th order local polynomial regression estimator of the $\nu$th derivative function $\int_0^t \bar{\alpha}_g^{(\nu)}(s,z_0)\,\dd s$, and derive a representation for the bias of this estimator. Next, in Section~\ref{subsec::limiting_normality}, we present two central limit theorems for this estimator. Throughout this section, Assumption~\ref{assumption::func_indep} (the functional independence assumption) is, for notational convenience, assumed to hold. If this assumption does not hold, all subsequent results are true with $\theta(t,z_0)$ replaced by $\theta_{\rm risk}(t,z_0)$ (see the discussion in Section~\ref{subsec::estimands}). Moreover, $g = 0,1$ is the index used to indicate the potential outcomes, or statistics depending on these, for non-treated and treated, respectively. Since the theory we develop is the same on both sides of the cut-off, a result concerning a quantity with subscript $g$ means that it applies for both $g = 0$ and $g = 1$. Most of the proofs of the claims made in the present section are deferred to the appendices. 
 
\subsection{The estimator and its bias}
For an integer $p \geq 1$, let $r_p(x) = (1,x,x^2,\ldots,x^p)^{\tr}$, and define the estimator $\widehat{B}_{g,p}$ by 
\begin{equation}
\dd \widehat{B}_{g}(t,h) = J_{n,h}(t)G_{g,p,n}(t,h)^{-1}\frac{1}{n}\sum_{i=1}^n I_{X_i = g}K_h(Z_i-z_0)r_p(Z_i-z_0)\,\dd N_i(t),
\label{eq::general_estimator1}
\end{equation} 
with
\begin{equation}
G_{g,p,n}(t,h) = \frac{1}{n}\sum_{i=1}^n I_{X_i = g}K_h(Z_i - z_0 )Y_i(t)r_{p}(Z_i-z_0)r_p(Z_i-z_0)^{\tr},
\notag
\end{equation} 
and 
\begin{equation}
J_{n,h}(t) = I\{\text{$G_{0,p,n}(t,h)$ and $G_{1,p,n}(t,h)$ are positive definite}\}.
\notag
\end{equation} 
If $J_{n,h}(t) = 0$, we take $\dd \widehat{B}_{g,p}(t,h) = 0$ for $g = 0,1$. Let $H_p(h) = {\rm diag}(1,h^{-1},\ldots,h^{-p})$. Using that $H_p(h)r_p(z) = r_p(z/h)$ and $r_p(z) = H_{p}(h)^{-1}r_p(z/h)$, the estimator above can be expressed as 
\begin{equation}
\dd \widehat{B}_{g,p}(t,h) = J_{n,h}(t)H_p(h)\Gamma_{g,p,n}(t,h)^{-1}\frac{1}{n}\sum_{i=1}^nI_{X_i = g}K_h(Z_i-z_0)r_p\big(\frac{Z_i-z_0}{h}\big)\,\dd N_i(t),
\notag
\end{equation}
with 
\begin{equation}
\Gamma_{g,p,n}(t,h) = \frac{1}{n}\sum_{i=1}^n 
I_{X_i = g}K_h(Z_i -z_0 )Y_i(t) r_p\big(\frac{Z_i-z_0}{h}\big)r_p\big(\frac{Z_i-z_0}{h}\big)^{\tr},
\label{eq::myGamma}
\end{equation}
and, since $G_{g,p,n}(t,h)$ is positive definite if and only if $\Gamma_{g,p,n}(t,h)$ is positive definite, $J_{n,h}(t) = I\{\text{$\Gamma_{0,p,n}(t,h)$ and $\Gamma_{1,p,n}(t,h)$ are positive definite}\}$. Let $e_{p,\nu}$ be the $(p+1)$-dimensional column vector with its $(\nu+1)$th element equal to $1$, and all other elements equal to zero, for example, $e_{1,0} = (1,0)^{\tr}$, $e_{2,1} = (0,1,0)^{\tr}$, $e_{3,2} = (0,0,1,0)^{\tr}$, and so on. Denote $\bar{\alpha}_{g}^{(\nu)}(t,z)$ the $\nu$th partial derivative of $\bar{\alpha}_{g}(t,z)$ with respect to $z$, so $\bar{\alpha}_{g}^{(0)}(t,z) = \bar{\alpha}_{g}(t,z)$. With this notation, 
\begin{equation}
\widehat{A}_{g,p}^{(\nu)}(t,h) = e_{p,\nu}^{\tr}\nu!\widehat{B}_{g,p}(t,h),\quad \text{for $g = 0,1$},
\label{eq::one_side_estimator}
\end{equation}
are estimators for $\int_0^t \bar{\alpha}_0^{(\nu)}(s,z_0)\,\dd s$ and $\int_0^t \bar{\alpha}_1^{(\nu)}(s,z_0)\,\dd s$, respectively. An estimator for the $\nu$th derivative of $\Theta(t,z) = \int_0^t \theta(s,z)\,\dd s$ with respect to $z$ and evaluated in $z_0$ (see~\eqref{eq::theta_z0}), based on a $p$th order polynomial regression estimator, is then 
\begin{equation}
\widehat{\Theta}_p^{(\nu)}(t,h) = \widehat{A}_{1,p}^{(\nu)}(t,h) - \widehat{A}_{0,p}^{(\nu)}(t,h). 
\label{eq::general_estimator}
\end{equation}
Before we state the bias lemma, we also need the following quantities,
\begin{equation}
\Gamma_{p} = \int_{0}^1 K(u) r_p(u)r_p(u)^{\tr}\,\dd u,\quad  \vartheta_{p,q} = \int_{0}^1  K(u)u^q r_p(u)\,\dd u,
\label{eq::RSplus}
\end{equation}
and 
\begin{equation}
\Psi_p = \int_0^1 K(u)^2 r_p(u)r_p(u)^{\tr}\,\dd u.
\notag
\end{equation}
The matrices $\Gamma_p$ and $\Psi_p$ are of dimension $(p+1)\times(p+1)$, while $\vartheta_{p,q}$ are $(p+1)$-dimensional column vectors. It will be assumed throughout that the kernel $K$ is chosen so that $\Gamma_p$, for all relevant $p$, is positive definite. This entails that the $\Gamma_{p}^{-1}$ is also positive definite, and that $\norm{\Gamma_{p}}$ and $\norm{\Gamma_{p}^{-1}}$ are finite, with $\norm{\cdot}$ here denoting the matrix norm. The $\Falg_t$-compensator of $I_{X= g}N^g(t)$, namely $I_{X=g}\int_0^tY^g(s)\bar{\alpha}_g(s,Z)\,\dd s$, is absolutely continuous. This entails that the $\Falg_t$-compensator of $\widehat{B}_{g,p}^{(\nu)}(t,h)$ is also absolutely continuous, the derivative of this compensator therefore exists, and we denote these derivatives $\E\,(\dd \widehat{B}_{g,p}^{(\nu)}(t,h) \mid \Falg_{t-})/\dd t$, and $\E\,(\dd \widehat{A}_{g,p}^{(\nu)}(t,h) \mid \Falg_{t-})/\dd t = e_{p,\nu}^{\tr}\nu!\E\,(\dd \widehat{B}_{g,p}^{(\nu)}(t,h) \mid \Falg_{t-})/\dd t$. Introduce the column vectors of partial derivatives 
\begin{equation}
\beta_{g,p}(t) = \{\bar{\alpha}_{g}(t,z_0),\bar{\alpha}_{g}^{(1)}(t,z_0)/1!,\bar{\alpha}_{g}^{(2)}(t,z_0)/2!,\ldots,\bar{\alpha}_{g}^{(p)}(t,z_0)/p!\}^{\tr}. 
\notag
\end{equation}
and the vector valued functions
\begin{equation}
\bfrak_{g,p}(t) =  \Gamma_{p}^{-1}\vartheta_{p,p+1}
\frac{\bar{\alpha}^{(p+1)}_g(t,z_0)}{(p+1)!}\,\dd s,\quad \text{and}
\quad 
\Bfrak_{g,p}(t) = \int_0^t \bfrak_{g,p}(s)\,\dd s. 
\label{eq::mu_bias}
\end{equation}
We can now state the bias lemma.
\begin{lemma}\label{lemma::general_bias}{\sc (Bias)} The condtions of Lemma~\ref{lemma::A4} hold, and 
Assumption~\ref{assumption::potential_params}(b) holds with $S \geq p+2$. Then
\begin{equation} 
\E\,(\dd\widehat{B}_{g,p}(t,h)\mid \Falg_{t-})/\dd t   = J_{n,h}(t)\beta_{g,p}^{(\nu)}(t,z_0)
+ H_{p}(h)h^{p+1}\bfrak_{g,p}(t) + H_p(h)O_p(h^{p+2}),
\notag
\end{equation} 
uniformly in $t \in [0,\tau]$, as $nh \to \infty$ and $h \to 0$.
\end{lemma}
\begin{proof} The proof is in Appendix~\ref{subsec::bias_lemma}. 
\end{proof}
\begin{remark}
This lemma is a relative of Theorem 3.1 in \citet[p.~62]{fan1996local}. The difference between $p - \nu$ odd and $p - \nu$ even discussed in relation to that theorem, does not apply here, however. This is becase the kernel $k$ is not symmetric to the left nor to the right of the cut-off. 
\end{remark}

\subsection{Limiting normality}\label{subsec::limiting_normality}
The central limit theorems coming up concern the sequences 
\begin{equation}
\widehat{\Theta}_{p}^{(\nu)}(t,h) - \int_0^t \{\bar{\alpha}_1^{(\nu)}(s,z_0) - \bar{\alpha}_0^{(\nu)}(s,z_0)\}\,\dd s,
\notag
\end{equation}
properly normalised. Since $\widehat{\Theta}_{p}^{(\nu)}(t,h) = \widehat{A}_{1,p}^{(\nu)}(t,h) - \widehat{A}_{0,p}^{(\nu)}(t,h)$, and the theory for both sides of the cut-off are the same, we consider only $\widehat{A}_{g,p}^{(\nu)}(t,h) - \int_0^t \bar{\alpha}_g^{(\nu)}(s,z_0)\, \dd s$ for a generic $g = 0,1$. These sequences are
\begin{equation}
\widehat{A}_{g,p}^{(\nu)}(t,h) - \int_0^t \bar{\alpha}_g^{(\nu)}(s,z_0)\,\dd s
= e_{p,\nu}^{\tr}\nu!\{\widehat{B}_{g,p}(t,h) - \int_0^t\beta_{g,p}(s)\,\dd s\},
\label{eq::estimator_minus_estimand}
\end{equation}  
and can be decomposed in different ways depending on which filtration we analyse it with respect to. From Lemma~\ref{lemma::composite_intensity} we have that $I_{X=g}M_i^g(t) = I_{X=g}\{N_i^g(t) - \int_0^tY_i^g(s)\alpha_g(s,Z_i,U_i)\,\dd s\}$ are martingales with respect to the filtration $\Galg_t = \Ealg_t \vee \Xalg$, and from Lemma~\ref{lemma::barMG} that $\bar{M}_i^g(t) = I_{X_i = g}\{N_{i}^g(t) - \int_0^t Y_i^g(s) \bar{\alpha}(s,Z_i)\,\dd s\}$ are martingales with respect to the filtration $\Falg_{t} = \Ealg_t \vee \Xalg^{\rm obs}$ of observables. Recall that $\Xalg^{\rm obs} \subset \Xalg$, and that the confounders are not measurable with respect to $\Falg_t$ (see Sections~\ref{sec::data_and_assumptions} and~\ref{sec::new_data_and_assumptions} for the definitions). This, in turn, leads to two decompositions of the sequence in~\eqref{eq::estimator_minus_estimand}, and, consequently, to two different central limit theorems for this sequence. Write,
\begin{equation}
\begin{split}
& \widehat{B}_{g,p}(t,h) - \int_0^t\beta_{g,p}(s,z_0)\,\dd s\\
& \qquad\quad = \widehat{B}_{g,p}(t,h) - \int_{0}^t \E\,\{\dd\widehat{B}_{g,p}(s,h)\mid \Falg_{s-}\} 
+ H_p(h){\rm Bias}_{g,p,n}(t,h),
\end{split}
\label{eq::Falg_decomp}
\end{equation}
where ${\rm Bias}_{g,p,n}(t,h)$ is, using Lemma~\ref{lemma::general_bias},
\begin{equation}
\begin{split}
{\rm Bias}_{g,p,n}(t,h)  & = \int_0^t \E\,\{\dd \widehat{B}_{g,p}(s,h)\mid \Falg_{s-}\} 
- \int_0^t \beta_{g,p}(s)\,\dd s\\
& = \int_0^t \{1 - J_{n,h}(s)\}\beta_{g,p}(s,z_0) \,\dd s
+ h^{p+1} \Bfrak_{g,p}(t) + O_p(h^{p+2}).
\end{split}
\label{eq::Bias}
\end{equation} 
The first term on the right in~\eqref{eq::Falg_decomp} can be written  
\begin{equation}
\widehat{B}_{g,p}(t,h) - \int_{0}^t \E\,\{\dd\widehat{B}_{g,p}(s,h)\mid \Falg_{s-}\} = H_p(h)\bar{M}_{g,p,n}(t,h),  
\notag
\end{equation}
where $\bar{M}_{g,p,n}(t,h)$ is a martingale with respect to the filtration $\Falg_t$, namely
\begin{equation}  
\begin{split}
& \bar{M}_{g,p,n}(t,h) = \int_0^t J_{n,h}(s)\Gamma_{g,p,n}(s,h)^{-1}\\
& \qquad \qquad \qquad  \times \frac{1}{n}\sum_{i=1}^nI_{X_i = g}K_h(Z_i-z_0)r_p\big(\frac{Z_i - z_0}{h}\big)\,\dd \bar{M}_i^g(s),
\end{split}
\label{eq::pdim_MGplus1} 
\end{equation} 
where $\bar{M}_i^g$ are the $\Falg_t$-martingales of Lemma~\ref{lemma::barMG}.

\begin{theorem}\label{theorem::Falg_clt}{\sc ($\Falg_t$-clt)} The condtions of Lemma~\ref{lemma::A4} hold. As $nh \to \infty$ and $h \to 0$
\begin{equation}
(nh)^{1/2} (\bar{M}_{0,p,n}(\cdot,h),\bar{M}_{1,p,n}(\cdot,h))
\Rightarrow (\bar{M}_{0,p},\bar{M}_{1,p}),
\notag
\end{equation}
where $\bar{M}_{0,p}$ and $\bar{M}_{1,p}$ are orthogonal mean zero Gaussian martingales with variation processes 
\begin{equation}
\langle \bar{M}_{g,p},\bar{M}_{g,p}\rangle_{t}
= \frac{1}{f_Z(z_0)} \int_0^t \frac{\bar{\alpha}_{g}(s,z_0)}{y_{g}(s,z_0)}\,\dd s \,
\Gamma_p^{-1}\Psi_p\Gamma_p^{-1}.
\label{eq::qv_barMG}
\end{equation}
\end{theorem} 
\begin{proof} Note first that $nh\langle \bar{M}_{0,p,n}(\cdot,h),\bar{M}_{1,p,n}(\cdot,h)\rangle_t = 0$ for all $t$ and $n$, and is therefore zero in the limit. Assume, without loss of generality, that $z_0 = 0$. Define $\widetilde{\Gamma}_{1,p}(t) = y_{1}(t,z_0)f_Z(z_0)\Gamma_{p}$ and $\widetilde{\Gamma}_{0,p}(t) = y_{0}(t,z_0)f_Z(z_0)H_p(-1)\Gamma_{p}H_p(-1)$, and set $D_{g,p,n}(t,h) = J_{n,h}(t)\Gamma_{g,p,n}(t,h)^{-1} - \widetilde{\Gamma}_{g,p}(t)^{-1}$ for $g = 0,1$. Then 
\begin{equation}
\bar{M}_{g,p,n}(t,h) = \xi_{g,p,n}(t,h) + r_{g,p,n}(t,h), 
\notag
\end{equation}
say, where,
\begin{equation}
\xi_{g,p,n}(t,h) = \int_0^t \widetilde{\Gamma}_{g,p}(s)^{-1}n^{-1}\sum_{i=1}^nI_{X_i=g}K_h(Z_i)r_p(Z_i/h)\,\dd \bar{M}_i^g(s),
\label{eq::xi_process}
\end{equation}
and
\begin{equation}
r_{g,p,n}(t,h) = \int_0^t D_{g,p,n}(s,h)n^{-1}\sum_{i=1}^nI_{X_i=g}K_h(Z_i)r_p(Z_i/h)\,\dd \bar{M}_i^g(s).
\notag
\end{equation}
For any $\nu = 0,1,\ldots,p$, the predictable quadratic variation of $e_{p,\nu}^{\tr} r_{g,p,n}(t,h)$ with respect to $\Falg_t$ is, using Lemma~\ref{lemma::simple_ineq} in the appendix, 
\begin{equation}
\begin{split}
&e_{p,\nu}^{\tr}\langle r_{g,p,n}(\cdot,h),r_{g,p,n}(\cdot,h)\rangle_{t}e_{p,\nu}\\
& = \frac{1}{(nh)^2}\sum_{i=1}^nI_{X_i = g}K(Z_i/h)^2
\int_0^t e_{p,\nu}^{\tr}D_{g,p,n}(s,h)r_p(Z_i/h)r_p(Z_i/h)^{\tr}D_{g,p,n}(s,h)e_{p,\nu}\\ 
& \qquad \qquad \qquad \qquad \times Y_i^g(s)\bar{\alpha}_g(s,Z_i)\,\dd s\\
& \leq \frac{1}{(nh)^2}\sum_{i=1}^nI_{X_i = g}K(Z_i/h)^2
\int_0^t \norm{D_{g,p,n}(s,h)}^2\norm{r_p(Z_i/h)}^2\,Y_i^g(s)\bar{\alpha}_g(s,Z_i)\,\dd s,
\end{split}
\notag
\end{equation}
Since $K(z/h)\,\norm{r_p(z/h)}^2 = K(z/h)\sum_{\nu=0}^p (z/h)^{2\nu} \leq K(z/h)\sum_{\nu=0}^p \kappa^{2\nu}$ and $K(u)$ is bounded (see Assumption~\ref{assumption::kernel}),
\begin{equation}
\begin{split}
& e_{p,\nu}^{\tr}\langle r_{g,p,n}(\cdot,h),r_{g,p,n}(\cdot,h)\rangle_{t}e_{p,\nu}  \\
& \qquad \quad \lesssim \sup_{t \in [0,\tau]}\norm{D_{g,p,n}(t,h)}^2
\frac{1}{(nh)^2} \sum_{i=1}^n K(Z_i/h) \int_0^{\tau} Y_i^g(s)\bar{\alpha}_g(s,Z_i)\,\dd s.
\end{split}
\notag
\end{equation} 
From Lemma~\ref{lemma::prob_limits}(i)--(iii) combined with Lemma~\ref{lemma::invertible_to_limit} in the appendix, we have that $\sup_{t\in[0,\tau]}\norm{D_{g,p,n}(t,h)} = o_p(1)$ as $nh \to \infty$ and $h \to 0$. For $h \kappa \leq \kappa_0$ (see Assumption~\ref{assumption::potential_params}), $\bar{\alpha}_g(t,z)$ is bounded, thus $(nh)^{-1} \sum_{i=1}^n K(Z_i/h) \int_0^{\tau} Y_i^g(s)\bar{\alpha}_g(s,Z_i)\,\dd s\lesssim (nh)^{-1} \sum_{i=1}^n K(Z_i/h) = O_p((nh)^{-1/2})$, from which  $e_{p,\nu}^{\tr}\langle r_{g,n,p}(\cdot,h),r_{g,n,p}(\cdot,h)\rangle_{\tau}e_{p,\nu}=o_p((nh)^{-3/2})$. Therefore, for $\nu = 0,\ldots,p$,  $\sup_{t\in[0,\tau]}\abs{(nh)^{1/2}e_{p,\nu}^{\tr}r_{g,p,n}(t,h)}\, = o_p(1)$, by Lenglart{'}s inequality \citep[Lemma~I.3.30, p.~35]{jacod2003limit}. The predictable quadratic variation of $(nh)^{1/2}\xi_{g,p,n}$ with respect to $\Falg_{t}$ is 
\begin{equation} 
\langle \xi_{g,p,n}(\cdot,h),\xi_{g,p,n}(\cdot,h)\rangle_t
= \frac{1}{n}\int_0^t \widetilde{\Gamma}_{g,p}(s)^{-1} \bar{\Psi}_{g,p,n}(s,h) \widetilde{\Gamma}_{g,p}(s)^{-1}\,\dd s,
\notag
\end{equation} 
where $\bar{\Psi}_{g,p,n}(t,h) = n^{-1}\sum_{i=1}^n I_{X_i = g} K_h(Z_i)^2 r_{p}(Z_i/h)r_{p}(Z_i/h)^{\tr} Y_i^g(t)\bar{\alpha}_g(t,Z_i)$. By Lemma~\ref{lemma::prob_limits}(viii) in the appendix $h\bar{\Psi}_{1,p,n}(t,h) \to_p y_1(t,z_0)\bar{\alpha}_1(t,z_0)f_{Z}(z_0)\Psi_p$ and $h\bar{\Psi}_{0,p,n}(t,h) \to_p y_0(t,z_0)\bar{\alpha}_0(t,z_0)f_{Z}(z_0)H_p(-1)\Psi_pH_p(-1)$ uniformly in $t$ as $nh \to \infty$ and $h \to 0$. We conclude that 
\begin{equation}
\begin{split}
& \frac{1}{nh}\int_0^t\widetilde{\Gamma}_{g,p}(s)^{-1}h\bar{\Psi}_{g,p,n}(s,h)\widetilde{\Gamma}_{g,p}(s)^{-1} \,\dd s \\
& \qquad\qquad = \frac{1}{nh}\big\{\frac{1}{f_Z(z_0)} \int_0^t \frac{\bar{\alpha}_{g}(s,z_0)}{y_{g}(s,z_0)}\,\dd s \,\Gamma_{p}^{-1}\Psi_{p}\Gamma_{p}^{-1} + o_p(1)\big\},
\end{split}
\label{eq::barMGvar}
\end{equation}
as $nh \to \infty$ and $h \to 0$, for all $t$. Next, we show that $(nh)^{1/2}\xi_{g,p,n}$ satisfies a Lindeberg condition. Write $H_{g,p,i}^n(t,h) = (nh)^{-1/2}\widetilde{\Gamma}_{g,p}(s)^{-1}I_{X_i = g}K(Z_i/h)r_p(Z_i/h)$ so that $\xi_{g,p,n}(t,h) = \sum_{i=1}^n\int_0^t H_{g,p,i}^n(s,h)\,\dd \bar{M}_i^g(s)$. Note that for each $i$ and $t$ and $\nu$,
\begin{equation}
\abs{e_{p,\nu}^{\tr}H_{g,p,i}^n(t,h)}^2\, 
\leq (nh)^{-1/2}\,\norm{\widetilde{\Gamma}_{g,p}(t)^{-1}}^2 \norm{K(Z_i/h)r_p(Z_i/h)}^2,
\notag
\end{equation}
by Lemma~\ref{lemma::simple_ineq} in the appendix. Now, $\norm{K(Z_i/h)r_p(Z_i/h)}\, \leq K_{\max}(\sum_{\nu = 0 }^{p}\kappa^{2\nu})^{1/2}$, where $K_{\max} = \sup_{u}K(u)$ is bounded by Assumption~\ref{assumption::kernel}. Furthermore, $\norm{\widetilde{\Gamma}_{g,p}(s)^{-1} }\, = \abs{y_g(s,z_0)f_Z(z_0)}^{-1}\,\norm{\Gamma_{p}^{-1}}$, by Assumption~\ref{assumption::censoring} we have that $y_g(t,z_0) \geq y_g(\tau,z_0) = \{1 - H(\tau)\}\bar{S}_g(\tau,z_0) > 0$, and by Assumption~\ref{assumption::potential_params}(a) that $f_{Z}(z_0) > 0$. Therefore, 
\begin{equation}
\sup_{s\in [0,\tau]} \norm{\widetilde{\Gamma}_{g,p}(s)^{-1} } \leq \abs{y_g(\tau,z_0)f_Z(z_0)}^{-1}\,\norm{\Gamma_{p}^{-1}} \eqqcolon C. 
\notag
\end{equation} 
This shows that 
\begin{equation}
\max_{\nu = 0,\ldots,p}\sup_{t \in [0,\tau]} \abs{e_{p,\nu}^{\tr}H_{g,p,i}^n(t,h)}\,
\leq \frac{C K_{\max}(\sum_{\nu = 0 }^{p}\kappa^{2\nu})^{1/2}}{(nh)^{1/2}},
\notag
\end{equation}
and, consequently, for all $\nu = 0,\ldots,p$, we have that for any $\eps > 0$,  
\begin{equation}
\int_0^t \sum_{i=1}^n \,\abs{e_{p,\nu}^{\tr}H_{g,p,i}^n(s,h)}^2\, I\{\abs{e_{p,\nu}^{\tr}H_{g,p,i}^n(s,h)}\, \geq \eps\}Y_i^g(s)\bar{\alpha}_g(s,Z_i)\,\dd s \overset{p}\to 0 ,
\notag
\end{equation} 
for all $t \in [0,\tau]$ as $nh \to \infty$. That is, the Lindeberg condition holds. By the Rebolledo type central limit theorem in  \citet[Theorem~I.2, p.~1115]{andersen1982cox}, this entails that $(nh)^{1/2}\xi_{g,p,n} \Rightarrow \bar{M}_{g,p}$ as $nh \to \infty$ and $h \to 0$. Since
\begin{equation}
(nh)^{1/2}\bar{M}_{g,p,n} = (nh)^{1/2}\xi_{g,p,n} + (nh)^{1/2}r_{g,p,n}, 
\notag
\end{equation}
and $(nh)^{1/2}r_{g,p,n} \to_p 0$ uniformly in $t$, Lemma~VI.3.31 in \citet[p.~352]{jacod2003limit} yields $(nh)^{1/2}\bar{M}_{g,p,n} \Rightarrow \bar{M}_{g,p}$ as $nh \to \infty$ and $h \to 0$. Because $\bar{M}_{0,p,n}$ and $\bar{M}_{1,p,n}$ are orthogonal, joint convergence follows.
\end{proof}
From the above theorem it is seen that the predictable quadratic variation of $\bar{M}_{g,p,n}$ is $\langle \bar{M}_{p,n}(\cdot,h),\bar{M}_{p,n}(\cdot,h)\rangle_t
= O_p((nh)^{-1})$, so, in particular 
\begin{equation}
e_{p,\nu}^{\tr}H_p(h)\langle \bar{M}_{p,n}(\cdot,h),\bar{M}_{p,n}(\cdot,h)\rangle_tH_p(h)e_{p,\nu}
= O_p(n^{-1}h^{-2\nu-1}),
\notag
\end{equation} 
Moreover, from Lemma~\ref{lemma::general_bias} the bias is seen to be 
\begin{equation}
e_{p,\nu}^{\tr}H_{p}(h){\rm Bias}_{g,p,n}(t,h){\rm Bias}_{g,p,n}(t,h)^{\tr}H_{p}(h)e_{p,\nu} = O_p\big(h^{2(p+1-\nu)}).
\notag
\end{equation} %
This shows that the mean squared error optimal bandwidth is $h = c n^{1/(2p+3)}$ for some constant $c > 0$.
Combining the bias lemma with Theorem~\ref{theorem::Falg_clt}, we get the following corollary. 

\begin{corollary}\label{corollary::Falg_Thetaconv} The condtions of Lemma~\ref{lemma::A4} hold. As $nh \to \infty$ and $h \to 0$, if $nh^{2p + 3} \to c $ for a constant $c \geq 0$, then 
\begin{equation} 
(nh^{2\nu + 1})^{1/2}\{\widehat{A}_{g,p}^{(\nu)}(\cdot,h)
- \int_0^{\cdot}\bar{\alpha}_g^{(\nu)}(s,z_0)\,\dd s \} 
\Rightarrow e_{p,\nu}^{\tr}\nu!(\bar{M}_{g,p} + c^{1/2}\Bfrak_{g,p}), 
\notag
\end{equation} 
for $g = 0,1$, with $\Bfrak_{g,p}$ as defined in~\eqref{eq::mu_bias}.
\end{corollary}
\begin{proof} By Lemma~\ref{lemma::general_bias} and the decomposition in~\eqref{eq::Falg_decomp}, 
\begin{equation}
\begin{split}
& \widehat{A}_{g,p}^{(\nu)}(t,h)
- \int_0^t \bar{\alpha}_g^{(\nu)}(s,z_0)\,\dd s
= e_{p,\nu}^{\tr}\nu! h^{-\nu} \bar{M}_{g,p,n}(t,h) + e_{p,\nu}^{\tr}\nu! h^{p+1 - \nu} \bfrak_{g,p}(t)\\
& \qquad\qquad\qquad + \int_0^t \{J_{n,h}(s) - 1\}\bar{\alpha}_g(s,z_0)\,\dd s + O_p(h^{p + 2 - \nu}).
\end{split}
\notag
\end{equation}
Upon multiplying $ \widehat{A}_{g,p}^{(\nu)}(t,h)
- \int_0^t \bar{\alpha}_g^{(\nu)}(s,z_0)\,\dd s$  by $(nh^{2\nu + 1})^{1/2}$, the first term on the right is $e_{p,\nu}^{\tr}\nu! (nh)^{1/2}\bar{M}_{g,p,n}(t,h)$. By Theorem~\ref{theorem::Falg_clt} and the Cram{\'e}r--Wold device, $e_{p,\nu}^{\tr}\nu! (nh)^{1/2}\bar{M}_{g,p,n}(\cdot,h)\Rightarrow e_{p,\nu}^{\tr}\nu! \bar{M}_{g,p}$ as $nh \to \infty$ and $h \to 0$, where $\bar{M}_{g,p}$ is the Gaussian martingale of said theorem. For the second term on the right, i.e.,~the bias term, $e_{p,\nu}^{\tr}\nu! (nh^{2p+3})^{1/2}\bfrak_{g,p}(t)\to e_{p,\nu}^{\tr}\nu! c^{1/2}\bfrak_{g,p}(t)$ as $nh^{2p+3} \to c \geq 0$, uniformly in $t$ since $\bfrak_{g,p}(t)$ is bounded, see Assumption~\ref{assumption::potential_params}(b). For the third term on the right, we can, by Lemma~\ref{lemma::prob_limits}(iii) in the appendix, find $n_0$ and $h_0$ such that $\sup_{t \in [0,\tau]}\abs{J_{n,h}(t) - 1}\, = 0$ a.s., for all $n \geq n_0$ and $h \leq h_0$. Thus, for all $n \geq n_0$ and $h \leq h_0$, $(nh^{2\nu + 1})^{1/2}\sup_{t\in [0,\tau]}\abs{\int_0^t \{J_{n,h}(s) - 1\}\bar{\alpha}_g(s,z_0)\,\dd s} \, \leq (nh^{2\nu + 1})^{1/2}\sup_{t\in [0,\tau]}\{J_{n,h}(s) - 1\}\int_0^{\tau}\bar{\alpha}_g(s,z_0)\,\dd s = 0$ a.s., using that $\bar{\alpha}_g(t,z_0)\geq 0$. Finally, $(nh^{2\nu + 1})^{1/2}O_p(h^{p+2 - \nu}) = O_p((nh^{2p+5})^{1/2}) \to 0$ as $nh^{2p+3} \to c$, and this convergence is uniform in $t$ by Lemma~\ref{lemma::general_bias}. Since $\bfrak_{g,p}(t)$ is continuous in $t$, the claim follows from Proposition~VI.3.17 and Lemma~VI.3.31 in \citet[pp.~350--352]{jacod2003limit}. 
\end{proof} 
The two results above, Theorem~\ref{theorem::Falg_clt} and Corollary~\ref{corollary::Falg_Thetaconv}, pertain to the sequences of $\Falg_t$-martingales $H_p(h)\bar{M}_{g,p,n} = \widehat{B}_{g,p}(t,h) - \int_{0}^t \E\,\{\dd\widehat{B}_{g,p}(s,h)\mid \Falg_{s-}\}$. Considering these martingales essentially means that, in the central limit theorem, we average out the confounders but condition on the forcing variable.\footnote{There is a parallel here to the difference between the observed information and the Fisher information in likelihood inference for regression models, where, in the latter, the covariates are averaged out.} We now turn to a central limit theorem relative to the filtration $\Galg_t = \Ealg_t \vee \Xalg$, that is, the filtration with respect to which the confounder is measurable. Because the $\Falg_t$-martingales $\bar{M}_{g,p,n}$ analysed above are {\it not} martingales with respect to $\Galg_t$, these results are slightly more involved. Consider the decomposition 
\begin{equation}
\begin{split}
& \widehat{B}_{g,p}(t,h) - \int_0^t\beta_{g,p}(s,z_0)\,\dd s \\
& \qquad \quad = H_p(h) \{M_{g,p,n}(t,h) +  L_{g,p,n}(t,h) +  {\rm Bias}_{g,p,n}(t,h)\},
\end{split}
\label{eq::Galg_decomp}
\end{equation}
where ${\rm Bias}_{g,p,n}(t,h)$ is as defined in~\eqref{eq::Bias}; $M_{g,p,n}(t,h)$ is the $\Galg_t$-martingale 
\begin{equation}  
M_{g,p,n}(t,h) = \int_0^t J_{n,h}(s)\Gamma_{g,p,n}(s,h)^{-1}\frac{1}{n}\sum_{i=1}^nI_{X_i = g}K_h(Z_i-z_0)r_p\big(\frac{Z_i - z_0}{h}\big)\,\dd M_i^g(s);
\notag
\end{equation}
see Lemma~\ref{lemma::composite_intensity}, and $L_{g,p,n}(t,h)$ is
\begin{equation}  
L_{g,p,n}(t,h) = \int_0^t J_{n,h}(s)\Gamma_{g,p,n}(s,h)^{-1}Q_{g,p,n}(s,h)\,\dd s,
\notag
\end{equation}
where $Q_{g,p,n}(s,h)$ is the average of i.i.d.~random variables given by 
\begin{equation}
Q_{g,p,n}(s,h) = \frac{1}{n}\sum_{i=1}^nI_{X_i=g}K_h(Z_i-z_0)r_p\big(\frac{Z_i-z_0}{h}\big)Y_i^g(s)\Delta_{g,i}(s), 
\label{eq::Qplus} 
\notag
\end{equation}
with $\Delta_{g,i}(t)$ being the difference between true hazard and the conditional hazards defined in~\eqref{eq::def_alpha0}, that is 
\begin{equation}
\Delta_{g,i}(t) = \alpha_g(t,Z_i,U_i) - \bar{\alpha}_g(t,Z_i), \quad \text{for $i = 1,\ldots,n$}.
\notag
\end{equation}
It turns out that the $\Galg_t$-martingales $M_{g,p,n}$ have the same limiting variance as the $\Falg_t$-martingales $\bar{M}_{g,p,n}$, however, the sequence $L_{g,p,n}$ is of the same order as the martingales, and thus contribute to the asymptotic variance of the estimator sequence. One might therefore view the variance stemming from $L_{g,p,n}$ as extra variance induced by not averaging out the confounder $U$. For the next theorem, in addition to Assumptions~\ref{assumption::indep}--\ref{assumption::kernel}, we impose a Lipschitz condition and a boundedness condition on the hazard rate functions of the potential lifetimes, as well as on the conditonal hazard rate functions. These assumptions are likely much stronger than necessary.

\begin{assumption}\label{assumption::Galg_clt_extra} There are constants $\ell_g$, $\bar{\ell}_g$, $\alpha_{g,\max}$, and $\bar{\alpha}_{g,\max}$ for $g = 0,1$ such that $\alpha_g(t,z,u) \leq \alpha_{g,\max}$, $\alpha_g(t,z) \leq \bar{\alpha}_{g,\max}$ for all $t,z,u$, and $\abs{\alpha_g(t,z,u) - \alpha_g(s,z,u)}\, \leq \ell_g\, \abs{t  - s }$ for all $t,s \in [0,\tau]$ and all $z,u$, and $\abs{\bar{\alpha}_g(t,z) - \bar{\alpha}_g(s,z)}\, \leq \bar{\ell}_g\, \abs{t  - s }$ for all $t,s \in [0,\tau]$ and all $z$.
\end{assumption}

\begin{theorem}{\sc($\Galg$-clt)}\label{theorem::clt1} Assumption~\ref{assumption::potential_params}(b)
holds with $S \geq p+2$, Assumption~\ref{assumption::Galg_clt_extra} is in force, and the distribution function $H(t)$ of $C$ has density $h(t)$ that is bounded on $[0,\tau]$. If $nh^{2p+3} \to c$ for some $c \geq 0$ as $nh \to \infty$ and $h \to 0$, then
 \begin{equation}
(nh^{2\nu + 1})^{1/2}\{\widehat{A}_{g,p}^{(\nu)}(\cdot,h) - \int_0^{\cdot}\bar{\alpha}_g^{(\nu)}(s,z_0)\,\dd s\} \Rightarrow 
e_{p,\nu}^{\tr}\nu!(M_{g,p} + L_{g,p} + c^{1/2}  \Bfrak_{g,p}),
\notag
\end{equation}
where $M_{g,p}$ and $L_{g,p}$ are $(p+1)$-dimensional mean zero Gaussian processes, with $(M_{0,p},L_{0,p})$ and $(M_{1,p},L_{1,p})$ independent, and finite-dimensional distributions characterised by 
\begin{equation}
\var\,M_{g,p}(t)
= \frac{1}{f_{Z}(z_0)}\int_0^t\frac{\bar{\alpha}_{g}(s,z_0)}{y_{g}(s,z_0)}\,\dd s\,\Gamma_{p}^{-1}\Psi_{p}\Gamma_p^{-1},
\notag
\end{equation}
and 
\begin{equation}
\cov\{M_{g,p}(s),L_{g,p}(t) \}
= -\frac{1}{f_Z(z_0)}\int_0^t \int_0^{u \wedge s} \frac{c_g(x,u \vee s,z_0)}{y_g(x,z_0)y_{g}(u,z_0)}\, \dd x\,\Gamma_{p}^{-1}\Psi_{p}\Gamma_p^{-1},
\notag
\end{equation}
and
\begin{equation}
\cov\{L_{g,p}(s),L_{g,p}(t) \}
= \frac{1}{f_Z(z_0)} \int_0^s \int_0^t \frac{y_g(u\vee v,z_0)c_g(u,v,z_0)}{y_g(u,z_0)y_g(v,z_0)}\,\dd u \,\dd v \,\Gamma_{p}^{-1}\Psi_{p}\Gamma_p^{-1},
\notag
\end{equation}
where $c_{g}(u,s,z) = \E\,(\Delta_{g}(u)\Delta_{g}(s)\mid Z = z,Y(u\vee s ) = 1)$; and the bias term $\Bfrak_{g,p}$ is as defined in~\eqref{eq::mu_bias}. 
\end{theorem}
\begin{proof} The proof is in Appendix~\ref{app::clt_proof}.  
\end{proof}

\begin{remark} There is an interesting affinity between Theorem~\ref{theorem::clt1} of the present paper and Theorem~2.1 in~\citet[p.~361]{hjort1992inference}. Hjort{'}s central limit theorem pertains to parametric hazard rate functions under model misspecification. In that theorem, an extra term corresponding to the $L_{g,p,n}$ appears. This term is due to the difference $\alpha(t,\theta_{\rm lf}) - \alpha_{0}(t)$ between the parametric model $\alpha(t,\theta)$, evaluated in the least false parameter value $\theta_{\rm lf}$, and the true hazard $\alpha_{0}(t)$, that is the hazard under which $N(t) - \int_0^t Y(s) \alpha(s)\,\dd s$ is a martingale. In our case it is the difference $\alpha_{g}(t,Z_i,U_i) - \bar{\alpha}_{g}(t,Z_i)$ that cause the extra term to appear, but this difference is a close nonparametric relative of the difference studied by Hjort.  The misspecification in Theorem~\ref{theorem::clt1} occurs because the $\Falg$-hazard $\bar{\alpha}_{g}(t,Z_i)$ is a form of model misspecifiation when the analysis takes place with respect to the filtration $\Galg_t$.
\end{remark}
For the results in the remainder of the paper, we use the $\Falg$-central limit theorem in Theorem~\ref{theorem::Falg_clt}. There are two reasons for this. First, estimating the limiting variance of the martingale is more straightforward than estimating the variances and covariances appearing in Theorem~\ref{theorem::clt1}. Second, the common approach in the regression discontinuity literature is to condition on the forcing variable, and average out all other covariates (see the discussion in the Section~\ref{sec::intro}).

\subsection{Variance estimation}\label{subsec::variance_estimation} The probability limit of the process $nh\langle \bar{M}_{g,p,n}(\cdot,h),\bar{M}_{g,p,n}(\cdot,h)\rangle_t$ is $\langle \bar{M}_{g,p},\bar{M}_{g,p}\rangle_t$, for which an expression is given in \eqref{eq::qv_barMG}. Consistent estimators for the variance processes $e_{p,\nu}^{\tr}H_p(h_n)\langle \bar{M}_{g,p},\bar{M}_{g,p}\rangle_tH_p(h_n)e_{p,\nu}$ can be developed along the lines of the standard variance
estimator for the variance of the estimator of the cumulative regression coefficients in the Aalen additive hazards model (see, for example, \citet[Section 2.1]{hjort2021partly}). One such estimator is  
\begin{equation}
\begin{split}
& V_{g,p,n}(t,h) = \frac{h}{n}\sum_{i=1}^n I_{X_i=g} K_h(Z_i-z_0)^2 
\int_0^t J_{n,h}(s)\Gamma_{g,p,n}(s,h)^{-1} \\
& \qquad\qquad \qquad\qquad \times r_{p}\big((Z_i - z_0)/h\big)r_{p}\big((Z_i - z_0)/h\big)^{\tr}
\Gamma_{g,p,n}(s,h)^{-1}\,\dd N_i(s).
\end{split}
\label{eq::variance_estimator}
\end{equation}

\begin{lemma}\label{lemma::var_consistency} The condtions of Lemma~\ref{lemma::A4} hold. As $nh \to \infty$ and $h \to 0$, 
\begin{equation}
V_{g,p,n}(t,h) \overset{p}\to \langle \bar{M}_{g,p},\bar{M}_{g,p}\rangle_t
\notag
\end{equation}
for each $t \in [0,\tau]$. 
\end{lemma}
\begin{proof} Assume that $z_0 = 0$.
Using the martingale decomposition in Lemma~\ref{lemma::barMG}, we see that $V_{g,p,n}(t,h) = nh  \langle \bar{M}_{g,p,n}(\cdot,h),\bar{M}_{g,p,n}(\cdot,h)\rangle_t + r_{g,p,n}(t,h)$. For $\nu = 0,\ldots,p$, the $\nu$th element of $r_{g,p,n}(t,h)$ is   
\begin{equation}
e_{p,\nu}^{\tr}r_{g,p,n}(t,h) = \frac{h}{n}\sum_{i=1}^n I_{X_i=g} K_h(Z_i)^2 
\int_0^tJ_{n,h}(s) \{e_{p,\nu}^{\tr}\Gamma_{g,p,n}(s,h)^{-1} r_{p}(Z_i/h) \}^{2}\,\dd \bar{M}_i^g(s).
\notag
\end{equation}
From Lemma~\ref{lemma::simple_ineq} we get the bound
\begin{equation}
\big(e_{p,\nu}^{\tr}\Gamma_{g,p,n}(s,h)^{-1} r_{p}(Z_i/h) \big)^{2}
\leq \norm{\Gamma_{g,p,n}(s,h)^{-1}}^2\norm{r_{p}(Z_i/h)}^2, 
\notag
\end{equation}
so the predictable quadratic variation of $e_{p,\nu}^{\tr}r_{g,p,n}(t,h)$, relative to $\Falg_t$, is 
\begin{equation}
\begin{split}
& e_{p,\nu}^{\tr}\langle r_{g,p,n}(\cdot,h),r_{g,p,n}(\cdot,h) \rangle_t e_{p,\nu}^{\tr}\\ 
& \; = 
\frac{h^2}{n^2}\sum_{i=1}^n I_{X_i=g} K_h(Z_i)^4 
\int_0^t J_{n,h}(s)\{e_{p,\nu}^{\tr}\Gamma_{g,p,n}(s,h)^{-1} r_{p}(Z_i/h) \}^{4}Y_i^g(s)\bar{\alpha}_g(s,Z_i)\,\dd s\\
& \; \leq \frac{1}{(nh)^2}\sum_{i=1}^n I_{X_i=g} K(Z_i/h)^4 
\int_0^t \norm{\Gamma_{g,p,n}(s,h)^{-1}}^4\norm{r_{p}(Z_i/h)}^4Y_i^g(s)\bar{\alpha}_g(s,Z_i)\,\dd s, 
\end{split}
\notag
\end{equation}
By Lemma~\ref{lemma::prob_limits}(i)--(iii) and Lemma~\ref{lemma::invertible_to_limit} in the appendix $\norm{J_{n,h}(s)\Gamma_{1,p,n}(s,h)^{-1}}$ converges in probability to $\norm{y_g(s,z_0)f_Z(z_0)\Gamma_{1}(s)^{-1}}$ and $\norm{J_{n,h}(s)\Gamma_{0,p,n}(s,h)^{-1}}$ converges in probability to $\norm{y_g(s,z_0)f_Z(z_0)H_p(-1)\Gamma_{1}(s)^{-1}H_p(-1)}$, uniformly in $s \in [0,\tau]$ as $nh \to \infty$ and $h \to 0$. Both limits are bounded by $\abs{y_g(\tau,z_0)f_Z(z_0)}^{-1}\,\norm{\Gamma_{p}^{-1}}$ and $K(z/h) \norm{r_{p}(z/h)} \leq K(z/h) \sum_{\nu = 0}^p\kappa^{2\nu}$, as argued in the proof of Theorem~\ref{theorem::Falg_clt}. Therefore, 
\begin{equation}
e_{p,\nu}^{\tr}\langle r_{g,p,n}(\cdot,h),r_{g,p,n}(\cdot,h) \rangle_t e_{p,\nu} \lesssim (nh)^{-2}\sum_{i=1}^n  
K(Z_i/h)\int_0^t Y_i^g(s)\bar{\alpha}_g(s,Z_i)\,\dd s. 
\notag
\end{equation} 
For $h \kappa < \kappa_0$, the conditional hazard $\bar{\alpha}_g(t,z)$ is bounded and $\int_0^{\tau} Y_i(s)\bar{\alpha}_g(s,Z_i)\,\dd s = O_p(1)$ for all $i$. Moreover, $(nh)^{-1}\sum_{i=1}^n K(Z_i/h)$ is a sum of i.i.d.~random variables with finite variance (Assumption~\ref{assumption::kernel}) converging in probability to $\int_{-\kappa}^{\kappa} K(u) f_Z(hu)\,\dd u$ which is finite provided $h \kappa < \kappa_0$. This gives that $e_{p,\nu}^{\tr}\langle r_{g,p,n}(\cdot,h),r_{g,p,n}(\cdot,h) \rangle_t e_{p,\nu} = O_p((nh)^{-1})  = o_p(1)$ as $nh \to \infty$. By Lenglart{'}s inequality we now conclude that $e_{p,\nu}^{\tr} r_{g,p,n}(t,h) = o_p(1)$ for each $\nu$, and therefore $r_{g,p,n}(t,h) = o_p(1)$. 
\end{proof} 
This means that a consistent estimator of the limiting variance of the sequence $(nh^{2\nu + 1})^{1/2}\{\widehat{A}_{g,p}^{(\nu)}(t,h) - \int_0^{t}\bar{\alpha}_g^{(\nu)}(s,z_0)\,\dd s \}$ of Corollary~\ref{corollary::Falg_Thetaconv} is $h^{2\nu} (\nu!)^2 e_{p,\nu}^{\tr} V_{g,p,n}(t,h)e_{p,\nu}$. In particular, with $\Phi(z)$ the standard normal cumulative distribution function, 
\begin{equation}
\widehat{A}_{g,p}^{(\nu)}(t,h) \pm \Phi^{-1}(1 - \alpha/2)(\nu!)^2 e_{p,\nu}^{\tr} V_{g,p,n}(t,h)e_{p,\nu}/(nh)^{1/2}
\notag
\end{equation}
are approximate $(1 - \alpha)100$ percent pointwise confidence intervals provided $nh^{2p+3} \to 0$. If $nh^{2p+3} \to c$ for $c>0$, these confidence intervals are not valid due to the bias term appearing in the limit in Corollary~\ref{corollary::Falg_Thetaconv}. The topic of the next section is how this may be fixed.

\subsection{Bias correction}\label{sec::bias_correction}
In this section we follow the conventional bias correction approach in local polynomial regression (see, e.g.,~\citet[Section~4.4]{fan1996local}) and study the estimators given by
\begin{equation}
\widehat{\Theta}_{p,q}^{(\nu),{\rm bc}}(t,h,b) = \widehat{\Theta}_{1,p,q}^{(\nu),{\rm bc}}(t,h,b) - \widehat{\Theta}_{0,p,q}^{(\nu),{\rm bc}}(t,h,b),\quad
\text{$q \geq p + 1$}
\notag
\end{equation}
where
\begin{equation} 
\widehat{\Theta}_{g,p,q}^{(\nu),{\rm bc}}(t,h,b) = \widehat{\Theta}_{g,p}^{(\nu)}(t,h) - h^{p+1-\nu}e_{p,\nu}^{\tr}\nu! \frac{\Gamma_{p}^{-1}\vartheta_{p,p+1}}{(p+1)!} \widehat{\Theta}_{g,q}^{(p+1)}(t,b),
\label{eq::subtract_off_bias}
\end{equation}
Here $b$ is a so-called pilot bandwidth sequence, typically larger than $h$ (because it is harder to estimate the $(p+1)$th derivative than the $\nu$th derivative when $\nu \leq p$, as it is here). In~\eqref{eq::subtract_off_bias} the second term on the right is an estimator of the bias term appearing in the limiting distribution of Corollary~\ref{corollary::Falg_Thetaconv}. Under certain conditions on the bandwidth sequences $h$ and $b$ made precise below, subtracting off a bias estimate removes the asymptotic bias term from the limiting distribution in Corollary~\ref{corollary::Falg_Thetaconv}, even when $nh^{2p+3}$ tends to a positive constant. In particular, the mean squared error optimal bandwidth $h \propto n^{-1/(2p+3)}$ may be employed, and the limiting martingale processes are the same as those given in said corollary. 

As pointed out by \citet[p.~2302]{calonico2014robust}, this large-sample approximation relies on the condition $h/b \to 0$ as $h,b \to 0$, which makes the variability of the bias correction estimate disappear. That is, provided $h/b \to 0$, the estimator $\widehat{\Theta}_{g,q}^{(p+1)}(t,b)$ in~\eqref{eq::subtract_off_bias} does not contribute to the limiting variance of $\widehat{\Theta}_{g,p,q}^{(\nu),{\rm bc}}(t,h,b)$. Since $h/b$ is never zero in finite samples, the idea of \citet{calonico2014robust} is to remove this requirement, and instead let $h/b \to \rho > 0$ as $h,b\to 0$, and thereby get limiting distributions of the estimator sequence where the variability of the bias estimator is accounted for. These ideas are formalised in Corollary~\ref{corollary::bias_correction} below.

By Lemma~\ref{lemma::general_bias} and using the decomposition in~\eqref{eq::Falg_decomp} twice, we obtain 
%
\begin{equation}
\begin{split}
& (nh^{2\nu+1})^{1/2}\{\widehat{\Theta}_{g,p,q}^{(\nu),{\rm bc}}(t,h,b) 
- \int_0^{t}\bar{\alpha}_g^{(\nu)}(s,z_0)\,\dd s \}\\
& \; = (nh)^{1/2}e_{p,\nu}^{\tr}\nu!\big\{\bar{M}_{g,p,n}(t,h) 
- (h/b)^{p+1}\Gamma_{p}^{-1}\vartheta_{p,p+1} e_{q,p+1}^{\tr} \bar{M}_{g,q,n}(t,b)\big\}\\
& \qquad \quad
- (nh)^{1/2}h^{p+1}b^{q-p} e_{p,\nu}^{\tr}\Gamma_{p}^{-1}\vartheta_{p,p+1} 
e_{q,p+1}^{\tr} \Gamma_{q}^{-1}\vartheta_{q,q+1} \int_0^t \bar{\alpha}_g^{(q+1)}(s,z_0)\,\dd s\\
%
& \qquad\qquad \qquad   + O_p((nh^{2p+3})^{1/2}(h + b^{q-p})).
\end{split}
\label{eq::bias_correct_decomp}
\end{equation} 
From Lemma~\ref{lemma::prob_limits}(iii) in the appendix, this approximation holds with probability one when $n,h,b$ are so that $J_{n,h}(t) = 1$ and $J_{n,b}(t) = 1$ for all $t \in [0,\tau]$, thus ensuring that $(nh^{2\nu+1})^{1/2}\int_0^t \{J_{n,h}(s) - 1\}\bar{\alpha}_g^{(\nu)}(s,z_0)\,\dd s\ + (nh)^{1/2}h^{p+1}e_{p,\nu}\nu! \kappa_p \int_0^t \{J_{n,b}(s) - 1\}\bar{\alpha}_g^{(p+1)}(s,z_0)\,\dd s = 0$ for all $t$, almost surely.

\begin{corollary}\label{corollary::bias_correction} The condtions of Lemma~\ref{lemma::A4} hold. Assume that $\bar{\alpha}_g(t,z)$ for $g = 0,1$ are at least $q+2 \geq p + 3$ times continuously differentiable in $z$ for all $t$; that $n \min(h_n,b_n) \to \infty$, $\max(h_n,b_n) \to 0$, $h_n/b_n \to \rho \geq 0$ and $nh_n^{2p+3}\max\{h_n^2,b_n^{2(q - p)}\}$ tends to zero as $n \to \infty$, and that $h_n \leq b_n$ for all $n$. Then 
\begin{equation}
(nh_n^{2\nu + 1})^{1/2}\{\widehat{\Theta}_{p,q}^{(\nu),{\rm bc}}(\cdot,h,b) - \Theta^{(\nu)}(\cdot,z_0)\}	
\Rightarrow e_{p,\nu}^{\tr}\nu! (\bar{M}_{1,p}^{(\rho)} - \bar{M}_{0,p}^{(\rho)} ),
\notag
\end{equation}
where $\bar{M}_{0,p}^{(\rho)}$ and $\bar{M}_{1,p}^{(\rho)}$ are orthogonal mean zero Gaussian martingales with variation processes 
\begin{equation}
\langle \bar{M}_{g,p}^{(\rho)},\bar{M}_{g,p}^{(\rho)} \rangle_t 
 =  \langle \bar{M}_{g,p},\bar{M}_{g,p}\rangle_t + \rho^{2p+3} V_{g,p,q}(t) - 2\rho^{p+2} C_{g,p,q}(t),
\notag
\end{equation}
with
\begin{equation}
V_{g,p,q}(t) = \Gamma_p^{-1}\vartheta_{p,p+1} e_{q,p+1}^{\tr}
\langle \bar{M}_{g,q},\bar{M}_{g,q} \rangle_t e_{q,p+1}\vartheta_{p,p+1}^{\tr}\Gamma_{p}^{-1},
\notag
\end{equation}
and
\begin{equation}
C_{g,p,q}(t) = \langle \bar{M}_{g,p},\bar{M}_{g,q}\rangle_t e_{q,p+1} \vartheta_{p,p+1}^{\tr}\Gamma_p^{-1},
\notag
\end{equation}
where $\langle \bar{M}_{g,r},\bar{M}_{g,r} \rangle_t$ for $r = p,q$ are as defined in~\eqref{eq::qv_barMG}, while $\langle \bar{M}_{g,p},\bar{M}_{g,q}\rangle_t$ is the $(p+1)\times (q+1)$ matrix   
\begin{equation}
\langle \bar{M}_{g,p},\bar{M}_{g,q}\rangle_t = \frac{1}{f_{Z}(z_0)} \int_0^1 \frac{\bar{\alpha}_g(s,z_0)}{y_g(s,z_0)}\,\dd s \,\Gamma_{p}^{-1}\Psi_{p,q}(\rho) \Gamma_q^{-1}, 
\notag
\end{equation}
where $\Psi_{p,q}(\rho) = \int_0^{\infty}K(u)K(\rho u) r_{p}(u)r_q(\rho u)^{\tr}\,\dd u$. 
\end{corollary}
\begin{proof} It suffices to look at one of the sides of the cut-off. Define the $\Falg$-martingales 
\begin{equation}
\bar{M}_{g,p,q,n}(t,h_n,b_n) = 
\bar{M}_{g,p,n}(t,h_n) 
- (h_n/b_n)^{p+1}\Gamma_{p}^{-1}\vartheta_{p,p+1} e_{q,p+1}^{\tr} \bar{M}_{g,q,n}(t,b_n),
\notag
\end{equation}
so that $\bar{M}_{g,p,q,n}(t,h_n,b_n)$ is the martingale in the curly brackets in~\eqref{eq::bias_correct_decomp}. The predictable quadratic variation of $(nh_n)^{1/2}\bar{M}_{g,p,q,n}(t,h_n,b_n)$ is 
\begin{equation}
\begin{split}
& nh_n \langle \bar{M}_{g,p,q,n}(\cdot ,h_n,b_n),\bar{M}_{g,p,q,n}(\cdot ,h_n,b_n)\rangle_t 
= nh_n\langle \bar{M}_{g,p,n}(\cdot,h_n),\bar{M}_{g,p,n}(\cdot,h_n)\rangle_t\\
& \;  + (nh_n)(h_n/b_n)^{2(p+1)} \Gamma_{p}^{-1}\vartheta_{p,p+1}e_{q,p+1}^{\tr}
\langle \bar{M}_{g,q,n}(\cdot,b_n),\bar{M}_{g,q,n}(\cdot,b_n)\rangle_t e_{q,p+1}\vartheta_{p,p+1}^{\tr}\Gamma_{p}^{-1}\\
&\; - 2(nh_n) (h_n/b_n)^{p+1} \langle \bar{M}_{g,p,n}(\cdot,h_n),\bar{M}_{g,q,n}(\cdot,b_n)\rangle_t e_{q,p+1}
\vartheta_{p,p+1}^{\tr}\Gamma_p^{-1}.
\end{split}
\notag
\end{equation}
From Theorem~\ref{theorem::Falg_clt}, we get that $nh_n\langle \bar{M}_{g,p,n}(\cdot,h_n),\bar{M}_{g,p,n}(\cdot,h_n)\rangle_t \to_p \langle \bar{M}_{g,p},\bar{M}_{g,p}\rangle_t$ with an expression for this limit given in~\eqref{eq::qv_barMG}, and also that the second term converges in probability to $\rho^{2p + 3}V_{g,p,q}(t)$. The third variation process in the third term is 
\begin{equation}
\begin{split}
&\langle \bar{M}_{g,p,n}(\cdot,h_n),\bar{M}_{g,q,n}(\cdot,b_n)\rangle_t\\
& \; = \frac{1}{n b_n}\int_0^t J_{n,h_n}(s) \Gamma_{g,p,n}(s,h_n)^{-1}b_n\bar{\Psi}_{g,p,q,n}(s,h_n,b_n)J_{n,b_n}(s)\Gamma_{g,q,n}(s,b_n)^{-1},
\end{split}
\label{eq::qv_pq}
\end{equation}  
where
\begin{equation}
\bar{\Psi}_{g,p,q,n}(s,h,b)
= \frac{1}{n}\sum_{i=1}^n I_{X_i = g}
K_{h}(Z_i)K_{b}(Z_i)r_{p}(Z_i/h)r_{q}(Z_i/b)
Y_i^{g}(s)\bar{\alpha}_g(s,Z_i).
\label{eq::barPsi_hb}
\end{equation} 
By Lemma~\ref{lemma::prob_limits}(xii) $b_n\bar{\Psi}_{1,p,q,n}(s,h_n,b_n) \to_p  y_1(s,z_0)\bar{\alpha}_1(s,z_0)f_Z(z_0) \Psi_{p,q}(\rho)$ and (xiii) $b_n\bar{\Psi}_{0,p,q,n}(s,h_n,b_n) \to_p  y_0(s,z_0)\bar{\alpha}_0(s,z_0)f_Z(z_0) H_{p}(-1)\Psi_{p,q}(\rho)H_{q}(-1)$, uniformly in $s \in [0,\tau]$. Combining this with Lemma~\ref{lemma::prob_limits}(i)--(iii) and Lemma~\ref{lemma::invertible_to_limit}, we have that the integrand in~\eqref{eq::qv_pq} converges in probability, uniformly in $s \in [0,\tau]$, to $f_Z(z_0)^{-1}\{\bar{\alpha}_g(s,z_0)/y_g(s,z_0)\} \Gamma_{p}^{-1}\Psi_{p,q}(\rho) \Gamma_q^{-1}$, and $nb_n  \langle \bar{M}_{g,p,n}(\cdot,h_n),\bar{M}_{g,q,n}(\cdot,b_n)\rangle_t \to_p \langle \bar{M}_{g,p},\bar{M}_{g,q}\rangle_t$. It now follows from Theorem~\ref{theorem::Falg_clt} that $(nh_n)^{1/2}\bar{M}_{g,p,q,n}(\cdot ,h_n,b_n) \Rightarrow \bar{M}_{g,p}^{(\rho)}$ for $g = 0,1$. Consider now second and the third term in~\eqref{eq::bias_correct_decomp}. The second term is a nonrandom continuous function, and $(nh_n)^{1/2} h_n^{p+1}b^{q - p} = (nh_n^{2p + 3} b_n^{2(q - p)})^{1/2}$ tends to zero by assumption. The third and final term is $O_p((nh^{2p+3})^{1/2}(h + b^{q-p})) = O_p((nh^{2p+3}\max\{h^2,b^{2(q-p)}\} )^{1/2}) = o_p(1)$ by assumption. Again, Proposition VI.3.17 and the Cram{\'e}r--Slutsky like Lemma~VI.3.31 in~\citet[pp.~350--352]{jacod2003limit} yield the result.
\end{proof}

\section{Concluding remarks}
The ideas underlying designs such as the {\rdd}, the difference in difference design, and the intstrumental variable design are not bound to any particular estimand, model, or type of data. The statistical theory for these designs, however, are much more developed for the type of estimands, models, and data often encountered in economics, than for estimands, models, and data typically encountered in other fields of application. This paper is an attempt at taking one of these designs from the estimation of a conditional average treatment effect (\cate) based on uncensored data, to the estimation of a {\cate}-like object, namely the difference of two cumulative hazards, based on right-censored survival data. 

A few directions the results of the present paper can be extended in are: First, the estimator developed in this paper can be used as a building block in the estimation of $\theta(t,z_0)$ using kernel smoothing techniques similar to those introduced by \citet{ramlau1983smoothing} (research in this direction is underway). Second, the results of this paper is limited to the sharp {\rdd}, and ought to be extended to the fuzzy {\rdd}. Third, the estimator of this paper can be used to test whether the parameter of interest in a Cox regression model with unobserved confounders equals zero or not, but it can not provide asymptotically unbiased estimates of this parameter. Whether the {\rdd} can be used to identify the parameter of interest in proportional hazards models under confounding, can be studied.    

\appendix 
\section{Results used throughout the article}
As the title says, this appendix contains results used throughout the article.

\subsection{Conditional expectations}\label{app::cond_exp} 
This section contains two lemmata that are used in the proof of Lemma~\ref{lemma::A4}.


\begin{lemma} Let $X$ be a random variable on $(\Omega,\Gscr,\pr)$, and $\Fscr$ a sub-{\sigalg} of $\Gscr$. Assume that there is an event $A \in \Fscr$ such that $\omega \mapsto \E\,(X \mid \Fscr)(\omega)$ is constant over $A$, and set $Z = I_A\E\,(X \mid A) + I_{A^c}\E\,(X \mid \Fscr)$. Then $Z$ is a version of $\E\,(X \mid \Fscr)$, and, in particular, $I_A \E\,(X \mid \Fscr) = I_A\E\,(X \mid A)$ a.s. 
\end{lemma}
\begin{proof} By assumption $\E\,(X \mid \Fscr) = a$ say, for every $\omega$ in $A$. Then $\E\,(X \mid A) = \E\,(I_AX)/\pr(A) = \E\,\{I_A\E\,(X\mid \Fscr)\}/\pr(A) = \E\,(I_A \,a )/\pr(A) = a$, and so $I_A \E\,(X \mid \Fscr) - I_A\E\,(X \mid A) = 0$. To see that $Z$ is a version of the conditional expectation, note that, for any event $B$ in $\Fscr$, $\E\,(I_BZ) = \E\,\{I_{B \cap A}\E\,(X\mid A)\} + \E\,\{I_{B\cap A^c}\E\,(X \mid \Fscr)\} = \E\,\{I_{B \cap A}\E\,(X\mid \Fscr)\} + \E\,\{I_{B\cap A^c}\E\,(X \mid \Fscr)\} = \E\,\{I_{B}\E\,(X \mid \Fscr)\} = \E\,(I_{B}X)$.
\end{proof}
  
\begin{lemma}\label{lemma::lemmaA2} Let $X$ be a random variable on $(\Omega,\Gscr,\pr)$, and $\Fscr$ a sub-{\sigalg} of $\Gscr$. Let $A \in \Fscr$, and assume that there are no nonempty sets in $\Fscr$ that is a proper subset of $A$. Then $\omega \mapsto \E\,(X \mid \Fscr)(\omega)$ is constant over $A$.  
\end{lemma}
\begin{proof} We prove that if $\E\,(X \mid \Fscr)$ is not constant over $A$, then there must be at least one nonempty proper subset of $A$ in $\Fscr$. Let $Y = \E\,(X \mid \Fscr)$, and assume that there are $\omega_1\neq \omega_2$ in $A$ such that $Y(\omega_1) = a_1 \neq a_2 = Y(\omega_2)$. Then $A \cap Y^{-1}(a_1) \in \Fscr$ since $Y$ is a conditional expectation and $\Fscr$ is a {\sigalg}. Moreover, $A \cap Y^{-1}(a_1)$ is not empty since it must contain $\{\omega_1\}$, and $A \cap Y^{-1}(a_1) \subset A$.   
\end{proof}

\subsection{Proof of Lemma~\ref{lemma::A4}}\label{subsec::proofA4}
\begin{proof} Since $n = 1$, we write $\sigma(Z)$ instead of $\Xalg^{\rm obs}$, and $I_{T \geq t}$ instead of $Y(t)$ when convenient. Recall that $\Falg_{t} = \Ealg_t \vee \sigma(Z)$ with $\Ealg_t = \sigma(\{N(s),Y(s)\}_{s \leq t})$, and note that because $N(t)$ is c{\`a}dl{\`a}g, adapted, and nondecreasing, $T$ is an $\Falg_t$-stopping time \citep[Proposition~I.1.28, p.~7]{jacod2003limit}. We start by showing that the two processes in the lemma are modifications of each other, thereafter we show that they are both a.s.~left-continuous, hence indistinguishable \citep[Lemma~3.2.10, p.~79]{cohen2015stochastic}. 

{\it Modifications}: For each $t \in (0,\tau]$, let $\Pi_{t-}^{\Ealg}$ be the collection of events 
\begin{equation}
\Pi_{t-}^{\Ealg}
= \{\{Y(s) = 1\}_{s \leq t},\{N(s) = 0\}_{s < t},
\{Y(s) = 1\}_{s\leq t}\cap \{N(s) = 0\}_{s< t},\Omega\},
\notag
\end{equation}
and set $\Pi_{t-}^{\Ealg} = \Pi_{0}^{\Ealg}$ when $t = 0$. By $\{Y(s) = 1\}_{s\leq t}\cap \{N(s) = 1\}_{s < t}$, we mean all intersections of the form $\{Y(s_1) = 1\}\cap \{N(s_2) = 1\}$ for $s_1 \leq t,s_2 < t$. The collection of events $\Pi_{t-}^{\Ealg}$ generates $\Ealg_{t-}$. We now show that $\Pi_{t-}^{\Ealg}$ is a $\pi$-system: Since $Y(t) = I_{T \geq t}$ and $\{T \geq t\} \subset \{T \geq s\}$ for all $t \geq s$, we get that for all $s_1,s_2 \leq t$ 
\begin{equation}
\{Y(s_1) = 1\} \cap \{Y(s_2) = 1\} = \{Y(s_1\vee s_2) = 1 \} \in\Pi_{t-}^{\Ealg}.
\notag
\end{equation}
By the same argument, for all time points $s_1,s_2\leq t$ and all $s_3 < t$
\begin{equation}
\begin{split}
& \{Y(s_1) = 1\} \cap (\{Y(s_2) = 1\}\cap \{N(s_3) = 0\})\\
& \qquad \qquad= \{Y(s_1 \vee s_2) = 1\}\cap \{N(s_3) = 0\} \in \Pi_{t-}^{\Ealg}.
\end{split}
\notag
\end{equation}
The step function $N(t)$ is nondecreasing, so $\{N(t) = 0\} \subset \{N(s) = 0\}$ for all $s \leq t$, therefore, for all $s_1,s_3 < t$ and all $s_2 \leq t$,   
\begin{equation}
\begin{split}
& \{N(s_1) = 1\} \cap (\{Y(s_2) = 1\}\cap \{N(s_3) = 0\})\\
& \qquad \qquad  = \{Y(s_2) = 1\}\cap \{N(s_1 \vee s_3) = 0\} \in \Pi_{t-}^{\Ealg}.
\end{split}
\notag
\end{equation}
Finally, all intersections of the type $\{Y(s_1) = 1\}\cap \{N(s_2) = 0\}$ for $s_1\leq t,s_2 < t$ are in $\Pi_{t-}^{\Ealg}$ by definition, and for any event in $A\in \Pi_{t-}^{\Ealg}$, $A \cap \Omega = A \in \Pi_{t-}^{\Ealg}$. This shows that $\Pi_{t-}^{\Ealg}$ is a $\pi$-system. For each $t \in [0,\tau]$, define the collection
\begin{equation}
\Pi_{t-}^{\Falg} = \{A \cap B \colon A \in \Pi_{t-}^{\Ealg}, B \in \sigma(Z)\}.
\notag
\end{equation}
If $F_1,F_2 \in \Pi_{t-}^{\Falg}$, then $F_1 \cap F_2 = (A_1 \cap B_2) \cap (A_2 \cap B_2) = (A_1 \cap A_2) \cap (B_1 \cap B_2)$, and $A_1 \cap A_2 \in \Pi_{t-}^{\Ealg}$ since $\Pi_{t-}^{\Ealg}$ is a $\pi$-system, and $B_1 \cap B_2 \in \sigma(Z)$ since $\sigma(Z)$ is a {\sigalg}. This shows that $\Pi_{t-}^{\Falg}$ is also a $\pi$-system.

Now, since $\Pi_{t-}^{\Ealg}$ generates $\Ealg_{t-}$, $\Pi_{t-}^{\Falg}$ generates $\Falg_{t-}$. To see this, note that for all $A \in \Pi_{t-}^{\Ealg}$, $A = A \cap \Omega \in \Pi_{t-}^{\Falg}$ because $\Omega \in \sigma(Z)$. Similarly, for any  $B\in \sigma(Z)$, $B = \Omega \cap B  \in \Pi_{t-}^{\Falg}$ because $\Omega \in \Pi_{t-}^{\Ealg}$ (which explains why we include $\Omega$ in $\Pi_{t-}^{\Ealg}$). This shows that $\sigma(Z) \subset \Pi_{t-}^{\Falg}$ and $\Pi_{t-}^{\Ealg} \subset \Pi_{t-}^{\Falg}$, therefore $\sigma(Z) \subset \sigma( \Pi_{t-}^{\Falg})$ and $\Ealg_{t-} = \sigma(\Pi_{t-}^{\Ealg}) \subset \sigma(\Pi_{t-}^{\Falg})$, which entails that $\Falg_{t-} = \Ealg_{t-} \vee \sigma(Z) \subset \sigma(\Pi_{t-}^{\Falg})$. On the other hand, $\Pi_{t-}^{\Falg} \subset \Falg_{t-}$ because $\Pi_{t-}^{\Ealg} \subset \Ealg_{t-}$, and so $\sigma( \Pi_{t-}^{\Falg}) \subset \Falg_{t-}$. We conclude that $\Falg_{t-} = \sigma( \Pi_{t-}^{\Falg})$. 

The key to the modifications argument is the following: For any event $C$ in $\Pi_{t-}^{\Falg}$, meaning that $C = A \cap B$ for $A \in \Pi_{t-}^{\Ealg}$ and $B \in \sigma(Z)$,
\begin{equation}
\{T \geq t\} \cap C = \{T \geq t\} \cap B \in \Zalg_t, 
\label{eq::A4key}
\end{equation}
with $B \in \sigma(Z)$. This is because 
\begin{equation}
\{T \geq t\}\subset \{T \geq s\},\quad \text{for all $s \leq t$},
\notag
\end{equation}
and
\begin{equation}
\{T \geq t\}\subset \{N(s) = 0\},\quad \text{for all $s < t$},
\notag
\end{equation}
which means that there is no set in $\Pi_{t-}^{\Ealg}$ contained in $\{T \geq t\} = \{Y(t) = 1\}$, and therefore $\{T \geq t\}\cap A = \{T \geq t\}$ for all $A \in \Pi_{t-}^{\Ealg}$, which implies~\eqref{eq::A4key}. 

Define $\zeta_t$ by
\begin{equation}
\zeta_t = I_{T \geq t}\E\,(\xi_t \mid \Zalg_t) + I_{T < t}\E\,(\xi_t \mid \Falg_{t-}),
\label{eq::def_zeta}
\end{equation} 
by which we mean that $\zeta_t$ is a modification of the process on the right. Let $C$ be any event in $\Pi_{t-}^{\Falg}$. Then $\{T \geq t\} \cap C \in \Zalg_t$ by~\eqref{eq::A4key}. Since $\{T  < t\} \in \Falg_{t-}$ and $C \in \Pi_{t-}^{\Falg} \subset \Falg_{t-}$, it is the case that $\{T  < t\} \cap C \in \Falg_{t-}$. Therefore, by the definition of conditional expectation, used twice, we get that for any $C \in \Pi_{t-}^{\Falg}$
\begin{equation}
\begin{split}
\E\, I_C \zeta_t & = \E\, I_{\{T \geq t\}\cap C}\E\,(\xi_t \mid \Zalg_t) + \E\,I_{\{T < t\}\cap C}\E\,(\xi_t \mid \Falg_{t-})\\
& = \E\, I_{\{T \geq t\}\cap C}\xi_t  + \E\,I_{\{T < t\}\cap C}\xi_t
= \E\, I_C \xi_t.
\end{split}
\label{eq::lemmaA4_key}
\end{equation} 
This shows that for fixed $t$, the set functions $C \mapsto \E\, I_C \zeta_t$ and $C \mapsto \E\, I_C \xi_t$ agree on the $\pi$-system $\Pi_{t-}^{\Falg}$.
By assumption $\E\,\sup_{t\leq \tau} \xi_t < \infty$ (recall also that $\xi_t \geq 0$). Therefore $\E\,I_{\Omega} \xi_t \leq \E\,\xi_t \leq \E\, \sup_{t \leq \tau}\xi_t < \infty$, and since $\zeta_t \leq \E\,(\xi_t \mid \Zalg_t) + \E\,(\xi_t \mid \Falg_{t-})$ a.s., we have $\E\,I_{\Omega} \zeta_t = \E\,\zeta_t \leq 2 \E\,(\xi_t) \leq 2\,\E\,\sup_{t \leq \tau}\xi_t < \infty$. This shows that both $C \mapsto \E\, I_C \zeta_t$ and $C \mapsto \E\, I_C \xi_t$ are finite measures on $\Falg_{t-}$. By Dynkin{'}s lemma, two finite measures that agree on a $\pi$-system agree on the {\sigalg} generated by that $\pi$-system (see, e.g.,~\citet[Lemma~1.6(a), p.~19]{williams1991probability}). Since $\Pi_{t-}^{\Falg}$ generates $\Falg_{t-}$, this means that \eqref{eq::lemmaA4_key} holds for any $C \in \Falg_{t-}$, and we can conclude that,  
\begin{equation}
\zeta_t = \E\,(\xi_t \mid \Falg_{t-}),\;\text{a.s., for all $t$}, 
\label{eq::key1}
\end{equation}
which is to say that $\zeta_t$ is a modification of $\E\,(\xi_t \mid \Falg_{t-})$.

{\it Indistinguishable}: We now show that $\zeta_t$ and $\E\,(\xi_t \mid \Falg_{t-})$ are both a.s.~left-continuous for $t \in [0,\tau]$,  as this implies that they are indistinguishable (see, e.g., \citet[p.~3]{jacod2003limit} or \citet[Lemma~3.2.10, p.~79]{cohen2015stochastic}). 

By the triangle inequality, using linearity of conditional expectation and Jensen{'}s inequality \citep[Lemma~2.4.11, p.~58]{cohen2015stochastic}, we have
\begin{equation}
\abs{ \E\,(\xi_s \mid \Falg_{s-}) - \E\,(\xi_t \mid \Falg_{t-})		}\, 
\leq \E\,(\abs{\xi_s - \xi_t}\,\mid \Falg_{s-}) + 
\abs{\E\,(\xi_t \mid \Falg_{s-}) - \E\,(\xi_t \mid \Falg_{t-})},
\label{eq::leftcont_part1}
\end{equation}
almost surely. The process $\xi_t$ is a.s.~left-continuous. Given $\eps > 0$, we can find $\delta > 0$ so that $\abs{\xi_s - \xi_t} \,\leq \eps/2$ a.s.~whenever $s \in (t-\delta,t)$, and therefore, by monotonicity of conditional expectation, $\E\,(\abs{\xi_s - \xi_t}\,\mid \Falg_{s-}) < \eps/2$ a.s.~whenever $s \in (t-\delta,t)$. This show that the first term on the right in~\eqref{eq::leftcont_part1} is a.s.~left-continuous. For the second term, keep $t$ fixed and let $\eta_s$ be a modification of $\E\,(\xi_t \mid \Falg_{s})$, thus
\begin{equation}
\eta_s = \E\,(\xi_t \mid \Falg_{s}),\; \text{a.s.~for all $s$}.
\notag
\end{equation} 
Without loss of generality, we can take $\eta_s$ to be an uniformly integrable martingale \citep[Theorem~I.1.42, p.~11]{jacod2003limit} and c{\`a}dl{\`a}g (by right-continuity of $\Falg_t$, see, e.g.~Corollary~5.1.9 in \citet{cohen2015stochastic}). The process $\eta_{s-} = \lim_{u \uparrow s}\eta_u$ is left-continuous and predictable \citep[Proposition~2.5, p.~17]{jacod2003limit}, and from Theorem~6.2.18 in \citet{cohen2015stochastic}, we have that $\eta_{S-} = \E\, (\eta_S \mid \Falg_{S-})$ for any predictable stopping time $S$. But then
\begin{equation}
\eta_{S-} = \E\, (\eta_S \mid \Falg_{S-})
= \E\, (\E\,\{\xi_t \mid \Falg_{S}\} \mid \Falg_{S-})
= \E\, (\xi_t \mid \Falg_{S-}),\quad \text{a.s.},
\label{eq::leftcont_part2}
\end{equation}
for any predictable stopping time $S$. The processes $\eta_{s-}$ and $\E\,(\xi_t \mid \Falg_{s-})$ are both predictable. Combining this fact with~\eqref{eq::leftcont_part2} gives that $\eta_{s-}$ and $\E\, (\xi_t \mid \Falg_{s-})$ are indistinguishable \citep[Proposition~I.2.18, p.~20]{jacod2003limit}. Hence, $s \mapsto\E\, (\xi_t \mid \Falg_{s-})$ must be a.s.~left-continuous since $\eta_{s-}$ is a.s.~left-continuous. This means that given $\eps > 0$ we can find $\delta > 0$ such that $\abs{\E\,(\xi_t \mid \Falg_{s-}) - \E\,(\xi_t \mid \Falg_{t-})}\,< \eps/2$ a.s.~whenever $s \in(t - \delta,t)$. In summary, given $\eps > 0$ we can find $\delta >  0$ such that 
\begin{equation}
\abs{ \E\,(\xi_s \mid \Falg_{s-}) - \E\,(\xi_t \mid \Falg_{t-})}\, < \eps/2 + \eps/2 = \eps,\quad \text{a.s.},
\notag
\end{equation}
whenever $s \in (t - \delta,t)$, which is the desired result.

Next, we show that $\zeta_t$ is a.s.~left-continuous. Recall that $\zeta_t$ is a modification of $I_{T \geq t}\E\,(\xi_t \mid \Zalg_t) + I_{T < t}\E\,(\xi_t \mid \Falg_{t-})$. The process $I_{T < t}$ is left-continuous, and we have just shown that $\E\,(\xi_t \mid \Falg_{t-})$ is left-continuous. It therefore suffices to show that $I_{T \geq t}\E\,(\xi_t \mid \Zalg_t)$ is a.s.~left-continuous. Assume without loss of generality that $Z \geq 0$. By assumption, $Z$ has density $f_Z(z)$ with support $[z_1,z_2]$, we can therefore, without loss of generality, take $[z_1,z_2] = [0,1]$. For $n = 1,2,\ldots$ and $k = 0,1,\ldots,2^n$, let $G_{n,k} = Z^{-1}[k/2^n,(k+1)/2^n)$, define the sequence of simple functions given by
\begin{equation}
Z_n = \sum_{k=0}^{2^n}\frac{k}{2^n}I_{G_{n,k}},
\notag
\end{equation} 
and define $\Zalg_t^{n} = \sigma(Y(t))\vee \sigma(Z_n)$. Let $S$ be a predictable stopping time w.r.t.~$\Falg_t$. The {\sigalg} $\Zalg_t^{n}$ evaluated in $S$ is 
\begin{equation}
\Zalg_S^{n} = \sigma(Y(S))\vee \sigma(Z_n) = \{\{T \geq S\},\{T < S\},\emptyset,\Omega\}\vee \sigma(Z_n). 
\notag
\end{equation}
For each $n$ and $k$, the set $\{T \geq S\}\cap G_{n,k}$ has no nonempty proper subset in $\Zalg_S^{n}$ (this is because the $G_{n,0},\ldots,G_{n,2^n}$ are disjoint). From Lemma~\ref{lemma::lemmaA2} in the appendix we therefore get that 
\begin{equation}
I_{T \geq S}\E\,(\xi_S \mid \Zalg_S^{n}) 
= \sum_{k=0}^{2^n} I_{\{T \geq S\}\cap G_{n,k}}\E\,(\xi_S \mid \{T \geq S\}\cap G_{n,k}),\quad \text{a.s.},
\label{eq::use_LemmaA2}
\end{equation}
which entails that the processes (in $t$) $I_{T \geq t}\E\,(\xi_t \mid \Zalg_t)$ and $\sum_{k=0}^{2^n} I_{\{T \geq t\}\cap G_{n,k}}\E\,(\xi_t \mid \{T \geq t\}\cap G_{n,k})$ are indistinguishable \citep[Proposition~I.2.18, p.~20]{jacod2003limit}. We now show that the process $\sum_{k=0}^{2^n} I_{\{T \geq t\}\cap G_{n,k}}\E\,(\xi_t \mid \{T \geq t\}\cap G_{n,k})$ is a.s.~left-continuous. Almost all the sample paths $t \mapsto I_{\{T \geq t\}\cap G_{n,k}}$ are left-continuous, and, by the definition of conditional expectation,
\begin{equation}
\E\,(\xi_t \mid \{T \geq t\}\cap G_{n,k})
= \frac{\E\,\xi_t I_{T \geq t}I_{G_{n,k}}}{\pr(\{T \geq t\}\cap G_{n,k})}
= \pr(G_{n,k}) \frac{	\E\,\xi_t I_{T \geq t}I_{G_{n,k}}		}{	 \E\,I_{T \geq t}I_{G_{n,k}}	}.
\notag
\end{equation}
Since $\xi_t$ and $I_{T \geq t}$ are a.s.~left-continuous, $\lim_{s \uparrow t}\xi_s I_{\{T \geq s\}}I_{G_{n,k}} = \xi_t I_{\{T \geq t\}}I_{G_{n,k}}$ a.s. Clearly $\abs{\xi_t I_{\{T \geq t\}}I_{G_{n,k}}} \, \leq \xi_{t}$, and by assumption $\E\,\abs{\xi_t}\,< \infty$, so by the dominated convergence theorem, $\lim_{s \uparrow t}\E\,\xi_s I_{\{T \geq s\}}I_{G_{n,k}} = \E\,\xi_t I_{\{T \geq t\}}I_{G_{n,k}}$. Essentially the same argument shows that $\lim_{s \uparrow t}  \E\,I_{T \geq s}I_{G_{n,k}} =  \E\,I_{T \geq t}I_{G_{n,k}}$. We then have that $I_{T \geq t}\E\,(\xi_t \mid \Zalg_t)$ is indistinguishable from an a.s.~left-continuous process, hence it is itself a.s.~left-continuous.

The sequence of sets $\{G_{n,0},\ldots,G_{n,2^n}\}_{n \geq 1}$ is nested in the sense that for any $n\geq 1$, a set $G_{n,k}$ can be written as a union, $G_{n+1,j}\cap G_{n+1,j+1}$ say. This nestedness property entails that $\sigma(Z_n) \subset \sigma(Z_{n+1})$ for all $n$, and therefore $\Zalg_t^{n} \subset \Zalg_t^{n+1}$ for all $n$. This means that $(\Zalg_t^{n})_{n \geq	1}$ is a filtration, and, since for each $t$ the random variable $\xi_t$ is integrable, Theorem~7.23 in \citet[p.~132]{kallenberg2002foundations} gives that as $n \to \infty$,
\begin{equation}
\E\,(\xi_t \mid \Zalg_t^{n}) \to \E\,(\xi_t \mid \Zalg_t^{\infty}),\;\text{a.s.}
\notag
\end{equation} 
where $\Zalg_t^{\infty} = \vee_{n \geq 1}\Zalg_t^n$. The dyadic rationals are dense in $[0,1]$, which means that the sets $G_{n,k}$ for $k = 0,1,\ldots,2^{n},n = 1,2,\ldots$ generate $\sigma(Z)$. Since $\sigma(Z)$ is the smallest {\sigalg} that contains all the $G_{n,k}$, $\sigma(Z) \subset \vee_{n \geq 1}\sigma(Z_n)$. For each $n$, $\sigma(Z_n) \subset \sigma(Z)$ and $\sigma(Z_n) \subset \vee_{n \geq 1}\sigma(Z_n)$, but $\vee_{n \geq 1}\sigma(Z_n)$ is the smallest {\sigalg} that contains all the $\sigma(Z_n)$, thus $\vee_{n \geq 1}\sigma(Z_n) \subset \sigma(Z)$, consequently $\vee_{n \geq 1}\sigma(Z_n) = \sigma(Z)$, and we conclude that $\Zalg_t^{\infty} = \Zalg_t = \sigma(Y(t))\vee \sigma(Z)$.

We now have that $\{\E\,(\xi_t \mid \Zalg_t^{n})\}_{n \geq 1}$ is a sequence of a.s.~left-continuous processes, and that $\E\,(\xi_t \mid \Zalg_t^{n}) \to \E\,(\xi_t \mid \Zalg_t)$ a.s.~for each $t$ as $n \to \infty$. In order to conclude that $I_{T \geq t}\E\,(\xi_t \mid \Zalg_t)$ is a.s.~left-continuous, we must show that the convergence of $I_{T \geq t}\E\,(\xi_t \mid \Zalg_t^{n})$ to $I_{T \geq t}\E\,(\xi_t \mid \Zalg_t)$ is uniform in $t$. Write 
\begin{equation}
\mu_{n,k}(t) = \E\,( \xi_t \mid \{T \geq t\}\cap G_{n,k}), \;\text{for $k = 0,\ldots,2^n$ and $n \geq 1$},
\notag
\end{equation}
so that $I_{T \geq t}\E\,(\xi_t \mid \Zalg_t^n) = I_{T \geq t}\sum_{k=0}^{2^n} I_{G_{n,k}} \mu_{n,k}(t)$ a.s. Suppose that $G_{n+1,j}$ and $G_{n+1,j+1}$ are such that $G_{n,k} = G_{n+1,j}\cup G_{n+1,j+1}$. Then, on the event $G_{n+1,j}$, 
\begin{equation}
\begin{split}
& \abs{ I_{T \geq t}\E\,(\xi_t \mid \Zalg_t^{n}) - I_{T \geq t}\E\,(\xi_t \mid \Zalg_t^{n+1})}\\
& \qquad = I_{T \geq t} \abs{I_{G_{n,k}}\mu_{n,k}(t) - I_{G_{n+1,j}}\mu_{n+1,j}(t) - I_{G_{n+1,j+1}}\mu_{n+1,j+1}(t) } \\
& \qquad = I_{T \geq t} \abs{I_{G_{n+1,j}}\{\mu_{n,k}(t) - \mu_{n+1,j}(t)\} 
+ I_{G_{n+1,j+1}}\{\mu_{n,k}(t) - \mu_{n+1,j+1}(t)\}  }\\
& \qquad 
\leq  I_{T \geq t}I_{G_{n+1,j}} \abs{\mu_{n,k}(t) - \mu_{n+1,j}(t)}\,
+ I_{T \geq t}I_{G_{n+1,j+1}} \abs{\mu_{n,k}(t) - \mu_{n+1,j+1}(t)}\\
& \qquad  = I_{T \geq t}I_{G_{n+1,j}} \abs{\mu_{n,k}(t) - \mu_{n+1,j}(t)}.
\end{split}
\notag
\end{equation}
The difference $\abs{\mu_{n,k}(t) - \mu_{n+1,j}(t)}$ is
\begin{equation}
\begin{split}
& \abs{\mu_{n,k}(t) - \mu_{n+1,j}(t)}  
= \abs{ \E\,(\xi_t \mid \{T \geq t\}\cap G_{n,k}) - \E\,(\xi_t \mid \{T \geq t\}\cap G_{n+1,j})}\\
& \quad 
= \abs{ \frac{\E\,\xi_t I_{\{T \geq t\}}I_{G_{n,k}}}{\pr(\{T \geq t\}\cap G_{n,k})}		 	- \frac{\E\,\xi_t I_{\{T \geq t\}}I_{G_{n+1,j}}}{\pr(\{T \geq t\}\cap G_{n+1,j})}}\\
& \quad 
\leq \frac{\E\,\abs{\xi_t \big( I_{\{T \geq t\}\cap G_{n,k}}\pr(\{T \geq t\}\cap G_{n+1,j}) - I_{\{T \geq t\}\cap G_{n+1,j}}\pr(\{T \geq t\}\cap G_{n,k})\big)  }   }{\pr(\{T \geq t\}\cap G_{n,k})\pr(\{T \geq t\}\cap G_{n+1,j})}\\
& \quad 
\leq \frac{\norm{\xi_t}_2\, \norm{I_{\{T \geq t\}\cap G_{n,k}}\pr(\{T \geq t\}\cap G_{n+1,j}) - I_{\{T \geq t\}\cap G_{n+1,j}}\pr(\{T \geq t\}\cap G_{n,k})}_2     }{\pr(\{T \geq t\}\cap G_{n,k})\pr(\{T \geq t\}\cap G_{n+1,j})},
\end{split} 
\notag
\end{equation} 
using H{\"o}lder{'}s inequality, and writing $\norm{X}_{2}\,= (\E\,X^2)^{1/2}$. Here, by expanding the square on the last line above and taking expectations,
\begin{equation}
\begin{split}
& \norm{I_{\{T \geq t\}\cap G_{n,k}}\pr(\{T \geq t\}\cap G_{n+1,j}) - I_{\{T \geq t\}\cap G_{n+1,j}}\pr(\{T \geq t\}\cap G_{n,k})}_2^2  \\
& \quad = \pr(\{T \geq t\}\cap G_{n,k})
\pr(\{T \geq t\}\cap G_{n+1,j})\pr(\{T \geq t\}\cap G_{n+1,j+1}).
\end{split} 
\notag
\end{equation} 
Inserting the square root of this product in the expression for $\abs{\mu_{n,k}(t) - \mu_{n+1,j}(t)} $ above, and using that $\pr(\{T \geq t\}\cap G_{n+1,j+1}) \leq \pr(\{T \geq t\}\cap G_{n,k})$, 
we get 
\begin{equation}
I_{\{T \geq t \}\cap G_{n+1,j}}\abs{\mu_{n,k}(t) - \mu_{n+1,j}(t)} \, \leq 
\frac{\norm{\xi_t}_2 I_{\{T \geq t \}\cap G_{n+1,j}}}{
\pr(\{T \geq t\}\cap G_{n+1,j})^{1/2}}.
\label{eq::bounding_it1}
\end{equation}
By assumption, $z\mapsto \pr(T \geq \tau \mid Z = z)$ is bounded below on $[0,1]$, say $\pr( T \geq \tau \mid Z = z) \geq c > 0$, and, also by assumption, the density $f_Z(z)$ of $Z$ is bounded below on $[0,1]$. Let $f_{Z,\min} > 0$ be the lower bound of $f_Z(z)$, and denote $F_Z(z)$ the cumulative distribution function of $Z$. Using these bounds and the mean value theorem, we have
\begin{equation}
\begin{split}
\pr(\{T \geq t\}\cap G_{m,j}) & \geq \pr(\{T \geq \tau\}\cap G_{m,j})
 = \int_{j/2^m}^{(j+1)/2^m} \pr( T \geq \tau \mid Z = z) f_Z(z)\,\dd z\\
& \geq c \{F((j+1)/2^m) - F(j/2^m)\}
= cf_Z(d_j)\frac{1}{2^m}
\geq cf_{Z,\min}\, \frac{1}{2^m},
\end{split}
\notag
\end{equation}
where $d_j$ is some point in $G_{m,j}$. Inserting this lower bound in place of the denominator on the right hand side in~\eqref{eq::bounding_it1}, we get
\begin{equation}
\frac{\norm{\xi_t}_2 I_{\{T \geq t \}\cap G_{n+1,j}}}{
\pr(\{T \geq t\}\cap G_{n+1,j})^{1/2}} 
\leq \frac{\norm{\xi_t}_2 I_{\{T \geq t \}\cap G_{n+1,j}}}{
cf_{Z,\min} \,2^{-(n+1)/2}} 
\leq \frac{\norm{\xi_t}_2 I_{G_{n+1,j}}}{
cf_{Z,\min} \,2^{-(n+1)/2}}. 
\notag
\end{equation}
This gives that 
\begin{equation}
\sup_{t\in [0,\tau]}\,\abs{ I_{T \geq t}\E\,(\xi_t \mid \Zalg_t^{n}) - I_{T \geq t}\E\,(\xi_t \mid \Zalg_t^{n+1})}  
\,\leq \sup_{t\in [0,\tau]}\norm{\xi_t}_2\,\frac{ I_{G_{n+1,j}}}{
c f_{Z,\min} \,2^{-(n+1)/2}},\quad \text{a.s.},
\notag
\end{equation}
where, we recall, $\norm{\xi_t}_2 = (\E\,\xi_t^2)^{1/2}$. For any $\eps > 0$, by Markov{'}s inequality  
\begin{equation}
\begin{split}
& \pr\big( \sup_{t\in [0,\tau]}\,\abs{ I_{T \geq t}\E\,(\xi_t \mid \Zalg_t^{n}) - I_{T \geq t}\E\,(\xi_t \mid \Zalg_t^{n+1})}\, > \eps  \big) \\
& \qquad \leq \pr\big(\sup_{t\in [0,\tau]}\norm{\xi_t}_2\,\frac{ I_{G_{n+1,j}}}{
cf_{Z,\min} \,2^{-(n+1)/2}} \geq \eps\big) \\
& \qquad \leq \frac{1}{\eps} \sup_{t\in [0,\tau]}\norm{\xi_t}_2 
\frac{ \pr(G_{n+1,j})}{
cf_{Z,\min} \,2^{-(n+1)/2}} 
\leq \frac{1}{\eps}\frac{f_{Z,\max} \,2^{-(n+1)/2}}{cf_{Z,\min}},
\end{split}
\notag
\end{equation} 
where $f_{Z,\max} = \sup_{z\in [0,1]} f_Z(z) < \infty$, since $f_{Z}(z)$ is continuous. Thus, for any $\eps > 0$,
\begin{equation}
\sum_{n = 1}^{\infty}\pr\big( \sup_{t\in [0,\tau]}\,\abs{ I_{T \geq t}\E\,(\xi_t \mid \Zalg_t^{n}) - I_{T \geq t}\E\,(\xi_t \mid \Zalg_t^{n+1})}\, > \eps  \big)
\lesssim \sum_{n = 1}^{\infty} \frac{1}{2^{n/2}} = 1 + \sqrt{2},
\notag
\end{equation}
so by the Borel--Cantelli lemma, the Cauchy criterion (see e.g., \citet[Theorem~7.8, p.~147]{rudin1976principles}) holds with probability one. It then follows from Lemma~5.5.6 in \citet[p.~121]{cohen2015stochastic} that the limit $I_{T \geq t}\E\,(\xi_t \mid \Zalg_t)$ is a.s.~left-continuous. 
\end{proof}

\begin{remark} A theorem in statistics that also uses the martingale convergence theorem (see \citet[Theorem~7.23, p.~132]{kallenberg2002foundations}) that is used in the proof of Lemma~\ref{lemma::A4}, is Doob{'}s consistency theorem for posterior distributions in Bayesian nonparametric statistics. See \citet{miller2018detailed} for a detailed treatment.
\end{remark}

\subsection{Sequences involving the at-risk processes}\label{app::uniform_conv_lemma}
When working with versions of the Aalen additive hazards model one often encounters averages involving integrals of the form $\int_0^t n^{-1}\sum_{i=1}^n Q(Z_i)Y_i(s)\,\dd s$, for covariates $Z_i$, at-risk processes $Y_i(t) = I\{T_i \geq t\}$, and a function $Q$ that may be real-valued, a vector, or a matrix, and one needs to ascertain whether such integrals converge in probability uniformly in $t$. If the pairs $(T_i,Z_i),\,i = 1,\ldots,n$ are i.i.d., then the law of large numbers yields $n^{-1}\sum_{i=1}^n Q(Z_i)Y_i(t) \to_p \E\, Q(Z_1)Y_1(t)$ for each $t$. In order to conclude that 
\begin{equation}
\int_0^t n^{-1}\sum_{i=1}^n Q(Z_i)Y_i(s)\,\dd s \overset{p}\to\int_0^t  \E\, \{Q(Z_1)Y_1(s)\}\,\dd s, 
\notag
\end{equation}
we need something stronger than pointwise convergence in probability of the integrand on the left to the integrand on the right. A sufficient condition is that $\sup_{s\in[0,t]}\norm{n^{-1}\sum_{i=1}^n Q(Z_i)Y_i(s) - \E\, \{Q(Z_1)Y_1(s)\}}\to_p 0$. Other conditions than convergence in probability of the integrands uniformly in $t$ may be imposed, see for example~\citet[Prop.~II.5.2, p.~85]{ABGK93} and \citet[Lemma~A3]{HjortPollard93}. In the next lemma, used repeatedly in the present paper, we establish conditions under which the convergence in probability of $n^{-1}\sum_{i=1}^n Q(Z_i)Y_i(t)$ to $\E\,\{Q(Z_1)Y_1(t)\}$ is uniform in $t$. Not surprisingly, parts of the proofs is rather similar to the proof Glivenko--Cantelli theorem, see for example \citet[Theorem~19.1, p.~266]{vandervaart1998asymptotic} or \citet[p.~23]{ferguson1996largesample}. The full force of the almost sure convergence in the lemma below, rather than mere convergence in probability, is, in the present paper, only of importance as regards $\sup_{t \in [0,\tau]}\abs{J_{n,h}(t) - 1}\,\to 0$ a.s. See for example the proof of Corollary~\ref{corollary::Falg_Thetaconv} of the main text.

\begin{lemma}\label{lemma::uniform_convergence_lemma} Let $(T_i,\delta_i,Z_i,C_i),\,i = 1,\ldots,n$ be independent replicates of $(T,\delta,Z,C)$, for $T = \widetilde{T}\wedge C$ on $[0,\tau]$, with $\widetilde{T}$ the uncensored lifetime; indicators for non-censoring $\delta = I\{\widetilde{T} \leq C\}$, censoring times $C$, and a covariate $Z$. Let $Y_i(t) = I\{T_i \geq t\}$ for $i = 1,\ldots,n$ and $Y(t) I\{T \geq t\}$ be the at-risk indicators. Assume that $C$ stems from an absolutely continuous distribution with distribution function $H$ such that $H(\tau) < 1$; and that $\widetilde{T} \indep C \mid Z$. 
\begin{itemize}\itemsep-0.2em
\item[{\rm(a)}] Set $F(t \mid z) = \pr(\widetilde{T} \leq t \mid Z = z)$, suppose that $F(\tau\mid z) < 1$ for all $z \in (-\kappa_0,-\kappa_0)$, that $F(t \mid z)$ is continuous in $t \in [0,\tau]$ for all $z$, and that $F(t \mid z)$ continuous for $z \in (-\kappa_0,\kappa_0)$ for all $t \in [0,\tau]$.
\item[{\rm(b)}] For some $\kappa > 0$, let $q(u)$ be a nonnegative function that {\rm (i)} is zero outside of $[0,\kappa]$; {\rm (ii)} is continuous on $(0,\kappa)$, and {\rm (iii)} is such that $\int_{0}^{\kappa} q(u)^2\,\dd u < \infty$.
\item[{\rm(c)}] Let $\xi(t,z)$ be a non-negative function that is continuous in $t \in [0,\tau]$ for all $z$, and continuous for $z \in (-\kappa_0,\kappa_0)$ for all $t \in [0,\tau]$.
\item[{\rm(d)}] The covariate $Z$ has a density $f_Z(z)$ that is continuous on $(-\kappa_0,\kappa_0)$. 
\end{itemize}
Define $q_h(z) = q(z/h)/h$ for $h > 0$, and denote $G(t)=\lim_{h \to 0}\E\, q_h(Z)\xi(t,Z)Y(t)$. Existence of this limit is part of the claim of the lemma. Then 
\begin{equation}
\sup_{t\in [0,\tau]}\abs{n^{-1}\sum_{i=1}^n q_h(Z_i)\xi(t,Z_i)Y_i(t) - G(t)}\to 0,
\label{eq::uniform_convergence_lemma}
\end{equation}    
almost surely, as $n h \to \infty$ and $h \to 0$.
\end{lemma}
\begin{proof} Denote $G_{n,h}(t) = n^{-1}\sum_{i=1}^n q_h(Z_i)Y_i(t)$, $G_h(t) = \E\, G_{n,h}(t)$, and $G(t)= \lim_{h\to 0}G_h(t)$. Let $\bar{h}>0$ be a number such that $\bar{h}\kappa < \kappa_0$. If we restrict $f_Z(z)$ and $\xi(t,z)$ to $z \in [0,\bar{h}\kappa] \subset [0,\kappa_0)$, then since $[0,\bar{h}\kappa]$ and $[0,\tau]\times[0,\bar{h}\kappa]$ are compact sets, assumptions (c) and (d) entail that the restrictions of $f_Z(z)$ and $\xi(t,z)$ are bounded above. We denote these bounds by $f_{Z,\max}$ and $\xi_{\max}$, respectively.  

By the i.i.d.~assumption and the conditional independence $\widetilde{T} \indep C \mid Z$, the function $G_h(t)$ is
\begin{equation}
\begin{split}
G_h(t) & = \E\, G_{n,h}(t)  = \E\, q_h(Z)\xi(t,Z)Y(t) = \E\, [q_h(Z)\xi(t,Z) \E\, (Y(t) \mid Z)]\\
& = \{1 - H(t)\} \, \E\, [q_h(Z)\xi(t,Z)\{1 - F(t\mid Z) \} ].
\end{split}
\notag
\end{equation}
By a change of variable, 
\begin{equation}
\begin{split}
\E\, [q_h(Z)\xi(t,Z)\{1 - F(t\mid Z)  ] & = \frac{1}{h} \int_{0}^{h \kappa} q(z/h) \xi(t,z)\{1 - F(t \mid z)\}f_Z(z) \,\dd z\\
& = \int_{0}^{\kappa} q(u)\xi(t,hu) \{1 - F(t \mid hu)\}f_Z(hu) \,\dd u.
\end{split}
\notag
\end{equation}
When $h$ is so $h\kappa < \bar{h}\kappa$, then 
\begin{equation}
\abs{q(u)\xi(t,hu) \{1 - F(t \mid hu)\}f_Z(hu)} \,\leq \abs{q(u)}\,\xi_{\max} f_{Z,\max},
\notag
\end{equation} 
and by H{\"o}lder{'}s inequality 
\begin{equation}
\int_{0}^{\kappa} \abs{q(u)}\, \dd u \leq \kappa^{1/2} (\int_{0}^{\kappa} q(u)^2\,\dd u)^{1/2} < \infty, 
\label{eq::holder1}
\end{equation}
by Assumption (b)(iii). Thus, by dominated convergence,  
\begin{equation}
G_{h}(t) = G(t) + o(1), \;\text{for each $t$ as $h \to \infty$},
\notag
\end{equation}
with
\begin{equation}
G(t) = \{1 - H(t)\}\xi(t,0) \{1 - F(t \mid z_0)\}f_Z(0) \int_{0}^{\kappa} q(u) \dd u,
\notag
\end{equation}
Next, we show that the convergence $G_h(t) \to G(t)$ as $h \to 0$ is uniform over $[0,\tau]$. Let $h_n$ be a nonincreasing sequence of positive numbers such that $h_n \to 0$ as $n \to \infty$. Assume that $n < m$, so that $h_m/h_n < 1$, and that $h_m\kappa \leq \bar{h}\kappa < \kappa_0$, then
\begin{equation} 
\begin{split}
\abs{G_{h_{n}}(t) - G_{h_{m}}(t)} & \leq  \{1 - H(t)\} \,\E\,[ \abs{ q_{h_n}(Z) - q_{h_m}(Z)}\,\xi(t,Z)\{1 - F(t\mid Z) \} ]\\
& \leq  \xi_{\max} \int_{0}^{\infty} \abs{ h_n^{-1}q(z/h_n) - h_m^{-1}q(z/h_m)}\,f_Z(z)\,\dd z\\
& \leq  \xi_{\max} f_{Z,\max}  \int_{0}^{\infty} \abs{ (h_m/h_n) q(u h_m/h_n) - q(u)}\,\dd u.
\end{split} %
\label{eq::Cauchy_key}
\end{equation}
Since $h_m/h_n \to 1$ as $n,m \to \infty$ and $q(u)$ is continuous on $(0,\kappa)$, 
\begin{equation}
\abs{ (h_m/h_n) q(u h_m/h_n) - q(u)}\,\to 0,\; \text{as $n,m \to \infty$}.
\notag
\end{equation}
Moreover, $\abs{ (h_m/h_n) q(u h_m/h_n) - q(u)} \, \leq \abs{  q(u h_m/h_n)}\, + \abs{q(u)}$ and by a change of variable, $\int_{0}^{\infty} (\abs{ q(u h_m/h_n)}\, + \abs{q(u)})\,\dd u = \{(h_n/h_m) + 1\} \int_{0}^{\kappa}\abs{q(u)}\,\dd u < \infty$. Then $\int_{0}^{\infty} \abs{ (h_m/h_n) q(u h_m/h_n) - q(u)}\,\dd u \to 0$ as $n,m \to \infty$ by dominated convergence. Since the right hand side of~\eqref{eq::Cauchy_key} is independent of $t$ and converges to zero, this means that for any $\eps > 0$ we can find an $n_0 \geq 1$ such that $\abs{G_{h_{n}}(t) - G_{h_{m}}(t)} \, < \eps$ whenever $n,m \geq n_0$. This shows that the sequence $h \to G_h(t)$ satisfies the Cauchy criterion for uniform convergence (see, e.g.,~\citet[Theorem~7.8, p.~147]{rudin1976principles}), which entails that 
\begin{equation}
\sup_{t \in [0,\tau]} \abs{G_h(t) - G(t)}\, \to 0,\;\text{as $h \to 0$}.
\label{eq::uniformity_step1}
\end{equation}
We now turn to the convergence of $G_{n,h}(t)$ to $G_h(t)$ as $n \to \infty$. When $h\kappa \leq \bar{h}\kappa < \kappa_0$, and using the inequality in~\eqref{eq::holder1}
\begin{equation}
\begin{split}
\E\, \abs{q_h(Z)}\, & = h^{-1}\int_{h\kappa_a}^{h\kappa_b}\abs{q(z/h)}\,f_{Z}(z)\,\dd z\\ 
& = \int_{\kappa_a}^{\kappa_b}\abs{q(u)}\,f_{Z}(uh)\,\dd u \leq f_{Z,\max} \int_{\kappa_a}^{\kappa_b}\abs{q(u)}\,\dd u < \infty, 
\end{split}
\notag
\end{equation}
therefore $\E\, \abs{q_{h}(Z)\xi(t,Z)Y(t)}\, \leq \xi_{\max}f_{Z,\max}\,\E\, \abs{q_{h}(Z)}\, < \infty$, and it follows from the strong law of large numbers that, for $0 < h\leq \bar{h}$ fixed, $G_{n,h}(t) \to G_h(t)$ almost surely as $n \to \infty$ for every $t \in [0,\tau]$. It remains to show that this convergence is uniform in $t$. Write $S(t\mid z) = 1 - F(t\mid z)$. Assume that $h\kappa \leq \bar{h}\kappa < \kappa_0$, so that 
\begin{equation}
q_h(Z) =  q_h(Z) I_{[0,\bar{h}\kappa]}(Z).
\notag
\end{equation} 
Then
\begin{equation}
\begin{split}
& \abs{G_h(t) - G_h(s)} \,
 = \abs{ \{1 - H(t)\}\,\E\,q_h(Z) \xi(t,Z) S(t\mid Z)\\ 
& \qquad\qquad  - \{1 - H(s)\}\,\E\,q_h(Z) \xi(s,Z) S(s\mid Z) } \\
&\quad = 	\abs{ \{H(s) - H(t)\}\E\,q_h(Z) \xi(t,Z) S(t\mid Z)\\ 
& \qquad+ \{1 - H(s)\}\,\E\,\{ q_h(Z)I_{[0,\bar{h}\kappa]}(Z) [\xi(t,Z) S(t\mid Z) -  \xi(s,Z) S(s\mid Z)]\} }\\	
& \quad\leq 	\abs{ H(s) - H(t)}\, \E\,\abs{q_h(Z) \xi(t,Z) }\\ 
& \qquad\qquad +  \E\,\abs{ q_h(Z) }\,\abs{I_{[0,\bar{h}\kappa]}[ \xi(t,Z) S(t\mid Z) - \xi(s,Z) S(s\mid Z)]}\\
& \quad\leq 	\abs{ H(s) - H(t)}\, \E\,\abs{q_h(Z) \xi(t,Z) }\\ 
& \qquad\qquad +  \norm{ q_h(Z) }_2\,\norm{ I_{[0,\bar{h}\kappa]}\{\xi(t,Z) S(t\mid Z) - \xi(s,Z) S(s\mid Z)\}}_2\\
\end{split}
\notag
\end{equation} 
We now look at each of the two terms on the right of this expression. By the triangle inequality and since $S(t\mid Z)  \, < 1$, 
\begin{equation}
\begin{split}
& \norm{I_{[0,\bar{h}\kappa]}(Z)\{\xi(t,Z)S(t\mid Z) - \xi(s,Z)S(s\mid Z)\}}_2\\
& \quad = \norm{I_{[0,\bar{h}\kappa]}(Z)\{[\xi(t,Z) - \xi(s,Z)]S(t\mid Z) + \xi(s,Z)[S(t\mid Z) -S(s\mid Z)]\}}_2\\
& \quad \leq 
\norm{I_{[0,\bar{h}\kappa]}(Z)\{\xi(t,Z) - \xi(s,Z)\}}_2 + \norm{I_{[0,\bar{h}\kappa]}(Z)\xi(s,Z)\{F(s\mid Z) -F(t\mid Z)\}}_2.
\end{split}
\notag
\end{equation} 
By assumption (c) $\xi(t,z)$ is continuous on the compact set $[0,\tau]\times [0,\bar{h}\kappa]$, which implies that $\xi(s,z)$ is uniformly continuous on this set, and so, for any $\eps > 0$ we can find $\delta > 0$ such that  
\begin{equation}
\begin{split}
\norm{I_{[\kappa_a,\kappa_b]}(Z)\{\xi(t,Z) - \xi(s,Z)\}}_2^2
& = \E\,I_{[\kappa_a,\kappa_b]}(Z)[\xi(t,Z) - \xi(s,Z)]^2 < \frac{\eps^2}{9}, 
\end{split}
\notag
\end{equation} 
whenever $\abs{t - s} \, < \delta$. By the same argument, using that by assumption (a), $F(t\mid z)$ is uniformly continuous on $[0,\tau]\times [0,\bar{h}\kappa]$, we can find $\delta > 0$ such that $\norm{I_{[0,\bar{h}\kappa]}(Z)\{F(s\mid Z) -F(t\mid Z)\}}_2 < \eps/(3\xi_{\max})$ whenever $\abs{t - s} \, < \delta$, and so, since $\abs{\xi(s,z) }\, \leq \xi_{\max}$ for $(s,t) \in [0,\tau]\times [0,\bar{h}\kappa]$,
\begin{equation}
\begin{split}
& \norm{I_{[0,\bar{h}\kappa]}(Z)\xi(s,Z)\{F(s\mid Z) -F(t\mid Z)\}}_2^2\\
& \qquad \qquad \leq \xi_{\max}^2\,\norm{I_{[0,\bar{h}\kappa]}(Z)\{F(s\mid Z) -F(t\mid Z)\}}_2^2 < \frac{\eps^{2}}{9},
\end{split}
\notag
\end{equation}
whenever $\abs{t - s} \, < \delta$. Finally, when $h \kappa \leq \bar{h}\kappa$, $\E\,\abs{q_h(Z) \xi(t,Z) }\, \leq \xi_{\max} \E\, \abs{q_h(Z)}\, < \infty$ by the inequality in~\eqref{eq::holder1}. The distribution function $H(t)$ is continuous on $[0,\tau]$, meaning that for any $\eps > 0$ we can find $\delta > 0$ such that $\abs{H(s) - H(t)}\, < \eps/(3 \xi_{\max} \E\, \abs{q_h(Z)})$ whenever $\abs{t - s} \,< \delta$. In summary, provided $h \kappa \leq \bar{h}\kappa < \kappa_0$, for any $\eps > 0$ we can find $\delta > 0$ such that 
\begin{equation}
\abs{G_h(t) - G_h(s)} \, < \eps/3 + \eps/3 + \eps/3 = \eps, 
\notag
\end{equation}
whenever $\abs{t - s}\, < \delta$. This shows that $G_h(t)$ is a continuous function, and since $G_h(t)$ is continuous on the closed and bounded interval $[0,\tau]$, it is uniformly continuous on this interval. The uniform continuity of $G_h(t)$ entails that given $\eps > 0$, we can find a grid $0 = t_0 < t_1 < \cdots < t_k = \tau$, so that $\abs{G_h(t_i) - G_h(t_{i-1})}\, < \eps$ for all $i = 1,\ldots,k$. Using that $G_{n,h}(t)$ and $G_h(t)$ are both nonincreasing in $t$ (since $q(u) \geq 0$ by assumption (b)), this entails that for $t \in [t_{i-1},t_i)$, 
\begin{equation}
\begin{split}
G_{n,h}(t) - G_h(t) & = G_{n,h}(t) - G_h(t_{i-1}) + G_h(t_{i-1}) - G_h(t)\\ 
& < G_{n,h}(t) - G_h(t_{i-1}) + \eps < G_{n,h}(t_{i-1}) - G_h(t_{i-1}) + \eps, 
\end{split}
\notag
\end{equation}
and, similarly, 
\begin{equation}
\begin{split}
G_{n,h}(t) - G_h(t) & =  G_{n,h}(t) - G_h(t_i) + G_h(t_i) - G_h(t)\\ 
& > G_{n,h}(t) - G_h(t_i) - \eps \geq G_{n,h}(t_i) - G_h(t_i) - \eps,
\end{split}
\notag
\end{equation}
Thus, for $t \in [t_{i-1},t_i)$ and our chosen grid, 
\begin{equation}
G_{n,h}(t_{i}) - G_h(t_{i}) - \eps \leq G_{n,h}(t) - G_h(t) \leq G_{n,h}(t_{i-1}) - G_h(t_{i-1}) + \eps.
\notag
\end{equation}  
Because this is true for every $i = 1,\ldots,k$, we must have that 
\begin{equation}
\sup_{t \in [0,\tau]}\abs{G_{n,h}(t) - G_h(t)}\, \leq \max_{0 \leq i \leq k}
\abs{G_{n,h}(t_{i}) - G_h(t_{i})}\, + \,\eps.
\notag
\end{equation}  
Since the convergence $G_{n,h}(t) \to G_h(t)$ a.s., is uniform over $\{t_0,t_1,\ldots,t_k\}$, the right hand side of the above display tends almost surely to $\eps$, and so $\sup_{t \in [0,\tau]}\abs{G_{n,h}(t) - G_h(t)} \to \eps$ almost surely as $n \to \infty$. Therefore, by the triangle inequality 
\begin{equation}
\sup_{t \in [0,\tau]}\abs{G_{n,h}(t) - G(t)}\, \leq \sup_{t \in [0,\tau]}\abs{G_{n,h}(t) - G_h(t)}\, + \sup_{t \in [0,\tau]}\abs{G_{h}(t) - G(t)} 
\,\to \eps\,
\label{eq::uniformity_conclusion}
\end{equation}
almost surely, as $n \to \infty$ and $h \to 0$, using that the second term on the right converges to zero as $h \to 0$ by~\eqref{eq::uniformity_step1}.
Since \eqref{eq::uniformity_conclusion} is true for every $\eps >0$, the claim in \eqref{eq::uniform_convergence_lemma} follows.    
\end{proof}  

\begin{remark} Lemma~\ref{lemma::uniform_convergence_lemma} obviously holds for $\xi(t,z) = 1$. If we assume that $F(t\mid z)$ and $\xi(t,z)$ are uniformly continuous on $[0,\tau]\times S$, where $S$ is the support of $Z$, then the lemma also holds when $q_h(u) = q(u)$ and $q$ is a continuous function on $\real$. 
\end{remark}



\subsection{Sequences of random matrices}
Versions of the results in this section can be found in \citet{mckeague1988asymptotic} or \citet{mckeague1988counting}. They are included here for completeness.
\begin{lemma}\label{lemma::A_n_pr_pd} Let $A_n(t),\, t \in [0,\tau]$ be a sequence of symmetric matrix valued functions converging in probability (resp.~almost surely) to a matrix valued function $A(t)$ that is positive definite for all $t\in [0,\tau]$. Let 
\begin{equation}
B_n = \{\text{$A_n(t)$ is positive definite for all $t\in[0,\tau]$}\}. 
\notag
\end{equation}
Then $\pr(B_n) \to 1$ as $n \to \infty$ (resp.~$\exists\,n_0 \geq 0$ such that $\pr(B_n, \,\forall \,n \geq n_0) = 1$). 
\end{lemma} 
\begin{proof} We only prove the convergence in probability part. Write $A_n(t) = \{a_{i,j}^n(t)\}$ and $A(t) = \{a_{i,j}(t)\}$. For each $i$ and $j$, $a_{i,j}^n(t) = a_{i,j}(t) + r_{i,j}^n(t)$, where $\sup_t\abs{r_{i,j}(t)}\to_p 0$ as $n \to \infty$. Let $z$ be an arbitrary nonzero real column vector, then $z^{\tr}A_n(t) z = z^{\tr}A(t)z + \sum_{i}\sum_{j} r_{i,j}^n(t) z_i z_j$, where $z^{\tr}A(t)z > 0$ for all $t$, by assumption. Since $z^{\tr}A(t) z > 0$ there is some $\eps > 0$ such that $z^{\tr}A(t) z > \eps > 0$. For such an $\eps > 0$, we can for any $\delta >0$ find an $n_0$ such that 
\begin{equation}
\pr( \sup_{t}\,\abs{\sum_{i}\sum_{j} r_{i,j}^n(t) z_i z_j}\,<  \eps ) > 1 - \delta, 
\notag
\end{equation}
for all $n \geq n_0$, but then $\pr\{ \text{$z^{\tr}A_n(t) z > 0$ for all $t \in [0,\tau]$}\} > 1 - \delta$ for all $n \geq n_0$.   
\end{proof}

\begin{lemma}\label{lemma::invertible_to_limit} Let $A_n(t)$ and $A(t)$ be as in Lemma~\ref{lemma::A_n_pr_pd}. Let $\widetilde{A}_{n}(t) = A_n(t)$ if $A_n(t)$ is positive definite, and $\widetilde{A}_{n}(t)$ be equal to some invertible matrix otherwise. Then $\sup_{t \in [0,\tau]}\abs{\widetilde{A}_{n}^{-1}(t) - A^{-1}(t)}\, = o(1)$ almost surely.
\end{lemma}
\begin{proof} The following matrix inequality can be found in~\citet[p.~231]{mckeague1988counting}: If $C$ and $D$ are nonsingular matrices such that $\norm{C - D} < \norm{D}$, then 
\begin{equation}
\norm{C^{-1} - D^{-1}} \, \leq \frac{\norm{D^{-1}}^2\norm{C - D} }{1 - \norm{D^{-1}}\norm{C - D}}.
\label{eq::invertible_matrix_ineq}
\end{equation}
Since $\sup_{t\in [0,\tau]}\abs{\widetilde{A}_{n}(t) - A(t)} \to 0$ almost surely, there is an $n_0 \geq 1$ such that the inequality above, with $C =  \widetilde{A}_{n}(t)$ and $D = A(t)$ for all $t \in [0,\tau]$ holds almost surely for all $n \geq n_0$. But the right hand side converges in almost surely to zero. 
\end{proof}
Here is a useful lemma that does not necessarily involve neither sequences nor randomness. 

\begin{lemma}\label{lemma::simple_ineq} Let $x = (x_1,\ldots,x_p)^{\tr}$ be a $p$-dimensional column vector, $A = \{a_{i,j}\}$ a $p\times p$ matrix, and $e_{\nu}$ the $p$-dimensional column vector with the $\nu$th entry equal to $1$, and all other equal to zero. Then $e_{\nu}^{\tr} Axx^{\tr}A^{\tr}e_{\nu}  \leq \norm{A}^2 \norm{x}^2$.
\end{lemma}
\begin{proof} First, $e_{\nu}^{\tr} Axx^{\tr}A^{\tr}e_{\nu} =  (e_{\nu}^{\tr} Ax)^2$, and 
\begin{equation}
\abs{e_{\nu}^{\tr} Ax}^2\, \leq \max_{1 \leq \nu \leq p} 
\abs{\sum_{j=1}^p a_{\nu,j}x_j }^2\,
\leq \sum_{\nu=1}^p \,\abs{\sum_{j=1}^p a_{\nu,j}x_j }^2\,
= \norm{Ax}^2.
\notag
\end{equation}
Provided that $\norm{x}$ is positive, $\norm{Ax} = (\norm{Ax}/\norm{x})\norm{x} \leq \sup_{\norm{y} \neq 0} (\norm{Ay}/\norm{y})\norm{x} = \norm{A}\norm{x}$.
\end{proof}

\subsection{Frequently used limits}
The following lemma establishes the convergence in probability of the sequences $\Gamma_{g,p,n}(t,h)$, $J_{n,h}(t)$, $\vartheta_{g,p,q,n}(t,h)(t,h)$, $\bar{\Psi}_{g,p,n}(t,h)$, and $\Psi_{g,p,n}(t,h)$ for $g = 0,1$.
\begin{equation}
\vartheta_{g,p,q,n}(t,h) = \frac{1}{n}\sum_{i=1}^n I_{X_i = g}K_h(Z_i)r_p(Z_i/h)(Z_i/h)^{q}Y_i^g(t),\; \text{for $g = 0,1$},
\label{eq::myXi}
\end{equation}

\begin{lemma}\label{lemma::prob_limits} Assumptions~\ref{assumption::indep}--\ref{assumption::kernel} hold. As $nh \to \infty$ and $h \to 0$, the following holds almost surely
\begin{itemize}
\item[{\rm (i)}] $\sup_{t\in[0,\tau]}\norm{\Gamma_{1,p,n}(t,h) - y_1(t,z_0) f_Z(z_0)\Gamma_{p}}\, = o(1)$;
\item[{\rm (ii)}] $\sup_{t\in[0,\tau]}\norm{\Gamma_{0,p,n}(t,h) - y_0(t,z_0) f_Z(z_0)H_{p}(-1)\Gamma_{p}H_{p}(-1)}\, = o(1)$;
\item[{\rm (iii)}] $\sup_{t\in[0,\tau]}\abs{J_{n,h} - 1} = o(1)$;
\item[{\rm (iv)}] $\sup_{t\in[0,\tau]}\norm{\vartheta_{1,p,q,n}(t,h) - y_{1}(t,z_0)f_Z(z_0) \vartheta_{p,q}} = o(1)$;
\item[{\rm (v)}] $\sup_{t\in[0,\tau]}\norm{\vartheta_{0,p,q,n}(t,h) - y_{0}(t,z_0)f_Z(z_0) (-1)^qH_p(-1)\vartheta_{p,q}} = o(1)$;
\item[{\rm (vi)}] $\sup_{t\in[0,\tau]}\norm{J_{n,h}(t)\Gamma_{1,p,n}(t,h)^{-1}\vartheta_{1,p,q,n}(t,h) - \Gamma_p^{-1}\vartheta_{p,q}} = o(1)$;
\item[{\rm (vii)}] $\sup_{t\in[0,\tau]}\norm{J_{n,h}(t)\Gamma_{0,p,n}(t,h)^{-1}\vartheta_{0,p,q,n}(t,h) - \Gamma_p^{-1}\vartheta_{p,q}} = o(1)$.
\item[{\rm (viii)}] $\sup_{t\in[0,\tau]}\norm{ h \bar{\Psi}_{1,p,n}(t,h) - y_1(t,z_0)\bar{\alpha}_1(t,z_0)f_Z(z_0)\Psi_p} = o(1)$;
\item[{\rm (ix)}] $\sup_{t\in[0,\tau]}\norm{ h \bar{\Psi}_{0,p,n}(t,h) - y_0(t,z_0)\bar{\alpha}_0(t,z_0)f_Z(z_0)H_p(-1)\Psi_p H_p(-1)} = o(1)$;
\item[{\rm (x)}] $\sup_{t\in[0,\tau]}\norm{ h \Psi_{1,p,n}(t,h) - y_1(t,z_0)\bar{\alpha}_1(t,z_0)f_Z(z_0)\Psi_p} = o(1)$;
\item[{\rm (xi)}] $\sup_{t\in[0,\tau]}\norm{ h \Psi_{0,p,n}(t,h) - y_0(t,z_0)\bar{\alpha}_0(t,z_0)f_Z(z_0)H_p(-1)\Psi_p H_p(-1)} = o(1)$;
\end{itemize}
Set $m_n = \min(h,b)$, and assume that $m_n/\max(h,b) \to \rho \geq 0$ as $m_n \to 0$. Then, 
\begin{itemize}
\item[{\rm (xii)}] $\sup_{t\in[0,\tau]}\norm{ \frac{hb}{m_n} \bar{\Psi}_{1,p,q,n}(t,h,b) - y_1(t,z_0)\bar{\alpha}_1(t,z_0)f_Z(z_0)\Psi_{p,q}(\rho) }$, 
\item[{\rm (xiii)}] $\sup_{t\in[0,\tau]}\norm{\frac{hb}{m_n} \bar{\Psi}_{0,p,q,n}(t,h,b) - y_0(t,z_0)\bar{\alpha}_0(t,z_0)f_Z(z_0)H_p(-1)\Psi_{p,q}(\rho) H_q(-1)}$,
\end{itemize}
both tend almost surely to zero as $m_n \to 0$ and $nm_n \to \infty$. 
\end{lemma}
\begin{proof} Throughout the proof we assume that the cut-off is $z_0 = 0$, but write $z_0$ instead of $0$ when this adds to the clarity of the exposition. First, (i): All the elements of the matrix 
\begin{equation}
\Gamma_{1,p,n}(t,h) = n^{-1}\sum_{i=1}^n I_{Z_i \geq z_0}K_h(Z_i) r_{p}(Z_i/h)r_{p}(Z_i/h)^{\tr} Y_i^1(t),
\notag
\end{equation}
are of the form $n^{-1}\sum_{i=1}^n q_{h,\nu}(Z_i) Y_i^{1}(t)$,  with $q_{h,\nu}(z) = q_{\nu}(z/h)/h$ and $q_{\nu}(u) = I_{u \geq z_0}K(u) u^{\nu}$ for $\nu = 0,\ldots,2p$. Due to the definition of the kernel function $K$ (see Assumption~\ref{assumption::kernel} of the main text), the functions $q_{\nu}(u)$ satisfy Conditions (i)--(iii) of Lemma~\ref{lemma::uniform_convergence_lemma}. This means that, as $nh \to \infty$ and $h \to 0$, every element of $\Gamma_{1,p,n}(t,h)$ converges almost surely to the corresponding element of $\lim_{h \to 0}\E\, \Gamma_{1,p,n}(t,h)$, uniformly in $t$. Since every element converges, the matrix converges, and we therefore only need to find this limit. By a change of variable and using the i.i.d.~assumption,  
\begin{equation}
\begin{split}
\E\, \Gamma_{1,p,n}(t,h) & = h^{-1}\,\E\, \{I_{Z \geq z_0}K(Z/h) r_{p}(Z/u)r_{p}(Z/h)^{\tr} Y^1(t) \}\\
& = h^{-1}\,\E\, \{I_{Z \geq z_0} K(Z/h) r_{p}(Z/u)r_{p}(Z/h)^{\tr} \,\E\,[Y^1(t) \mid Z] \}\\
& = h^{-1}\,\int_0^{\infty} K(z/h) r_{p}(z/u)r_{p}(z/h)^{\tr} y_1(t,z) f_Z(z)\,\dd z\\
& = \int_0^{1} K(u) r_{p}(u)r_{p}(u)^{\tr} y_1(t,hu) f_Z(hu)\,\dd u.
\end{split}
\notag
\end{equation}
Since $y_1(t,z)$ is continuous for $z \in (-\kappa_0,\kappa_0)$ for all $t \in [0,\tau]$, and $f_Z(z)$ is continuous for $z \in (-\kappa_0,\kappa_0)$, $y_1(t,hu) f_Z(hu) \to y_1(t,z_0) f_Z(z_0)$ as $h \to 0$ (recall that $z_0 = 0$). Every element of $\int_0^{1} K(u) r_{p}(u)r_{p}(u)^{\tr} y_1(t,hu) f_Z(hu)\,\dd u$ is of the form $\int_0^{1} K(u) u^{\nu} y_1(t,hu) f_Z(hu)\,\dd u$ for $\nu = 0,\ldots,2p$, and when $h \kappa < \kappa_0$,  
\begin{equation}
\abs{\int_0^{1}K(u) u^{\nu} y_1(t,hu) f_Z(hu)\,\dd u}\, \leq f_{Z,\max} \int_0^1K(u) \,\dd u < \infty. 
\notag
\end{equation}
Therefore, by bounded convergence, 
\begin{equation}
\int_0^{1} K(u) r_{p}(u)r_{p}(u)^{\tr} y_1(t,hu) f_Z(hu)\,\dd u = y_1(t,z_0)f_Z(z_0) \Gamma_{p} + o(1),
\notag
\end{equation}
as $h \to 0$. The proof of (ii) is similar and also follows from Lemma~\ref{lemma::uniform_convergence_lemma}. There are two differences. First, all the elements of $\Gamma_{0,p,n}(t,h)$ are of the form $n^{-1} \sum_{i=1}^n q_{h,\nu}(Z_i)Y_i^0$, with $q_{h,\nu}(z) = q_{\nu}(z/h)/h$ where $q_{\nu}(u) = I_{u < z_0} K(u)u^{\nu}$ is negative when $\nu$ is odd. We can replace $q_{\nu}(u)$ with $q_{\nu}^{+}(u) = (-1)^{\nu}q_{\nu}(u)$, and Lemma~\ref{lemma::uniform_convergence_lemma} applies. The second slight difference is the limit of $\E\, \Gamma_{0,p,n}(t,h)$ as $h \to 0$. Assuming that $h \kappa < \kappa_0$, it is
\begin{equation}
\begin{split}
\E\, \Gamma_{0,p,n}(t,h) & = h^{-1}\,\E\, \{I_{Z < z_0}K(Z/h) r_{p}(Z/u)r_{p}(Z/h)^{\tr} Y^0(t) \}\\
& = h^{-1}\,\int_{-1}^{0} K(z/h) r_{p}(z/u)r_{p}(z/h)^{\tr} y_0(t,z) f_Z(z)\,\dd z\\
& = \int_0^{1} K(-u) r_{p}(-u)r_{p}(-u)^{\tr} y_0(t,-hu) f_Z(-hu)\,\dd u\\
& = H_p(-1)\int_0^{1} K(u) r_{p}(u)r_{p}(u)^{\tr} y_0(t,-hu) f_Z(-hu)\,\dd u\,H_p(-1)\\
& = y_0(t,z_0) f_Z(z_0)H_p(-1) \Gamma_{p}\,H_p(-1) + o(1),
\end{split}
\notag
\end{equation}
as $h \to 0$, because $K(-u) = K(u)$ and $r_p(-u) = H_p(-1)r_p(u)$. 

Next, (iii):  By Lemma~\ref{lemma::A_n_pr_pd}, the results in (i) and (ii) combined with the fact that $\Gamma_{p}$ is positive definite and $y(t,z_0)f_Z(z_0)$ is bounded below for all $t \in [0,\tau]$, entails that there are $n_0 \geq 1$ and $h_0 > 0$ such that $\Gamma_{1,p,n}(t,h)$ and $\Gamma_{0,p,n}(t,h)$ are positive definite for all $t \in [0,\tau]$ for all $n \geq n_0$ and $h \leq h_0$ has probability one, which is (iii).

We now turn to (iv): The elements of the vector 
\begin{equation} 
\vartheta_{1,p,q,n}(t,h) = \frac{1}{n}\sum_{i=1}^n I_{X_i = 1}K_h(Z_i)r_p(Z_i/h)(Z_i/h)^{q}Y_i^1(t),
\notag
\end{equation} 
are of the form $n^{-1}\sum_{i=1}^n q_h(Z_i/h) Y_i^1(t)$ with $q_h(Z_i/h) = h^{-1}q(u/h)$ and $q(u) = I\{u \geq z_0\}K(u)u^{\nu + q}$ for $\nu = 0,\ldots,p$. This function satisfies Conditions (i)--(iii) of Lemma~\ref{lemma::uniform_convergence_lemma}. Hence, uniform convergence of $\vartheta_{1,p,q,n}(t,h)$ to $\lim_{h \to 0}\E\,\vartheta_{1,p,q,n}(t,h)$ is ensured by Lemma~\ref{lemma::uniform_convergence_lemma}, and we only need find this limit. By a change of variable and using the i.i.d.~assumption
\begin{equation} 
\begin{split}
\E\,\vartheta_{1,p,q,n}(t,h) & = h^{-1}\,\E\, I_{Z \geq z_0}K(Z/h)r_p(Z/h)(Z/h)^{q}Y^1(t)\\
 & = h^{-1}\,\E\, I_{Z \geq z_0}K(Z/h)r_p(Z/h)(Z/h)^{q}\E\,\{Y^1(t) \mid Z\}\\
 & = h^{-1}\,\int_0^1 K(z/h)r_p(z/h)(z/h)^{q}y_1(t,z)f_Z(z)\,\dd z\\
  & = \int_0^1 K(u)r_p(u)u^{q}y_1(t,hu)f_Z(hu)\,\dd u\\
  & = y_1(t,z_0)f_Z(z_0)\int_0^1 K(u)r_p(u)u^{q}\,\dd u + o(1)\\ 
& = y_1(t,z_0)f_Z(z_0) \vartheta_{p,q} + o(1),
\end{split}
\notag
\end{equation} 
where the second to last equality follows from arguments identical to those used for the convergence of the corresponding limit in the proof of (i). The claim in (v) is similar, noting that 
\begin{equation} 
\begin{split}
\E\,\vartheta_{0,p,q,n}(t,h) & = h^{-1}\,\E\, I_{Z < z_0}K(Z/h)r_p(Z/h)(Z/h)^{q}Y^1(t)\\
& = h^{-1}\,\int_{-1}^0 K(z/h)r_p(z/h)(z/h)^{q}y_1(t,z)f_Z(z)\,\dd z\\
  & = \int_{0}^1 K(-u)r_p(-u)(-u)^{q}y_1(t,-hu)f_Z(-hu)\,\dd u\\
    & = (-1)^qH_{p}(-1)\int_{0}^1 K(u)r_p(u)u^{q}y_1(t,-hu)f_Z(-hu)\,\dd u,
\end{split}
\notag
\end{equation} 
once more because $K(-u) = K(u)$, $r_p(-u) = H_p(-1)r_p(u)$, and $(-u)^q = (-1)^qu^q$. Thus, $\E\,\vartheta_{0,p,q,n}(t,h) = (-1)^qH_{p}(-1)y_1(t,z_0)f_Z(z_0)\vartheta_{p,q} + o(1)$ as $h \to 0$. 

The claim in (vi) is a consequence of (i), (iii), and (iv). First, since $\Gamma_{p}$ is positive definite and $y_1(t,z_0)f_Z(z_0)$ is bounded below, the limit $y_1(t,z_0)f_Z(z_0)\Gamma_p$ in (i) is positive definite. From Lemma~\ref{lemma::invertible_to_limit}, we then get that 
\begin{equation}
\sup_{t\in [0,\tau]}\norm{ J_{n,h}(t) \Gamma_{1,p,n}(t,h)^{-1} - [y_1(t,z_0)f_Z(z_0)]^{-1}\Gamma_p^{-1}}\, = o(1), 
\label{eq::inverse_limit}
\end{equation}
almost surely, as $nh \to \infty$ and $h \to 0$. Second, by the triangle inequality, using the properties of a norm, and the inequality $\norm{Ax} \leq \norm{A}\norm{x}$ for a matrix $A$ and a vector $x$,  
\begin{equation} 
\begin{split}
& \norm{J_{n,h}(t)\Gamma_{1,p,n}(t,h)^{-1}\vartheta_{1,p,q,n}(t,h) - \Gamma_p^{-1}\vartheta_{p,q}}\\
& \quad\leq \norm{J_{n,h}(t)\{\Gamma_{1,p,n}(t,h)^{-1}\vartheta_{1,p,q,n}(t,h) - \Gamma_p^{-1}\vartheta_{p,q}\}}
+ \abs{J_{n,h}(t) - 1}\,\norm{\Gamma_p^{-1}\vartheta_{p,q}}\\
& \quad \leq \norm{J_{n,h}(t)\{\Gamma_{1,p,n}(t,h)^{-1} - [y_1(t,z_0)f_Z(z_0)]^{-1}\Gamma_p^{-1}\}\vartheta_{1,p,q,n}(t,h)}\\
& \quad \qquad + \abs{J_{n,h}(t)[y_1(t,z_0)f_Z(z_0)]^{-1}}\,\norm{\Gamma_p^{-1}\{\vartheta_{1,p,q,n}(t,h) - y_1(t,z_0)f_Z(z_0)\vartheta_{p,q}\}}\\
& \qquad \qquad\quad + \abs{J_{n,h}(t) - 1}\,\norm{\Gamma_p^{-1}\vartheta_{p,q}}\\
& \leq \norm{J_{n,h}(t)\{\Gamma_{1,p,n}(t,h)^{-1} - [y_1(t,z_0)f_Z(z_0)]^{-1}\Gamma_p^{-1}\}}\,\norm{\vartheta_{1,p,q,n}(t,h)}\\
& \quad + \abs{J_{n,h}(t)[y_1(t,z_0)f_Z(z_0)]^{-1}}\,\norm{\{\vartheta_{1,p,q,n}(t,h) - y_1(t,z_0)f_Z(z_0)\vartheta_{p,q}\}}\,\norm{\Gamma_p^{-1}}\\
& \qquad \qquad + \abs{J_{n,h}(t) - 1}\,\norm{\Gamma_p^{-1}\vartheta_{p,q}}.
\end{split}
\notag
\end{equation} 
For the first term on the right, by~\eqref{eq::inverse_limit} and (iv) 
\begin{equation}
\norm{J_{n,h}(t)\{\Gamma_{1,p,n}(t,h)^{-1} - [y_1(t,z_0)f_Z(z_0)]^{-1}\Gamma_p^{-1}\}}\,\norm{\vartheta_{1,p,q,n}(t,h)} = o_p(1)O_p(1) = o_p(1),
\notag
\end{equation}
uniformly over $[0,\tau]$. The second term: The inverse of a positive definite matrix is also positive definite, and the matrix norm of $\Gamma_p^{-1}$ equals the largest eigenvalue of $\Gamma_p^{-1}$, therefore $\norm{\Gamma_p^{-1}}$ is finite; since $y_1(t,z_0)f_Z(z_0)$ is bounded below, as $nh \to \infty$ and $h \to 0$, $\sup_{t \in [0,\tau]}\abs{J_{n,h}(t)[y_1(t,z_0)f_Z(z_0)]^{-1}}\, = O_p(1)$ by (iii). Therefore, by (iv) $\abs{J_{n,h}(t)[y_1(t,z_0)f_Z(z_0)]^{-1}}\,\norm{\{\vartheta_{1,p,q,n}(t,h) - y_1(t,z_0)f_Z(z_0)\vartheta_{p,q}\}}\,\norm{\Gamma_p^{-1}} = O_p(1)o_p(1)O(1) = o_p(1)$ uniformly in $t \in [0,\tau]$ as $nh \to \infty$ and $h \to 0$. The last term: $\norm{\Gamma_p^{-1}\vartheta_{p,q}} \leq \norm{\Gamma_p^{-1}}\,\norm{\vartheta_{p,q}} = O(1)$, and by (iii) $\sup_{t\in [0,\tau]} \abs{J_{n,h}(t) - 1}\, = o_p(1)$ as $nh \to \infty$ and $h \to 0$. This proves (vi).

The proof of (vii) is almost identical to the proof of (vi), using the limits from (ii) and (v).

Now, (viii): All the elements of the matrix 
\begin{equation}
h\bar{\Psi}_{1,p,n}(t,h) = \frac{h}{n}\sum_{i=1}^n I_{X_i = 1} K_h(Z_i)^2 r_p(Z_i/h)r_p(Z_i/h)^{\tr} Y_i^1(t) \bar{\alpha}_1(t,Z_i), 
\notag
\end{equation}
can be written on the form $n^{-1}\sum_{i=1}^n q_h(Z_i) Y_i^1(t)\bar{\alpha}_1(t,Z_i)$ where $q_h(z) = q(z/h)/h$ and $q(u) = I\{u \geq z_0\}K(u)^2 u^{2\nu}$ for $\nu = 0,\ldots,2p$. This function satisfies the conditions of Lemma~\ref{lemma::uniform_convergence_lemma}. By Assumption~\ref{assumption::potential_params} of the main text, $\bar{\alpha}_1(t,z)$ satisfies the conditions imposed on $\xi(t,z)$ in Lemma~\ref{lemma::uniform_convergence_lemma}. By said lemma, we therefore have that $h\bar{\Psi}_{1,p,n}(t,h)$ converges almost surely to $\lim_{h \to 0}h\,\E\,\bar{\Psi}_{1,p,n}(t,h)$, uniformly in $t$ as $nh \to \infty$ and $h \to 0$. Assuming that $h\kappa < \kappa_0$, this limit is
\begin{equation}
\begin{split}
h\,\E\,\bar{\Psi}_{1,p,n}(t,h) & = h^{-1} \, \E\,I_{Z \geq z_0} K(Z/h)^2 r_p(Z/h)r_p(Z/h)^{\tr} Y^1(t) \bar{\alpha}_1(t,Z)\\
& = h^{-1} \, \E\,I_{Z \geq z_0} K(Z/h)^2 r_p(Z/h)r_p(Z/h)^{\tr} \E\,\{Y^1(t) \mid Z\} \bar{\alpha}_1(t,Z)\\
& = \int_0^1 K(u)^2 r_p(u)r_p(u)^{\tr} y_1(t,hu) \bar{\alpha}_1(t,hu)f_Z(hu)\,\dd u\\
& = y_1(t,z_0) \bar{\alpha}_1(t,z_0)f_Z(z_0)\int_0^1 K(u)^2 r_p(u)r_p(u)^{\tr} \,\dd u + o(1)\\
& = y_1(t,z_0) \bar{\alpha}_1(t,z_0)f_Z(z_0)\Psi_{p} + o(1),\\
\end{split}
\notag
\end{equation}
using that, by Assumption~\ref{assumption::potential_params} of the main text, when $h \kappa < \kappa_0$, $z \mapsto y_1(t,z) \bar{\alpha}_1(t,z)$ is continuous and bounded for all $t \in [0,\tau]$. The proof of (ix) is the same, noting that 
\begin{equation}
h\,\E\,\bar{\Psi}_{0,p,n}(t,h) = y_1(t,z_0) \bar{\alpha}_1(t,z_0)f_Z(z_0)H_p(-1)\Psi_{p}H_p(-1) + o(1).
\notag
\end{equation}
The proofs of (x) and (xi) are identical to (viii) and (ix) upon noting that, due to Lemma~\ref{lemma::A4} of the main text,
\begin{equation}
\begin{split}
h\,\E\,\Psi_{g,p,n}(t,h) = h\,\E\,\E\,\{\Psi_{g,p,n}(t,h) \mid \Falg_{t-}\} = h\,\E\,\bar{\Psi}_{g,p,n}(t,h),
\notag
\end{split}
\end{equation}
for $g = 0,1$. 

Now, (xii). Recall that 
\begin{equation}
\bar{\Psi}_{g,p,q,n}(s,h,b)
= \frac{1}{n}\sum_{i=1}^n I_{X_i = g}
K_{h}(Z_i)K_{b}(Z_i)r_{p}(Z_i/h)r_{p}(Z_i/b)
Y_i^{g}(s)\bar{\alpha}_g(s,Z_i).
\notag
\end{equation}
Let $m_n = \min(h,b)$. By an argument similar to that used to prove (viii), we have that $(hb/m_n)\bar{\Psi}_{g,p,q,n}(s,h,b)$ converges almost surely to $\lim_{m_n \to 0}(hb/m_n)\E\,\bar{\Psi}_{g,p,q,n}(s,h,b)$, uniformly in $t$, as $nm_n \to \infty$ and $m_n \to 0$. Assuming that $m_n \kappa < \kappa_0$, and that, without loss of generality that $m_n = h$ for all $n$ and $m_n/b \to \rho \geq 0$ as $n \to \infty$, this limit is 
\begin{equation}
\begin{split}
& \frac{hb}{m_n}\E\,\bar{\Psi}_{1,p,q,n}(s,h,b)  
= \frac{1}{m_n}\int_0^{\infty} K\big(\frac{z}{h}\big)K\big(\frac{z}{b}\big)r_p\big(\frac{z}{h}\big)r_q\big(\frac{z}{b}\big)^{\tr}y_1(t,z) \bar{\alpha}_1(t,z)f_Z(z) \,\dd z\\
&\quad  = \int_0^{\infty} K\big(\frac{m_nu}{h}\big)K\big(\frac{m_nu}{b}\big)r_p\big(\frac{m_nu}{h}\big)r_q\big(\frac{m_nu}{b}\big)^{\tr}y_1(t,m_nu) \bar{\alpha}_1(t,m_nu)f_Z(m_nu) \,\dd u\\
&\quad  = y_1(t,z_0) \bar{\alpha}_1(t,z_0)f_Z(z_0)\int_0^{\infty} K(u) K(\rho u) r_p(u) r_q(\rho u)^{\tr}\,\dd u + o(1)\\
&\quad  = y_1(t,z_0) \bar{\alpha}_1(t,z_0)f_Z(z_0)\Psi_{p,q}(\rho) + o(1),
\end{split}
\notag
\end{equation}   
as $m_n \to 0$. This completes the proof.
\end{proof}

\section{Proof of Lemma~\ref{lemma::general_bias} }\label{subsec::bias_lemma} 

\begin{proof} We assume that $z_0 = 0$, but write $z_0$ instead of $0$ when this adds to the clarity of the exposition. From Lemma~\ref{lemma::barMG} of the main text, $\bar{M}_i^g(t) = I_{X_i = g}N_i^g(t) - I_{X_i = g}\int_0^t Y_i^g(s) \bar{\alpha}_{g}(s,Z_i)\,\dd s$ are $\Falg_t$-martingales. Combining this result with the Taylor expansion 
\begin{equation}
\bar{\alpha}_{g}(t,Z_i) = r_p(Z_i)^{\tr} \beta_{g,p}(t,z_0) 
+ \frac{\bar{\alpha}_{g}^{(p+1)}(t,z_0)}{(p+1)!}Z_i^{p+1}
+ \frac{\bar{\alpha}_{g}^{(p+2)}(t,\xi_{i})}{(p+2)!}Z_i^{p+2}, 
\notag
\end{equation}
with $\xi_{i}$ a (random) point between $Z_i$ and $z_0$, gives
\begin{equation} 
\begin{split}
& \E\,(\dd\widehat{B}_{g,p}(t,h)\mid \Falg_{t-})/\dd t 
 = J_{n,h}(t)\beta_{g,p}(t)\\  
& \quad + H_p(h)J_{n,h}(t)\Gamma_{g,p,n}(t,h)^{-1}\vartheta_{g,p,p+1,n}(t,h) h^{p+1} \frac{\bar{\alpha}_{g}^{(p+1)}(t,z_0)}{(p+1)!}\\ 
& \quad\; + H_p(h)J_{n,h}(t)\Gamma_{g,p,n}(t,h)^{-1}\big\{\vartheta_{g,p,p+2,n}(t,h)h^{p+2}\frac{\bar{\alpha}_{g}^{(p+2)}(t,z_0)}{(p+2)!} + \eps_{g,n}(t)\big\},
\end{split}
\label{eq::Bbias}
\end{equation}
where the $\vartheta_{g,p,q,n}(t,h)$ matrices are of~\eqref{eq::myXi}, and 
\begin{equation}
\eps_{g,n}(t) 
= \frac{h^{p+2}}{n}\sum_{i=1}^nI_{X_i = g}K_h(Z_i)r_p(Z_i/h)Y_i(t)(Z_i/h)^{p+2} 
\frac{\bar{\alpha}_{g}^{(p+2)}(t,\xi_i) - \bar{\alpha}_{g}^{(p+2)}(t,z_0)}{(p+2)!}. 
\notag
\end{equation} 
By Lemma \ref{lemma::prob_limits}(vi) and (vii), the convergence $J_{n,h}(t)\Gamma_{g,p,n}(t,h)^{-1}\vartheta_{g,p,q,n}(t,h) \to_p \Gamma_{1,p}^{-1}\vartheta_{1,p}$ as $nh \to \infty$ and $h \to 0$, is uniform in $t \in [0,\tau]$, and by Assumption~\ref{assumption::potential_params}(b) $\bar{\alpha}_{g}^{(p+1)}(t,z_0)$ is bounded on $[0,\tau]$. This yields uniform (in $t$) convergence of the second term on the right in~\eqref{eq::Bbias}. We need to show that $\eps_{g,n}(t) = O_p(h^{p+2})$ uniformly in $t$. For $\nu = 0,\ldots,p$, let $\eps_{g,n}(t)_{\nu}$ be the $\nu$th element of $\eps_{g,n}(t)$ and $\vartheta_{g,p,q,n}(t,h)_{\nu}$ be the $\nu$th element of $\vartheta_{g,p,q,n}(t,h)$, then for $\nu = 0,\ldots,p$
\begin{equation}
\begin{split}
\abs{ \eps_{g,n}(s)_\nu}\, & \lesssim 
\frac{h^{p+2}}{n}\sum_{i=1}^n \abs{I_{X_i = g} K_h(Z_i) \big(\frac{Z_i}{h}\big)^{\nu+p+2}Y_i(t)}
\abs{\bar{\alpha}_{g}^{(p+2)}(s,\xi_i) - \bar{\alpha}_{g}^{(p+2)}(s,z_0)}\\ 
&\leq \frac{h^{p+2}}{n}\sum_{i=1}^n \abs{I_{X_i = g} K_h(Z_i) \big(\frac{Z_i}{h}\big)^{\nu+p+2}Y_i(t)}
\,\max_{i \leq n}\abs{\bar{\alpha}_{g}^{(p+2)}(s,\xi_i) - \bar{\alpha}_{g}^{(p+2)}(s,z_0)}\\
& = h^{p+2}(-1)^{(1-g)(\nu + p + 2)}\vartheta_{g,p,q,n}(t,h)_{\nu} \,\max_{i \leq n}\abs{\bar{\alpha}_{g}^{(p+2)}(s,\xi_i) - \bar{\alpha}_{g}^{(p+2)}(s,z_0)}.
\end{split}
\notag
\end{equation}
From Lemma~\ref{lemma::prob_limits}(iv) and (v) we have that $\vartheta_{g,p,q,n}(t,h)_{\nu} = O_p(1)$ uniformly in $t$, thus we need to show that $\max_{i \leq n}\abs{\bar{\alpha}_{g}^{(p+2)}(s,\xi_i) - \bar{\alpha}_{g}^{(p+2)}(s,z_0)}$ is $O_p(1)$ uniformly in $t$. By Assumption~\ref{assumption::potential_params}(b) of the main text, the function $\bar{\alpha}_{g}^{(p+2)}(s,z)$ is continuous on $[0,\tau]\times (-\kappa_0,\kappa_0)$. Since $\xi_i$ is between $Z_i$ and $z_0 = 0$, we can assume that $h \kappa \leq \bar{h}\kappa < \kappa_0$, for som fixed $\bar{h} \geq h > 0$. Restricted to $[0,\tau]\times [-\bar{h}\kappa,\bar{h}\kappa]$, the function $\bar{\alpha}_{g}^{(p+2)}(s,z)$ is bounded, $\bar{\alpha}_{g}^{(p+2)}(s,z) \leq c$, say. Then $\max_{i \leq n}\abs{\bar{\alpha}_{g}^{(p+2)}(s,\xi_i) - \bar{\alpha}_{g}^{(p+2)}(s,z_0)} \leq 2c$ on  $[0,\tau]\times [-\bar{h}\kappa,\bar{h}\kappa]$, and so $\eps_{g,n}(t) = O_p(h^{p+2})$. In conclusion, 
\begin{equation}
\begin{split}
& \E\,(\dd\widehat{B}_{g,p}(t,h)\mid \Falg_{t-})/\dd t 
= J_{n,h}(t)\beta_{g,p}(t)\\  
& \qquad\qquad + H_p(h)J_{n,h}(t)\Gamma_{p}^{-1}\vartheta_{p,p+1} h^{p+1} \frac{\bar{\alpha}_{g}^{(p+1)}(t,z_0)}{(p+1)!}
+ H_p(h)O_p(h^{p+2}),
\end{split}
\notag
\end{equation}
uniformly in $t$. 
\end{proof}
The relation $\E\,(\dd\widehat{A}_{g,p}(t,h)\mid \Falg_{t-})/\dd t = e_{p,\nu}^{\tr}\E\,(\dd\widehat{B}_{g,p}(t,h)\mid \Falg_{t-})/\dd t$ and that $e_{p,\nu}^{\tr}H_p(h) = h^{-\nu}$ gives, as an immediate consequence of this lemma, that
\begin{equation} 
\begin{split}
&\E\,(\dd\widehat{A}_{g,p}^{(\nu)}(t,h)\mid \Falg_{t-})/\dd t   = J_{n,h}(t)\bar{\alpha}_g^{(\nu)}(t,z_0)\\ 
& \qquad\quad + h^{p+1-\nu}e_{p,\nu}^{\tr} \nu! \Gamma_p^{-1}\vartheta_{p,p+1}\frac{\bar{\alpha}_{g}^{(p+1)}(t,z_0)}{(p+1)!} + O_p(h^{p+1-\nu}).
\end{split}
\notag
\end{equation} 

\section{Proof of Theorem~\ref{theorem::clt1} (The $\Galg$-CLT)}\label{app::clt_proof} The proof of Theorem~\ref{theorem::clt1} is preceded by two lemmata. First, we include a general lemma that is used to prove weak convergence of $L_{g,p,n}$ from the weak convergence of $Q_{g,p,n}$ (see p.~\pageref{eq::Qplus} of the main text). Second, we prove that the c{\`a}dl{\`a}g version of the sequence $Q_{g,p,n}$ (as defined in~\eqref{eq::Qplus} of the main text) converges weakly to a mean zero Gaussian process. Finally, these two intermediate results are used to prove the theorem.

\begin{lemma}\label{lemma::integralmapping_CLT} Let $Z_{n} = (Z_{n,1},\ldots,Z_{n,p})^{\tr}$ be a sequence of $p$-dimensional  processes on the interval $[0,\tau]$ such that $Z_n \Rightarrow Z$ as $n \to \infty$. Let $A(s)$ be a deterministic $p \times p$ matrix of functions $a_{j,k}(s)$. Set $X_n(t) = \int_0^t A(s) Z_n(s)\,\dd s$, and assume that $\int_0^{\tau} \abs{\sum_{k=1}^p a_{j,k}(s)Z_{n,j}(s)}\,\dd s < \infty$ for each $j = 1,\ldots,p$ and all $n \geq 1$. Then $X_n \Rightarrow X$ where $X(t) = \int_0^t A(s) Z(s)\,\dd s$.  
\end{lemma}

\begin{proof} Write $X_n = (X_{n,1},\ldots,X_{n,p})$, where $X_{n,j}(t) = \int_0^t \sum_{k=1}^p a_{j,k}(s)Z_{n,k}(s)\,\dd s$, where $a_{j,k}(s)$ is element $(j,k)$ of the matrix $A(s)$. For each $t$, the mapping $g(z) = \int_0^t A(s)z(s)\,\dd s$ is continuous, so by the continuous mapping theorem 
\begin{equation}
(X_{n}(t_1),\ldots,X_{n}(t_k)) \overset{d}\to (X(t_1),\ldots,X(t_k)),
\notag
\end{equation} 
for any points $t_{1},\ldots,t_{k}\in [0,\tau]$. It remains to show that $X_n$ is tight. For $j = 1,\ldots,p$, let $V_a^b X_{n,j}$ be the variation process of $X_{n,j}$ over $[a,b]$. This means that $V_a^b X_{n,j}(\omega)$ is the total variation the path $s \mapsto X_n(s,\omega)$ on the interval $[a,b]$, see \citet[p.~27]{jacod2003limit}. For any $[a,b] \subset [0,\tau]$,
\begin{equation}
V_a^b X_{n,j} = \int_a^b \abs{\sum_{k=1}^p a_{j,k}(s)Z_{n,k}(s)}\,\dd s, 
\label{eq::Prop8.2}
\end{equation}
almost surely (see, e.g., \citet[Proposition~8.2, p.~286]{mcdonald2013course}). The process $t \mapsto V_0^t X_{n,j}$ is increasing, and, in view of~\eqref{eq::Prop8.2} where it is seen that $V_a^b X_{n,j}$ is a continuous mapping of $Z_{n}$, the continuous mapping theorem yields finite-dimensional convergence 
\begin{equation}
(V_0^{t_1} X_{n,j},\ldots,V_0^{t_k} X_{n,j}) \overset{d}\to (V_0^{t_1} X_{j},\ldots,V_0^{t_k}X_j),
\notag
\end{equation} 
for any $t_{1},\ldots,t_{k} \in[0,\tau]$, where $V_0^{t}X_j = \int_0^t \abs{\sum_{k=1}^p a_{j,k}(s) Z_{j}(s)}\,\dd s$. The process $V_0^{t}X_j$ is increasing and continuous, from which, by Theorem~VI.3.37 in \citet[p.~354]{jacod2003limit}, conclude that  
\begin{equation}
V_{0}^{\cdot} X_{n,j} \Rightarrow V_{0}^{\cdot} X_{j},
\notag
\end{equation}
for each $j = 1,\ldots,p$. Since for each $j$ the process $V_{0}^{t} X_{j}$ is continuous (in $t$), the process $\sum_{j=1}^k V_{0}^{t} X_{n,j}$ is $C$-tight (by Corollary~VI.3.33 in~\citep[p.~353]{jacod2003limit}). From this $C$-tightness it follows that $X_n$ is $C$-tight \citep[Proposition~3.36(a), p.~354]{jacod2003limit}. 
\end{proof}

\begin{lemma}\label{lemma::Qclt} Assumption~\ref{assumption::potential_params}{\rm (b)}
holds with $S \geq p+2$, Assumption~\ref{assumption::Galg_clt_extra} is in force, and the distribution function $H$ of $C$ has density $h(t)$ that is bounded on $[0,\tau]$. Set $c_g(s,t,z) = \E\,\{\Delta_g(s)\Delta_g(t)\mid Z = z,Y(s\vee t) = 1\}$, and let $Q_{g,p,n}(\cdot,h)$ be the right-continuous version of the process defined in~\eqref{eq::Qplus}. As $nh \to \infty$ and $h \to 0$,  
\begin{equation}
(nh)^{1/2}Q_{g,p,n}(\cdot,h) \Rightarrow Q_{g,p},
\notag
\end{equation}
where $Q_{g,p}$ is a mean zero continuous Gaussian process with 
\begin{equation}
\E\,Q_{1,p}(s)Q_{1,p}(t)^{\tr} = y_1(s\vee t,z_0)c_1(s,t,z_0)f_Z(z_0) \Psi_{p},
\notag
\end{equation}
while $\E\,Q_{0,p}(s)Q_{0,p}(t)^{\tr} = y_0(s\vee t,z_0)c_0(s,t,z_0)f_Z(z_0) H_p(-1)\Psi_{p}H_p(-1)$. 
\end{lemma}

\begin{proof} Define the right-continuous process $\breve{Y}^g(t) = I\{T^g > t\}$.\footnote{See also \citet[Theorem~5.1.8(iii), p.~113]{cohen2015stochastic}, and \citet[Theorem~3.13, p.~13]{karatzas1991brownian} for general results about sub- and super-martingales, that could have been applied here.} Since $\widetilde{T}^g$ and $C$ stem from continuous distributions (see Assumptions~\ref{assumption::lifetimes} and \ref{assumption::censoring}), $\breve{Y}^g$ is a version of $Y^g$, in particular $\E\,\breve{Y}^g(t) = \E\,Y^g(t) = \E\, y_g(t,Z)$ for all $t$. 
Define 
\begin{equation}
H_g(Z_i/h) = I_{X_i = g} K(Z_i/h) r_p(Z_i/h),\quad \text{and}
\quad 
V_{g,i}(t,h) = H_g(Z_i/h) \breve{Y}_i^g(t) \Delta_{g,i}(t),\
\notag
\end{equation} 
for $i = 1,\ldots,n$. 
A right-continuous version of the process defined in~\eqref{eq::Qplus} is 
\begin{equation}
Q_{g,p,n}(t,h) = (nh)^{-1}\sum_{i=1}^nH_g(Z_i/h) \breve{Y}_i^g(t) \Delta_{g,i}(t) =  (nh)^{-1}\sum_{i=1}^nV_{g,i}(t,h).
\label{eq::Qtilde}
\end{equation}
Let $\mathcal{Z}_t$ be the {\sigalg} as defined in Lemma~\ref{lemma::A4} (recall that this is not a filtration). Then $I_{X=g}\breve{Y}^g(t) \E\,\{\Delta_g(t)\mid \mathcal{Z}_t\} = I_{X=g}\breve{Y}^g(t) \E\,\{\alpha_g(t,Z,U) - \bar{\alpha}_g(t,Z)\mid \mathcal{Z}_t\} = 0$, and consequently $\E\, V_i(t,h)  = 0$. This means that, for $t$ fixed, the $V_1(t,h),\ldots,V_n(t,h)$ are i.i.d.~mean zero random variables. Notice also that
\begin{equation}
I_{X = g}\breve{Y}^g(t \vee s) \E\, \{\Delta_g(t)\Delta_g(s)\mid \mathcal{Z}_{t \vee s}\} 
= I_{X = g}\breve{Y}^g(t \vee s)c_g(t,s,Z).
\notag
\end{equation} 
Write $\psi_p(u)= K(u)^2r_p(u)r_p(u)^{\tr}$, so that $H_g(u)H_g(u)^{\tr} = I_{X=g}\psi_p(u)$. We now find the covariance function of the process $V_{1,i}(t,h)$. First,  
\begin{equation}
\begin{split}
\E\, V_{1,i}(t,h)V_{1,i}(t,h)^{\tr} & = 
\E\, I_{X=g}\psi_p(Z_i/h) \breve{Y}_i^g(t) \Delta_{g,i}(t) \breve{Y}_i^g(s) \Delta_{g,i}(s)\\
& = \E\, I_{X=g}\psi_p(Z_i/h) \breve{Y}_i^g(t \vee s)  \E\,\{ \Delta_{g,i}(t) \Delta_{g,i}(s)\mid \mathcal{Z}_{t \vee s}\}\\
& = \E\, I_{X=g}\psi_p(Z_i/h) \breve{Y}_i^g(t\vee s)c_g(t,s,Z_i)\\
& = \E\, I_{X=g}\psi_p(Z_i/h) \E\,\{\breve{Y}_i^g(t\vee s)\mid Z_i\} c_g(t,s,Z_i)\\
& =  \E\, I_{X=g}\psi_p(Z_i/h) y_g(t\vee s,Z_i)c_g(t,s,Z_i)\\
& = \int_{0}^{h \kappa} \psi_p(z/h) y_g(t\vee s,z)c_g(t,s,z)f_Z(z)\,\dd z\\
& = h\int_{0}^{\kappa} \psi_p(u) y_g(t\vee s,hu)c_g(t,s,hu)f_Z(hu)\,\dd u.
\end{split}
\notag
\end{equation} 
For $h$ such that $h\kappa < \kappa_0$, the functions $y_g(t\vee s,z)$, $c_g(t,s,z)$, and $f_Z(z)$ are continuous in $z$, therefore, using Assumptions~\ref{assumption::lifetimes}--\ref{assumption::potential_params} and Assumption~\ref{assumption::Galg_clt_extra}, dominated convergence yields   
\begin{equation}
\begin{split}
& h\int_{0}^{\kappa} \psi_p(u) y_g(t\vee s,hu)c_g(t,s,hu)f_Z(hu)\,\dd u\\
& \qquad = h y_1(t\vee s,z_0)c_1(t,s,z_0)f_Z(z_0)\int_{0}^{\kappa} K(u)^2 r_p(u)r_p(u)^{\tr} \,\dd x + o(h)\\
& \qquad = h y_1(t\vee s,z_0)c_1(t,s,z_0)f_Z(z_0) \Psi_{p} + o(h),
\end{split}
\notag
\end{equation}
as $h \to 0$. Similarly, 
\begin{equation}
\E\, V_{0,i}(t,h)V_{0,i}(t,h)^{\tr} = h y_1(t\vee s,z_0)c_1(t,s,z_0)f_Z(z_0) H_p(-1)\Psi_{p}H_p(-1) + o(h)
\notag
\end{equation}
as $h \to 0$. For any finite collection of points $t_1,\ldots,t_k$ in $[0,\tau]$, 
\begin{equation}
(nh)^{1/2}\{Q_{g,p,n}(t_1,h),\ldots,Q_{g,p,n}(t_k,h)\}
= (nh)^{-1/2}\sum_{i=1}^n \{V_i(t_1,h),\ldots,V_i(t_k,h)\},
 \notag
\end{equation}
where $\{V_1(t_1,h),\ldots,V_1(t_k,h)\},\ldots,\{V_n(t_1,h),\ldots,V_n(t_k,h)\}$ are i.i.d.~vectors. The multivariate central limit theorem for i.i.d~random vectors (e.g., Theorem~21.3 in \citet[p.~183]{jacod2004probability}), therefore takes care of finite-dimensional convergence: as $nh \to \infty$ and $h \to 0$, 
\begin{equation}
(nh)^{1/2}\{Q_{g,p,n}(t_1,h),\ldots,Q_{g,p,n}(t_k,h)\}
\overset{d}\to \normal(0,\Sigma_{t_1,\ldots,t_k}^g ),
 \notag
\end{equation}
where $\Sigma_{t_1,\ldots,t_k}^g$ consists of blocks of the form 
\begin{equation}
(\Sigma_{t_1,\ldots,t_k}^1)_{i,j}
= y_1(t_i\vee t_j,z_0)c_1(t_i,t_j,z_0)f_Z(z_0) \Psi_{p}.
\notag
\end{equation}
and 
\begin{equation}
(\Sigma_{t_1,\ldots,t_k}^0)_{i,j}
= y_0(t_i\vee t_j,z_0)c_0(t_i,t_j,z_0)f_Z(z_0) H_p(-1)\Psi_{p}H_p(-1).
\notag
\end{equation}
We now show that each element of $(nh)^{1/2}Q_{g,p,n}(t,h)$ is $C$-tight, as this implies that the whole vector is $C$-tight (see \citet[Def.~VI.3.25 and Corollary~VI.3.33 pp.~351--353]{jacod2003limit}). Let $h_{g,\nu}(z)$ be the $\nu$th element of $H_g(z)$,that is, $h_{g,\nu}(Z_i/h) = I_{X_i = g}K(Z_i/h) (Z_i/h)^{\nu}$. The $\nu$th element of $(nh)^{1/2}Q_{g,p,n}(t,h)$ is then
\begin{equation}
q_{g,p,n,\nu}(t,h)
 = (nh)^{-1/2}\sum_{i=1}^n h_{g,\nu}(Z_i/h) \breve{Y}_{i}(t)\Delta_{g,i}(t)
\notag
\end{equation}
for $\nu = 0,\ldots,p$. To show tightness we use the criterion given in \citet[Theorem~15.6, p.~128]{billingsley1968convergence} (see also \citet[Theorem~VI.4.1, p.~355]{jacod2003limit}). By Assumption~\ref{assumption::Galg_clt_extra}, both $\alpha_g(t,z,u)$ and $\bar{\alpha}_g(t,z)$ are Lipschitz in $t$, with Lipschitz constants $\ell_g$ and $\bar{\ell}_g$, respectively. Therefore, 
\begin{equation}
\begin{split}
\abs{\Delta_{g}(t) - \Delta_{g}(t)}\, 
& \leq 
\abs{\alpha_g(t,U,Z) - \alpha_g(t,Z,U)} \,+\, 
\abs{\bar{\alpha}_g(t,Z) - \bar{\alpha}_g(t,Z)}\\
& \leq \ell_g\,\abs{t - s} \,+ \, \bar{\ell}_g \,\abs{t - s}\,
= (\ell_g + \bar{\ell}_g) \abs{t - s}.
\end{split}
\notag
\end{equation} 
Now $\abs{q_{g,p,n,\nu}(t,h) - q_{g,p,n,\nu}(t,h)}\, \leq (nh)^{-1}\sum_{i=1}^n \abs{h_{g,\nu}(Z_i/h)} \,\abs{\breve{Y}_{i}(t)\Delta_{g,i}(t) - \breve{Y}_{i}(s)\Delta_{g,i}(s)} $, and using the triangle inequality and Assumption~\ref{assumption::Galg_clt_extra},
\begin{equation}
\begin{split}
\abs{\breve{Y}_{i}(t)\Delta_{g,i}(t) - \breve{Y}_{i}(s)\Delta_{g,i}(s)} \,
& \leq \abs{\breve{Y}_{i}(t) - \breve{Y}_{i}(s)}\,\abs{\Delta_{g,i}(t)} + \breve{Y}_{i}(s)\abs{\Delta_{g,i}(s) -\Delta_{g,i}(s)} \\
& \leq (\alpha_{g,\max} + \bar{\alpha}_{g,\max} )\abs{\breve{Y}_{i}(t) - \breve{Y}_{i}(s)}\, +\, (\ell_g + \bar{\ell}_g) \abs{t - s}\\
& \leq C (\abs{\breve{Y}_{i}(t) - \breve{Y}_{i}(s)}\, + \,\abs{t - s}),
\end{split} 
\notag
\end{equation}
where $C = 2 \max\{\alpha_{g,\max},\bar{\alpha}_{g,\max},\ell_g,\bar{\ell}_g \}$. Let $r < s < t$ be points in $[0,\tau]$. Define 
\begin{equation}
\begin{split}
\zeta_i & = h_{g,\nu}(Z_i/h)\{\breve{Y}_{i}(s)\Delta_{g,i}(s)
- \breve{Y}_{i}(r)\Delta_{g,i}(r)\},\\
\eta_i & = h_{g,\nu}(Z_i/h)\{\breve{Y}_{i}(t)\Delta_{g,i}(t)
- \breve{Y}_{i}(s)\Delta_{g,i}(s)\}. 
\end{split}
\notag
\end{equation}
Then, similarly to in \citet[p.~106]{billingsley1968convergence},
\begin{equation}
\begin{split}
& \E\,\abs{q_{g,p,n,\nu}(s,h) - q_{g,p,n,\nu}(r,h)}^2
\abs{q_{g,p,n,\nu}(t,h) - q_{g,p,n,\nu}(s,h)}^2\\
& \qquad\qquad 
= (nh)^{-2}\E\,\big(\sum_{i=1}^n  \zeta_i\big)^2
\big(\sum_{i=1}^n  \eta_i\big)^2\\ 
& \qquad \qquad\leq (nh)^{-2}\{ n \E\,(\zeta_1^2\eta_1^2) + 3 n(n-1)\E\,(\zeta_1^2)\,\E\,(\eta_2^2)\}.
\end{split}
\label{eq::billingsley_trick}
\end{equation}
We now find bounds on the two expectations on the right hand side. Using that $\abs{\breve{Y}_{i}(t)\Delta_{g,i}(t) - \breve{Y}_{i}(s)\Delta_{g,i}(s)} \,\leq C (\abs{\breve{Y}_{i}(t) - \breve{Y}_{i}(s)}\, + \,\abs{t - s})$, $\abs{\zeta_i}\, \leq C \abs{h_{g,\nu}(Z_i/h)}\,(\abs{\breve{Y}_{i}(s) - \breve{Y}_{i}(r)}\, + \,\abs{s - r})$ and $\abs{\eta_i} \leq C\abs{\abs{h_{g,\nu}(Z_i/h)}}\,(\abs{\breve{Y}_{i}(t) - \breve{Y}_{i}(s)}\, + \,\abs{t - s})$. Since $\pr\{\abs{\breve{Y}_{i}(s) - \breve{Y}_{i}(r)}\, = 1, \abs{\breve{Y}_{i}(t) - \breve{Y}_{i}(s)}\, = 0 \mid Z_i\} = y_g(r,Z_i) - y_g(s,Z_i)$, $\pr\{\abs{\breve{Y}_{i}(s) - \breve{Y}_{i}(r)}\, = 0, \abs{\breve{Y}_{i}(t) - \breve{Y}_{i}(s)}\, = 1 \mid Z_i\} = y_g(s,Z_i) - y_g(t,Z_i)$, and $\pr\{\abs{\breve{Y}_{i}(s) - \breve{Y}_{i}(r)}\, = 0, \abs{\breve{Y}_{i}(t) - \breve{Y}_{i}(s)}\, = 0 \mid Z_i\} = 1 - \{y_g(r,Z_i) - y_g(t,Z_i)\}$, we have that 
\begin{equation}
\begin{split}
\E\,(  \zeta_i^2 \mid Z_i) & \leq C^2  h_{g,\nu}(Z_i/h)^2 \E\, \{(\abs{\breve{Y}_i(s) - \breve{Y}_i(r)} \,+\, \abs{s - r})^2\mid Z_i\}\\
& = C^2  h_{g,\nu}(Z_i/h)^2  (1 + \abs{s - r})^2 \{y_g(r,Z_i) - y_g(s,Z_i)\}\\
& \qquad + C^2  h_{g,\nu}(Z_i/h)^2 \abs{s - r}^2 (1 - \{y_g(r,Z_i) - y_g(s,Z_i)\})\\
& \leq C^2  h_{g,\nu}(Z_i/h)^2\{(1 + \abs{t - r})^2\{y_g(r,Z_i) - y_g(s,Z_i)\} + \abs{s - r}^2\}.
\end{split}
\label{eq::condZ1}
\end{equation} 
Recall that $y_g(t,Z_i) = \{1 - H(t)\}\bar{S}_g(t,Z_i)$, and that $H(t)$ is assumed to have a bounded density $h(t)$, say $\abs{h(t)}\,\leq h_{\max}$. Since $\partial S_g(t,z,u)/\partial t = -\alpha_g(t,z,u)S_g(t,z,u)$ is continuous in $t$ for all $z$ and $u$ by Assumption~\ref{assumption::lifetimes} of the main text, and the derivative $\abs{\partial S_g(t,z,u)/\partial t} \,\leq \abs{\alpha_g(t,z,u)} \leq \alpha_{g,\max}$ is bounded (due to Assumption~\ref{assumption::Galg_clt_extra}), we can pass the derivative under the integral sign,\footnote{See, e.g., \citet[p.~124]{ferguson1996largesample}, the extension to conditional expectations holds because the dominated convergence theorem extends to conditional expectations, see, e.g.~\citet[Lemma~2.4.3, p.~57]{cohen2015stochastic}} 
\begin{equation}
\frac{\partial }{\partial t } \bar{S}_g(t,Z) = 
\frac{\partial }{\partial t } \E\, \{S_g(t,Z,U) \mid Z\}
= - \E\, \{\alpha_g(t,Z,U)S_g(t,Z,U) \mid Z\}, 
\notag
\end{equation}    
so that $\abs{\partial\bar{S}_g(t,Z) /\partial t} \, \leq \E\, \{\abs{\alpha_g(t,Z,U)S_g(t,Z,U)} \mid Z\} \, 
\leq \alpha_{g,\max}$. We then find a bound for $\partial y_g(t,z)/\partial t$,
\begin{equation}
\begin{split}
\abs{\frac{\partial}{\partial t}y_g(t,z) }\,
& = \abs{ -h(t)\bar{S}_g(t,z) + \{1 - H(t)\}\frac{\partial}{\partial t} \bar{S}_g(t,z)}\\
& \leq h(t) + \abs{\frac{\partial}{\partial t} \bar{S}_g(t,z)}
\leq h(t) + \alpha_{g,\max} \leq h_{\max} + \alpha_{g,\max}.
\end{split}
\notag
\end{equation}
We denote this upper bound $\dot{y}_{g,\max} = h_{\max} + \alpha_{g,\max}$. An application of the mean value theorem then yields $\abs{y_g(r,Z_i) - y_g(s,Z_i)}\, \leq \dot{y}_{g,\max}\abs{r - s}$, which inserted in~\eqref{eq::condZ1} gives 
\begin{equation}
\begin{split}
\E\,(  \zeta_i^2 \mid Z_i) &  \leq C^2  h_{g,\nu}(Z_i/h)^2\{\dot{y}_{g,\max}(1 + \abs{t - r})^2\abs{r - s} + \abs{s - r}^2\}\\
& \leq C^2  h_{g,\nu}(Z_i/h)^2 \{\dot{y}_{g,\max}(1 + \tau)^2 + \tau\}\abs{t - r}.
\end{split}
\notag
\end{equation} 
Write $k_1 = \dot{y}_{g,\max}(1 + \tau)^2 + \tau$. Then $\E\,(  \eta_i^2 \mid Z_i) \leq k_1C^2  h_{g,\nu}(Z_i/h)^2 \abs{t - r}$ by the same arguments, and $\E\,(  \zeta_i^2 \mid Z_i) \E\,(  \eta_i^2 \mid Z_i) \leq k_1^2C^4  h_{g,\nu}(Z_i/h)^4 \abs{t - r}^2$. The second conditional expectation is
\begin{equation}
\begin{split}
\E\,(  \zeta_i^2 \eta_i^2 \mid Z_i) 
& \leq C^2  h_{g,\nu}(Z_i/h)^2 (1 \,+\, \abs{s - r})^2\abs{t - s}^2 \{y_g(r,Z_i) - y_g(s,Z_i)\} \\
 & \qquad + C^2  h_{g,\nu}(Z_i/h)^2 \abs{s - r}^2 ( 1 \,+\, \abs{t - s})^2\{y_g(s,Z_i) - y_g(t,Z_i)\}\\
 & \qquad + C^2  h_{g,\nu}(Z_i/h)^2 \abs{s - r}^2 \abs{t - s}^2(1 - \{y_g(r,Z_i) - y_g(t,Z_i)\})\\
 & \leq C^2  h_{g,\nu}(Z_i/h)^2\{	3 (1 + \tau)^2 \abs{t - r}^2\},
\end{split}
\notag
\end{equation} 
so that $\E\,(  \zeta_i^2 \eta_i^2 \mid Z_i) \leq k_2 C^2  h_{g,\nu}(Z_i/h)^2\abs{t - r}^2$, writing $k_2 = 3 (1 + \tau)^2$. Given that $h \kappa < \kappa_0$ and using Assumptions~\ref{assumption::potential_params}(a) and \ref{assumption::kernel} with $K(u) \leq K_{\max}$, the expectation of $h_{g,\nu}(Z_i/h)^s$ for $s = 2,4$ is 
\begin{equation}
\begin{split}
\E\, h_{g,\nu}(Z/h)^s & = \int_0^{\infty} h_{g,\nu}(z/h)^sf_Z(z)\,\dd z
= \int_{-h\kappa}^{h\kappa} K(z/h)^s (z/h)^{\nu s}f_Z(z)\,\dd z\\
& = h \int_{-\kappa}^{\kappa} K(u)^s (u)^{\nu s}f_Z(hu)\,\dd u
\leq h \kappa^{\nu s} K_{\max}^s \int_{-\kappa}^{\kappa} f_Z(hu)\,\dd u\\
& \leq h 2\,\kappa^{\nu s +1} K_{\max}^s f_{Z,\max}.
\end{split}
\notag
\end{equation}
With these bounds, still assuming that $h \kappa < \kappa_0$, 
\begin{equation}
\E\,(  \zeta_1^2 ) \E\,(  \eta_2^2 ) 
= \E\,\{\E\,(  \zeta_1^2\mid Z_1 )\}\, \E\,\{\E\,(  \eta_2^2\mid Z_2 )\} 
= h k_1^2 C^4 \kappa^{4\nu + 1}K_{\max}^4 f_{Z,\max} \,\abs{t - r}^2,
\notag
\end{equation} 
and
\begin{equation}
\E\,(\zeta_1^2 \eta_1^2)  
= \E\,\{\E\,(\zeta_1^2 \eta_1^2 \mid Z_1)\}
= h k_2 C^2 \kappa^{2\nu + 1} K_{\max}^2 f_{Z,\max} \,\abs{t - r}^2.
\notag
\end{equation} 
Inserting these two expressions on the right hand side of~\eqref{eq::billingsley_trick}, gives
\begin{equation}
\E\,\abs{q_{g,p,n,\nu}(s,h) - q_{g,p,n,\nu}(r,h)}^2
\abs{q_{g,p,n,\nu}(t,h) - q_{g,p,n,\nu}(s,h)}^2\,
\lesssim \abs{t - r}^2, 
\notag
\end{equation}
which proves that the sequences (in $n$) of processes $q_{g,p,n,\nu}(t,h)$ are tight for all $\nu = 0,\ldots,p$. Let $\Delta q_{g,p,n,\nu}(t,h) = q_{g,p,n,\nu}(t,h) - q_{g,p,n,\nu}(t-,h)$ be the jump of $q_{g,p,n,\nu}(t,h)$ at $t$. For each $\nu$, since $q_{g,p,n,\nu}(t,h)$ is tight, it is $C$-tight provided 
\begin{equation} 
\lim_{n \to \infty}\pr(\sup_{t \leq \tau}\,\abs{\Delta q_{g,p,n,\nu}(t,h)}\, > \eps) = 0, 
\notag
\end{equation}
for any $\eps > 0$ \citep[Proposition~VI.3.26, p.~351]{jacod2003limit}. By continuity in $t$ of $\Delta_{g,i}(t)$ (Assumption~\ref{assumption::lifetimes}), $\Delta q_{g,p,n,\nu}(t,h)
= (nh)^{-1/2} \sum_{i=1}^n h_{g,\nu}(Z_i/h) \Delta_{g,i}(t)\Delta \breve{Y}_i(t)$
with $\Delta \breve{Y}_i(t) =\breve{Y}_i(t) - \breve{Y}_i(t-)$. Using Assumptions~\ref{assumption::kernel} and \ref{assumption::Galg_clt_extra} we have that this jump $\abs{\Delta q_{g,p,n,\nu}(t,h)}\, \leq K_{\max} (\alpha_{g,\max} + \bar{\alpha}_{g,\max}) (nh)^{-1/2}\sum_{i=1}^n \Delta\breve{Y}_i(t)$. Since $\Delta \breve{Y}_i(t)= 1$ if $T_i = t$, and equals zero otherwise, and no two censored lifetimes $T_i$ can jump at the same time, $\sup_{t \leq \tau}\abs{\Delta q_{g,p,n,\nu}(t,h)}\, \lesssim (nh)^{-1/2}$, and we conclude that the sequence $q_{g,p,n,\nu}(\cdot,h)$ is $C$-tight. Since $q_{g,p,n}^{(0)}(\cdot,h),\ldots,q_{g,p,n}^{(p)}(\cdot,h)$ are all $C$-tight, the process $(nh)^{1/2}Q_{g,p,n}(\cdot,h) = \{q_{g,p,n}^{(0)}(\cdot,h),\ldots,q_{g,p,n}^{(p)}(\cdot,h)\}^{\tr}$ is $C$-tight \citep[Corollary~VI.3.33, p.~353]{jacod2003limit}. 
\end{proof}
\begin{proof}{\sc (of Theorem~\ref{theorem::clt1})} Set $\Zscr_{g,p,n}(t,h_n) = (nh_n)^{1/2}\{ \xi_{g,p,n}(t,h_n)^{\tr},Q_{g,p,n}(t,h_n)^{\tr}\}^{\tr}$, where let $Q_{g,p,n}(\cdot,h_n)$ is the right-continuous version of the process in~\eqref{eq::Qplus}, and $\xi_{g,p,n}$ is as defined in~\eqref{eq::xi_process}, replacing $\bar{M}_i^g$ with $M_i^g$. Due to Assumption~\ref{assumption::Galg_clt_extra}, the arguments from the proof of Theorem~\ref{theorem::Falg_clt} go through when $\bar{\alpha}_g(t,Z_i)$ is replaced by $\alpha_g(t,Z_i,U_i)$, therefore $\Zscr_{g,p,n}(t,h_n) = (nh_n)^{1/2}\{ M_{g,p,n}(t,h_n)^{\tr},Q_{g,p,n}(t,h_n)^{\tr}\}^{\tr} + o_p(1)$ as $nh_n \to \infty$ and $h_n \to 0$ from the proof of that theorem. For any points $0 \leq t_1 < t_2 < \cdots t_k \leq \tau$, the vector $\{\Zscr_{g,p,n}(t_1,h_n)^{\tr},\ldots,\Zscr_{g,p,n}(t_k,h_n)^{\tr}\}^{\tr}$ is a sum of i.i.d.~$2(p+1)$ dimensional column vectors (the independence comes about because $\Gamma_{g,p,n}$ does not appear in the integrand of $\xi_{g,p,n}$, see~\eqref{eq::xi_process}). From the multivariate central limit for i.i.d.~vectors, we then get $\{\Zscr_{g,p,n}(t_1,h_n)^{\tr},\ldots,\Zscr_{g,p,n}(t_k,h_n)^{\tr}\}^{\tr} \to_d \{\Zscr_{g,p}(t_1,h_n)^{\tr},\ldots,\Zscr_{g,p}(t_k,h_n)^{\tr}\}^{\tr}$ as $nh_n \to \infty$ and $h_n \to 0$, where the limiting distribution has a mean zero normal distribution with a covariance matrix characterised by
\begin{equation}
\cov\{ \Zscr_{g,p}(t_i),\Zscr_{g,p}(t_j)\} = 
\begin{pmatrix}
\Sigma_{00}(t_i,t_j) & \Sigma_{01}(t_i,t_j)\\
\Sigma_{10}(t_i,t_j) & \Sigma_{11}^g(t_i,t_j) 
\end{pmatrix},
\notag
\end{equation}    
where, from 
with, from Theorem~\ref{theorem::Falg_clt}, 
\begin{equation}
\Sigma_{00}(t_i,t_j) = \frac{1}{f_Z(z_0)} \int_0^{t_i \wedge t_j} \frac{\bar{\alpha}_{g}(s,z_0)}{y_{g}(s,z_0)}\,\dd s \,\Gamma_p^{-1}\Psi_{p}\Gamma_{p}^{-1},
\notag
\end{equation}
and $\Sigma_{11}^0(t_i,t_j) = y_0(t_i\vee t_j,z_0)c_0(t_i,t_j,z_0)f_Z(z_0) H_{p}(-1)\Psi_{p}H_{p}(-1)$ and $\Sigma_{11}^1(t_i,t_j) = y_1(t_i\vee t_j,z_0)c_1(t_i,t_j,z_0)f_Z(z_0) \Psi_{p}$, from Lemma~\ref{lemma::Qclt}, and the off-diagonal blocks $\Sigma_{01}$ and $\Sigma_{10}^{\tr}$ are both of the form
\begin{equation}
\Sigma_{01}(t_i,t_j) = 
-\int_0^{t_i \wedge t_j} \frac{c_{g}(u,t_i \vee t_j,z_0)}{y_g(u,z_0)}\,\dd u\,y_g(t_i \vee t_j,z_0)\,\Gamma_{p}^{-1}\Psi_{p}.
\notag
\end{equation}
This takes care of finite-dimensional convergence of $\Zscr_{g,p,n}(t,h)$. We now turn to its tightness. By Lemma~\ref{lemma::Qclt}, the sequence $(nh)^{1/2}Q_{g,p,n}$ is $C$-tight. Moreover, the proof of Theorem~\ref{theorem::Falg_clt} can be used to show that $(nh)^{1/2}\xi_{g,p,n}$ converges weakly (that is, when $\bar{M}_i^g$ is replaced by $M_i^g$, as it is in this proof), $(nh)^{1/2}\xi_{g,p,n}$ is tight. Corollary~VI.3.33 in \citet[p.~353]{jacod2003limit} then gives that $\Zscr_{g,p,n}$ is tight. Combining the finite dimensional convergence and the tightness with Lemma~VI.3.31 in \citet[p.~352]{jacod2003limit}, we conclude that,    
\begin{equation}
(nh)^{1/2}(M_{g,p,n}(\cdot,h_n),Q_{g,p,n}(\cdot,h_n)) = \Zscr_{g,p,n}(\cdot,h_n) + o_p(1) \Rightarrow (M_{g,p},Q_{g,p}), 
\label{eq::key_convergence}
\end{equation}  
as $nh_n \to \infty$ and $h_n \to 0$. We now turn to $(nh_n)^{1/2}(M_{g,p,n}(\cdot,h_n),L_{g,p,n}(\cdot,h_n))$. Let $\widetilde{\Gamma}_{1,p}(t) = y_1(t,z_0)f_Z(z_0)\Gamma_p$ and $\widetilde{\Gamma}_{0,p}(t) = y_0(t,z_0)f_Z(z_0)H_p(-1)\Gamma_pH_p(-1)$, set $D_{g,p,n}(t,h_n) = J_{n,h_n}(t)\Gamma_{g,p,n}(t,h)^{-1} - \widetilde{\Gamma}_{g,p}(t)$, and define 
\begin{equation}
\widetilde{L}_{g,p,n}(t,h) = \int_{0}^t \widetilde{\Gamma}_{1,p}(s)^{-1}Q_{g,p,n}(s,h_n)\,\dd s,
\notag
\end{equation} 
and $\widetilde{r}_{g,p,n}(t,h) = \int_{0}^t D_{g,p,n}(s,h_n)Q_{g,p,n}(s,h_n)\,\dd s$. From Lemma~\ref{lemma::prob_limits}(i)--(iii) combined with Lemma~\ref{lemma::invertible_to_limit} in the appendix, we have that $\sup_{t\in[0,\tau]}\norm{D_{g,p,n}(t,h_n)} = o_p(1)$ as $nh_n \to \infty$ and $h_n \to 0$, and from Lemma~\ref{lemma::Qclt} $Q_{g,p,n}(s,h_n)$ is $O_p((nh_n)^{-1/2})$ and tight. Therefore, $\sup_{s \in [0,\tau] }\norm{D_{g,p,n}(s,h_n)\widetilde{Q}_{g,p,n}(s,h_n)} = o_p((nh_n)^{-1/2})$, which entails that $L_{g,p,n}(t,h_n) = \widetilde{L}_{g,p,n}(t,h_n) + o_p((nh_n)^{-1/2})$
uniformly in $t$, and therefore $(nh_n)^{1/2}\{M_{g,p,n}(t,h_n),L_{g,p,n}(t,h_n)\} = (nh_n)^{1/2}\{\xi_{g,p,n}(t,h_n),\widetilde{L}_{g,p,n}(t,h_n)\} + o_p(1)$ uniformly in $t$. Again, both $\xi_{g,p,n}(t,h_n)$ and $\widetilde{L}_{g,p,n}(t,h_n)$ are sums of i.i.d.~random vectors, so by the central limit theorem we get finite-dimensional convergence: For any $0 \leq t_1 < \cdots < t_k \leq \tau$, as $nh_n \to \infty$ and $h_n \to 0$
\begin{equation}
\begin{split}
& (nh_n)^{1/2}\{(\xi_{g,p,n}(t_1,h_n), \widetilde{L}_{g,p,n}(t_1,h_n)) ,\ldots,(\xi_{g,p,n}(t_k,h_n), \widetilde{L}_{g,p,n}(t_k,h_n)) \}\\
& \qquad\qquad \overset{d}\to 
\{(M_{g,p}(t_1), L_{g,p}(t_1)) ,\ldots,(M_{g,p}(t_k), L_{g,p}(t_k)) \},
\end{split}
\notag
\end{equation} 
where $\{(M_{g,p}(t_1), L_{g,p}(t_1)) ,\ldots,(M_{g,p}(t_k), L_{g,p}(t_k)) \}$ follows a mean zero normal distribution with 
\begin{equation}
\begin{split}
& \cov\{(M_{g,p}(t_i), L_{g,p}(t_i)),(M_{g,p}(t_j), L_{g,p}(t_j)\}\\ 
& \qquad = 
\begin{pmatrix}
\Sigma_{00}(t_i,t_j) & - \int_0^{t_i} \Sigma_{01}(u,t_j) \Gamma_{g,p}(u)^{-1}\,\dd u\\
- \int_0^{t_j} \Sigma_{10}(t_i,u) \Gamma_{g,p}(u)^{-1}\,\dd u & \Sigma_{11}^g(t_i,t_j)
\end{pmatrix}.
\end{split}
\label{eq::finite_dim}
\end{equation}
By Lemma~\ref{lemma::integralmapping_CLT}, $(nh)^{1/2}\breve{L}_{g,p,n}$ is $C$-tight, and, as we saw above, $(nh)^{1/2}\xi_{g,p,n}$ is tight, therefore $(nh)^{1/2}\{\xi_{g,p,n}(\cdot,h_n),\widetilde{L}_{g,p,n}(\cdot,h_n)\}$ is tight \citep[Corollary~VI.3.33, p.~353]{jacod2003limit}. Again, combining these results with Lemma~VI.3.31 in \citet[p.~352]{jacod2003limit}, we conclude that, as $nh_n \to \infty$ and $h_n \to 0$,
\begin{equation}
\begin{split}
& (nh_n)^{1/2}(M_{g,p,n}(\cdot,h_n), L_{g,p,n}(\cdot,h_n))\\ 
& \qquad \qquad =  (nh_n)^{1/2}(\xi_{g,p,n}(\cdot,h),\widetilde{L}_{g,p,n}(\cdot,h_n)) + o_p(1)
\Rightarrow (M_{g,p},L_{g,p}),
\end{split}
\notag
\end{equation}
where $(M_{g,p},L_{g,p})$ is a mean zero Gaussian process with finite-dimensional distributions as given in~\eqref{eq::finite_dim}. The claim of the theorem now follows from an application of the Cram{\'e}r--Wold device and by Proposition~VI.3.17 in~\citet[p.~350]{jacod2003limit} since $\Bfrak_{g,p}(t)$ is continuous in $t$.  
\end{proof}

\bibliography{refs_survdiscont}

\begin{thebibliography}{}

\bibitem[Aalen, 1978]{aalen1978nonparametric}
Aalen, O.~O. (1978).
\newblock Nonparametric inference for a family of counting processes.
\newblock {\em The Annals of Statistics}, 6:701--726.

\bibitem[Aalen, 1980]{Aalen80}
Aalen, O.~O. (1980).
\newblock A model for nonparametric regression analysis of counting processes.
\newblock {\em Lecture Notes in Statistics}, 2:1--25.

\bibitem[Aalen, 1989]{Aalen89}
Aalen, O.~O. (1989).
\newblock A linear regression model for the analysis of life times.
\newblock {\em Statistics in Medicine}, 8:907--925.

\bibitem[Aalen, 1993]{Aalen93}
Aalen, O.~O. (1993).
\newblock Further results on the nonparametric linear regression model in
  survival analysis.
\newblock {\em Statistics in Medicine}, 12:1569--1588.

\bibitem[Aalen et~al., 2010]{aalen2010history}
Aalen, O.~O., Andersen, P.~K., Borgan, {\O}., Gill, R.~D., and Keiding, N.
  (2010).
\newblock History of applications of martingales in survival analysis.
\newblock {\em arXiv preprint arXiv:1003.0188}.

\bibitem[Aalen et~al., 2008]{ABG08}
Aalen, O.~O., Borgan, {\O}., and Gjessing, H. (2008).
\newblock {\em Survival and Event History Analysis: A Process Point of View}.
\newblock Springer Verlag, Berlin.

\bibitem[Andersen et~al., 1993]{ABGK93}
Andersen, P.~K., Borgan, {\O}., Gill, R.~D., and Keiding, N. (1993).
\newblock {\em Statistical {M}odels {B}ased on {C}ounting {P}rocesses}.
\newblock Springer, Berlin.

\bibitem[Andersen and Gill, 1982]{andersen1982cox}
Andersen, P.~K. and Gill, R.~D. (1982).
\newblock Cox{'}s regression model for counting processes: {A} large sample
  study.
\newblock {\em The Annals of Statistics}, 10:1100--1120.

\bibitem[Billingsley, 1968]{billingsley1968convergence}
Billingsley, P. (1968).
\newblock {\em Convergence of {P}robability {M}easures}.
\newblock John Wiley \& Sons, New York.

\bibitem[Billingsley, 1995]{billingsley1995probability}
Billingsley, P. (1995).
\newblock {\em {P}robability and {M}easure. {T}hird {E}dition}.
\newblock John Wiley \& Sons, New York.

\bibitem[Calonico et~al., 2014a]{calonico2014robust}
Calonico, S., Cattaneo, M.~D., and Titiunik, R. (2014a).
\newblock Robust nonparametric confidence intervals for
  regression-discontinuity designs.
\newblock {\em Econometrica}, 82:2295--2326.

\bibitem[Calonico et~al., 2014b]{calonico2014supplement}
Calonico, S., Cattaneo, M.~D., and Titiunik, R. (2014b).
\newblock Supplement to {`}{R}obust nonparametric confidence intervals for
  regression-discontinuity designs{'}.
\newblock {\em Econometrica, supplemental material}.

\bibitem[Cattaneo and Titiunik, 2022]{cattaneo2022regression}
Cattaneo, M.~D. and Titiunik, R. (2022).
\newblock Regression discontinuity designs.
\newblock {\em Annual Review of Economics}, 14:821--851.

\bibitem[Cohen and Elliott, 2015]{cohen2015stochastic}
Cohen, S.~N. and Elliott, R.~J. (2015).
\newblock {\em Stochastic {C}alculus and {A}pplications. {S}econd {E}dition}.
\newblock Birkh{\"a}user, Heidelberg.

\bibitem[Fan and Gijbels, 1996]{fan1996local}
Fan, J. and Gijbels, I. (1996).
\newblock {\em Local {P}olynomial {M}odelling and {I}ts {A}pplications}.
\newblock Chapman {\&} Hall, London.

\bibitem[Ferguson, 1996]{ferguson1996largesample}
Ferguson, T.~S. (1996).
\newblock {\em A {C}ourse in {L}arge {S}ample {T}heory}.
\newblock Chapman {\&} Hall, London.

\bibitem[Hahn et~al., 2001]{hahn2001identification}
Hahn, J., Todd, P., and Van~der Klaauw, W. (2001).
\newblock Identification and estimation of treatment effects with a
  regression-discontinuity design.
\newblock {\em Econometrica}, 69:201--209.

\bibitem[Hjort, 1992]{hjort1992inference}
Hjort, N.~L. (1992).
\newblock On inference in parametric survival data models.
\newblock {\em International Statistical Review/Revue Internationale de
  Statistique}, 60:355--287.

\bibitem[Hjort and Pollard, 1993]{HjortPollard93}
Hjort, N.~L. and Pollard, D.~B. (1993).
\newblock Asymptotics for minimisers of convex processes.
\newblock Technical report, Department of Mathematics, University of Oslo.
\newblock Available on arXiv preprint arXiv:1107.3806.

\bibitem[Hjort and Stoltenberg, 2021]{hjort2021partly}
Hjort, N.~L. and Stoltenberg, E.~A. (2021).
\newblock The partly parametric and partly nonparametric additive risk model.
\newblock {\em Lifetime Data Analysis}.

\bibitem[Holland, 1986]{holland1986statistics}
Holland, P.~W. (1986).
\newblock Statistics and causal inference.
\newblock {\em Journal of the American Statistical Association}, 81:945--960.

\bibitem[Imbens, 2020]{imbens2020potential}
Imbens, G.~W. (2020).
\newblock Potential outcome and directed acyclic graph approaches to causality:
  Relevance for empirical practice in economics.
\newblock {\em Journal of Economic Literature}, 58:1129--1179.

\bibitem[Imbens and Lemieux, 2008]{imbens2008regression}
Imbens, G.~W. and Lemieux, T. (2008).
\newblock Regression discontinuity designs: A guide to practice.
\newblock {\em Journal of Econometrics}, 142:615--635.

\bibitem[Imbens and Rubin, 2015]{imbens2015causal}
Imbens, G.~W. and Rubin, D.~B. (2015).
\newblock {\em Causal {I}nference in {S}tatistics, {S}ocial, and {B}iomedical
  {S}ciences}.
\newblock Cambridge University Press, Cambridge.

\bibitem[Jacod and Protter, 2004]{jacod2004probability}
Jacod, J. and Protter, P. (2004).
\newblock {\em Probability {E}ssentials. {S}econd {E}dition}.
\newblock Springer, Berlin.

\bibitem[Jacod and Shiryaev, 2003]{jacod2003limit}
Jacod, J. and Shiryaev, A. (2003).
\newblock {\em Limit {T}heorems for {S}tochastic {P}rocesses. {S}econd
  {E}dition}.
\newblock Springer, Berlin.

\bibitem[Kallenberg, 2002]{kallenberg2002foundations}
Kallenberg, O. (2002).
\newblock {\em Foundations of {M}odern {P}robability. {S}econd {E}dition}.
\newblock Springer, Berlin.

\bibitem[Karatzas and Shreve, 1991]{karatzas1991brownian}
Karatzas, I. and Shreve, S. (1991).
\newblock {\em Brownian {M}otion and {S}tochastic {C}alculus. {S}econd
  {E}dition}.
\newblock Springer-Verlag, New York.

\bibitem[Lee and Lemieux, 2010]{lee2010regression}
Lee, D.~S. and Lemieux, T. (2010).
\newblock Regression discontinuity designs in economics.
\newblock {\em Journal of Economic Literature}, 48(2):281--355.

\bibitem[McDonald and Weiss, 2013]{mcdonald2013course}
McDonald, J.~N. and Weiss, N.~A. (2013).
\newblock {\em A {C}ourse in {R}eal {A}nalysis. {S}econd {E}dition}.
\newblock Academic Press, Waltham.

\bibitem[McKeague, 1988a]{mckeague1988asymptotic}
McKeague, I.~W. (1988a).
\newblock Asymptotic theory for weighted least squares estimators.
\newblock In {\em Statistical Inference from Stochastic Processes: Proceedings
  of the AMS-IMS-SIAM Joint Summer Research Conference 1987}, volume~80, pages
  139--152. American Mathematical Society.

\bibitem[McKeague, 1988b]{mckeague1988counting}
McKeague, I.~W. (1988b).
\newblock A counting process approach to the regression analysis of grouped
  survival data.
\newblock {\em Stochastic Processes and their Applications}, 28:221--239.

\bibitem[Miller, 2018]{miller2018detailed}
Miller, J.~W. (2018).
\newblock A detailed treatment of {D}oob{'}s theorem.
\newblock {\em arXiv preprint arXiv:1801.03122}.

\bibitem[Morgan and Winship, 2015]{morgan2015counterfactuals}
Morgan, S.~L. and Winship, C. (2015).
\newblock {\em Counterfactuals and {C}ausal {I}nference}.
\newblock Cambridge University Press, Cambridge.

\bibitem[Pearl, 2009]{pearl2009causality}
Pearl, J. (2009).
\newblock {\em {C}ausality: {M}odels, {R}easoning and {I}nference. {S}econd
  {E}dition}.
\newblock Cambridge University Press, Cambridge.

\bibitem[Pearl et~al., 2016]{pearl2016causal}
Pearl, J., Glymour, M., and Jewell, N.~P. (2016).
\newblock {\em Causal inference in statistics: {A} primer}.
\newblock John Wiley \& Sons, New York.

\bibitem[Pearl and Mackenzie, 2018]{pearl2018why}
Pearl, J. and Mackenzie, D. (2018).
\newblock {\em The {B}ook of {W}hy: {T}he {N}ew {S}cience of {C}ause and
  {E}ffect}.
\newblock Basic Books, New York.

\bibitem[Ramlau-Hansen, 1983]{ramlau1983smoothing}
Ramlau-Hansen, H. (1983).
\newblock Smoothing counting process intensities by means of kernel functions.
\newblock {\em The Annals of Statistics}, pages 453--466.

\bibitem[Rudin, 1976]{rudin1976principles}
Rudin, W. (1976).
\newblock {\em Principles of {M}athematical {A}nalysis. {T}hird {E}dition}.
\newblock McGraw--Hill Book Co., New York.

\bibitem[Van~der Klaauw, 2008]{van2008regression}
Van~der Klaauw, W. (2008).
\newblock Regression--discontinuity analysis: a survey of recent developments
  in economics.
\newblock {\em Labour}, 22:219--245.

\bibitem[van~der Vaart, 1998]{vandervaart1998asymptotic}
van~der Vaart, A.~W. (1998).
\newblock {\em Asymptotic {S}tatistics}.
\newblock Cambridge University Press, Cambridge.

\bibitem[Williams, 1991]{williams1991probability}
Williams, D. (1991).
\newblock {\em Probability with {M}artingales}.
\newblock Cambridge University Press, Cambridge.

\bibitem[Yadlowsky et~al., 2018]{yadlowsky2018bounds}
Yadlowsky, S., Namkoong, H., Basu, S., Duchi, J., and Tian, L. (2018).
\newblock Bounds on the conditional and average treatment effect with
  unobserved confounding factors.
\newblock {\em arXiv preprint arXiv:1808.09521}.

\end{thebibliography}
\bibliographystyle{apalike}

\end{document}